%% file: main_arxiv.tex
\documentclass[showpacs,amsmath,amssymb,superscriptaddress,notitlepage,preprint,pra]{revtex4-1}
\usepackage{color}
\usepackage{physics}
\usepackage{graphicx}
\usepackage{dcolumn}%
\usepackage{ulem}
\usepackage[hidelinks]{hyperref}
\usepackage{subfigure}

\usepackage{amsmath, amsthm, amssymb, bm}

\usepackage{appendix}


\newtheorem{corollary}{Corollary}

\newtheorem{fact}{Fact}

\usepackage{makecell}

\usepackage{multirow,booktabs,tabularx,boxedminipage,adjustbox}
\usepackage[noend,noline,ruled,linesnumbered]{algorithm2e}

\newcommand{\HFRC}[1][Hefei National Research Center for Physical Sciences at the Microscale and School of Physical Sciences, University of Science and Technology of China, Hefei, China]{\affiliation{#1}}
\newcommand{\SHRC}[1][Shanghai Research Center for Quantum Science and CAS Center for Excellence in Quantum Information and Quantum Physics, University of Science and Technology of China, Shanghai, China]{\affiliation{#1}}
\newcommand{\HFNL}[1][Hefei National Laboratory, University of Science and Technology of China, Hefei, China]{\affiliation{#1}}
\newcommand{\HNKL}[1][Henan Key Laboratory of Quantum Information and Cryptography, Zhengzhou, China]{\affiliation{#1}}

\newcommand{\JIQT}[1][Jinan Institute of Quantum Technology and Hefei National Laboratory Jinan Branch, Jinan, China]{\affiliation{#1}}
\newcommand{\NIM}[1][National Institute of Metrology, Beijing, China]{\affiliation{#1}}
\newcommand{\QCTek}[1][QuantumCTek Co., Ltd., Hefei, China]{\affiliation{#1}}
\newcommand{\SMX}[1][School of Microelectronics, Xidian University, Xi'an, China]{\affiliation{#1}}
\newcommand{\PKUCFCS}[1][Center on Frontiers of Computing Studies, Peking University, Beijing, China]{\affiliation{#1}}
\newcommand{\PKUSCS}[1][School of Computer Science, Peking University, Beijing, China]{\affiliation{#1}}
\newcommand{\SHRI}[1][University of Science and Technology of China, Shanghai Research Institute, Shanghai, China]{\affiliation{#1}}

\begin{document}

\title{One- and two-dimensional cluster states for topological phase simulation and measurement-based quantum computation}

\author{Tao Jiang}
\thanks{These authors contributed equally to this work.}
\HFRC
\SHRC
\HFNL
\author{Jianbin Cai}
\thanks{These authors contributed equally to this work.}
\HFRC
\SHRC
\HFNL
\author{Junxiang Huang}
\thanks{These authors contributed equally to this work.}
\PKUCFCS
\PKUSCS
\author{Naibin Zhou}
\thanks{These authors contributed equally to this work.}
\HFRC
\SHRC
\HFNL
\author{Yukun Zhang}
\PKUCFCS
\PKUSCS
\author{Jiahao Bei}
\SHRC
\author{Guoqing Cai}
\SHRC
\author{Sirui Cao}
\HFRC
\SHRC
\HFNL
\author{Fusheng Chen}
\HFRC
\SHRC
\HFNL
\author{Jiang Chen}
\SHRC
% ...

\begin{abstract}
Quantum entanglement is a fundamental resource for quantum information processing and serves as a critical benchmark for quantum hardware performance. Cluster states are a special class of entangled states that serve as universal resources for measurement-based quantum computation and possess an intrinsic symmetry-protected topological order, which confers robustness against symmetry-respecting noise. Here we report the scalable preparation and verification of genuine multipartite cluster states on the 105-qubit Zuchongzhi 3.1 superconducting processor. We achieve one-dimensional cluster states of up to 95 qubits and two-dimensional cluster states of up to 72 qubits. The symmetry-protected topological cluster states exhibit input-state-dependent robustness under symmetry-breaking perturbations due to an operational parity structure that enhances the performance of measurement-based quantum computation. Furthermore, we use our two-dimensional cluster states to implement the Deutsch--Jozsa algorithm within the measurement-based quantum computation framework, achieving higher output-state fidelity compared with traditional circuit-based models and a query efficiency advantage over classical approaches. Our work establishes a scalable platform that combines large-scale entanglement generation, symmetry-protected topological order and practical quantum algorithms to enable robust, fault-tolerant measurement-based quantum computation.
\end{abstract}

\maketitle
\clearpage

\section{Introduction}

Entanglement is a crucial resource in quantum information science~\cite{Nielsen2010,Horodecki2009}, playing a foundational role in many-body physics~\cite{Amico2008}, quantum communication~\cite{mccutcheon2016experimental} and quantum computation~\cite{Jozsa1997, Graham2022,barz2012demonstration}. 
In particular, cluster states, characterized by their highly structured entanglement, constitute universal resources for measurement-based quantum computation (MBQC)~\cite{Raussendorf2001,raussendorf2003measurement,Briegel2009,Greganti2016}.
In MBQC, computation proceeds through local measurements on a preprepared entangled resource state, with the measurement pattern determining the computational logic. This approach not only offers an alternative to the circuit-based model but also supports fault-tolerant architectures compatible with surface codes~\cite{Raussendorf_2006,Raussendorf_2007,Larsen_2021}.
Moreover, cluster states naturally realize a symmetry-protected topological (SPT) phase~\cite{Azses2020, Paszko2024}, enabling inherent error robustness in MBQC~\cite{symmetry2012, nautrup2015symmetry,liao2022topological,brown2020universal}.

The experimental generation of large-scale cluster states~\cite{Cao2023,Ferreira2022,Paesani2023} and general entangled states~\cite{carrera2024combining,lanyon2013measurement,bluvstein2022quantum,bao2024creating,Song2019,CaoH2023,Pont2024} has become a key benchmark for quantum hardware and represents a pivotal step towards universal quantum computation~\cite{Ladd2010,Wang2018,Omran2019,Gong2019,Thomas2022}.
This is enabled from the growing integration scale of quantum processors, and substantial advances in quantum control and hardware engineering, including improved gate precision~\cite{yan2018tunable, klimov2024optimizing,sung2021realization}, enhanced read-out fidelity~\cite{bengtsson2024model,blais2021circuit,jeffrey2014fast,macklin2015near}, surface participation optimization~\cite{woods2019determining,murray2018analytical,wenner2011surface} and advanced error mitigation methods~\cite{kamakari2022digital,ferracin2024efficiently,maurya2023scaling,karamlou2024probing}.
Yet, many experimental implementations remain fundamentally constrained by design limitations and intrinsic noise, limiting both scalability and practical utility of the entangled states~\cite{Guo2023,Krastanov2021}.
Moreover, experimental demonstrations of SPT order~\cite{de2019observation,Azses2020} and MBQC have been confined to small-scale systems~\cite{Walther2005,Tame2007}, leaving open the question of whether these capabilities can be scaled and harnessed for large-scale, practical quantum computation.

In this work, we tackle these experimental challenges using \textit{Zuchongzhi 3.1} (Fig.~\ref{fig1}a), a superconducting quantum processor featuring 105 qubits with substantially improved coherence, read-out and gate fidelity, compared with its predecessor, \textit{Zuchongzhi 2.0} (ref.~\citenum{Cao2023}). 
We certify genuine multipartite cluster states involving up to 95 qubits under the conventional read-out error mitigation model~\cite{carrera2024combining,Cao2023,bao2024creating}, nearly doubling the previous record. This advancement in large-scale entanglement generation establishes a robust platform for exploring the SPT properties intrinsic to cluster states and realizing practical MBQC (Fig.~\ref{fig1}b,c). 
Building on these resources, we experimentally probe the SPT nature of cluster states under realistic noise via teleportation circuits on a 95-qubit one-dimensional (1D) system. The input-state-dependent robustness observed under controlled symmetry-breaking perturbations reveals an operational parity structure, supporting the presence of underlying symmetry-protected order. 
Moreover, utilizing two-dimensional (2D) cluster states, we implement the Deutsch--Jozsa (DJ) algorithm within the MBQC paradigm and find that for larger input sizes, the MBQC realization attains higher output-state fidelities than its circuit-based counterpart, supporting its feasibility and scalability for executing multiqubit quantum algorithms on near-term hardware.

\begin{figure}[htbp]
    \centering
    \includegraphics[width=\linewidth,height=0.48\textheight,keepaspectratio]{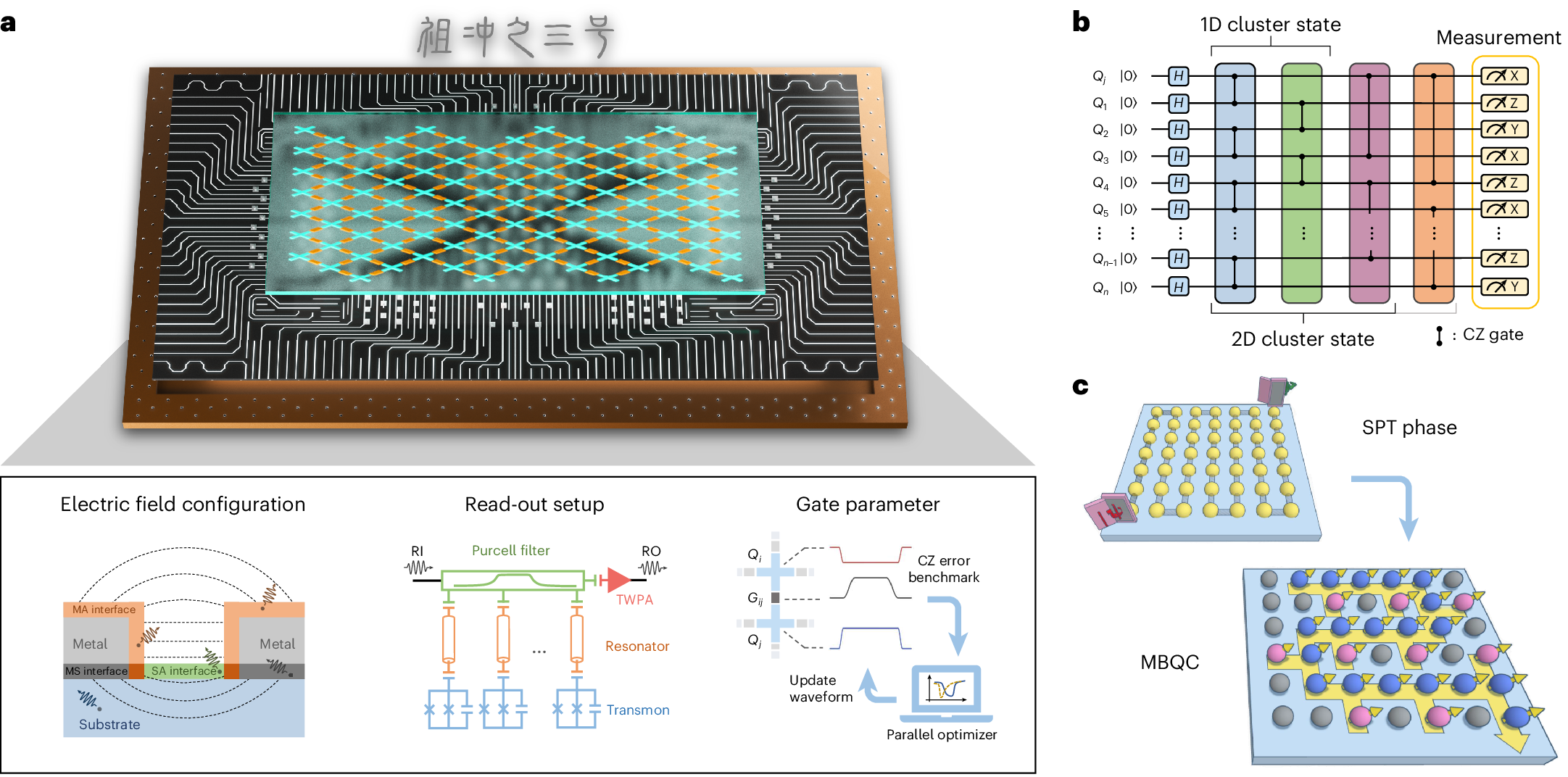}
    \caption{Chip architecture, key technologies, cluster state generation circuit and applications of the Zuchongzhi 3.1 processor. \textbf{a}, Schematic of the processor architecture and key technological innovations. The processor adopts a flip-chip design, with 105 transmon qubits and 182 tunable couplers on the top chip, and integrated read-in (RI), read-out (RO) and control circuitry on the bottom chip. To support large-scale cluster state generation, we improved the electric field configuration at the metal-air (MA), metal-substrate (MS) and substrate-air (SA) interfaces for enhanced coherence, integrated travelling-wave parametric amplifiers (TWPAs) for high-fidelity read-out, and developed parallelized CZ gate optimization. \textbf{b}, Quantum circuit for preparing 1D and 2D cluster states. All qubits are initialized in the \(| +\rangle\) state, followed by sequences of parallel CZ gates. Local Pauli measurements enable fidelity-based genuine entanglement certification and the implementation of MBQC protocols. \textbf{c}, Applications of high-fidelity cluster states. 1D cluster states support quantum teleportation experiments to probe SPT phases, whereas 2D cluster states serve as universal resources for MBQC.}
    \label{fig1}
\end{figure}

\section{Cluster state preparation and witness}
\label{sec:cluster_state}

We begin by introducing the general framework of graph states, which underlies the 1D and 2D cluster states realized in our experiments. 
Consider an $n$-vertex graph $G(V, E)$ with vertices $V$ and edges $E$, a stabilizer operator $S_i$ for each vertex $i$ is $S_i := X_i \bigotimes_{j \in \mathcal{N}(i)} Z_j$,
where $\mathcal{N}(i)$ denotes the set of neighbours of vertex $i$, that is, $\mathcal{N}(i) = \{x \mid x \in V, (x, i) \in E\}$, and $X_i$, $Y_i$ and $Z_i$ are the Pauli operators acting on the $i$th qubit.
The graph state $|\mathrm{G}\rangle$ is the common eigenstate of all stabilizers: $S_i |\mathrm{G}\rangle = |\mathrm{G}\rangle, \quad \forall i = 1, \ldots, n$.
Equivalently, any graph state has an explicit form $|\mathrm{G}\rangle = \prod_{(i, j) \in E} \mathrm{CZ}_{(i, j)} |+\rangle^{\otimes n}$,
where $\mathrm{CZ}_{(i, j)}$ represents the controlled-$Z$ gate acting on qubits $i$ and $j$, and $|+\rangle = (|0\rangle + |1\rangle)/\sqrt{2}$.
This form illustrates how the graph state can be constructively prepared from the computational basis states using a series of $\mathrm{CZ}$ gates.
Cluster states, such as those used in our experiments, are specific instances of graph states defined on regular 1D or 2D lattice graphs.
For an experimentally prepared graph state $\hat{\rho}_\mathrm{G}$, the fidelity with the ideal state $|\mathrm{G}\rangle$ is defined as $F = \langle\mathrm{G}|\hat{\rho}_\mathrm{G}|\mathrm{G}\rangle$,
which certifies genuine multipartite entanglement~(not a mixture of states that is biseparable under any bipartitions) with $F > 0.5$ (ref.~\citenum{Guhne2009}).
Measuring the fidelity requires implementing the projector $|\mathrm{G}\rangle \langle \mathrm{G}| = \prod_{i \in V} ({\mathbb{I} + S_i})/{2}$,
which can be efficiently obtained by randomly measuring Pauli operators $P^{(m)} = \prod_{i=1}^n P_i^{(m)}$ with each $P_i^{(m)}$ randomly drawn from $\{\mathbb{I}, S_i\}$ (ref.~\citenum{Flammia2011}). 

In a previous study~\cite{Cao2023}, we generated a 51-qubit cluster state on the \textit{Zuchongzhi 2.0} quantum processor. To achieve high-fidelity and large-scale cluster states, we implemented several hardware improvements targeting coherence, gate control and read-out fidelity. 
We optimized the qubit layout through surface participation engineering and adopted tantalum~\cite{place2021new} to enhance coherence, increasing the median $T_1$ from 30.8~$\mu$s (ref.~\citenum{Wu2021}) to 49.7~$\mu$s (Fig.~\ref{fig2}a). 
The read-out chain was improved by redesigning the Purcell filters and integrating a travelling-wave parametric amplifier, improving the read-out fidelity from 95.09\% (ref.~\citenum{Cao2023}) to 99.35\% (Fig.~\ref{fig2}b). 
Parallel CZ gate performance was enhanced through refined frequency arrangement and gate calibration, increasing the median fidelity from 99.05\% (ref.~\citenum{Cao2023}) to 99.50\% across 150 calibrated gates (Fig.\ref{fig2}c).
Details of these hardware improvements are provided in Supplementary Section I.

As the processor scales up, correlated errors---particularly those arising during simultaneous two-qubit gates and read-out---have become increasingly prominent, posing challenges that extend beyond conventional coherence-time limitations. Although single-qubit gate crosstalk can be effectively mitigated with calibrated compensation pulses, correlated errors in two-qubit gates and read-out remain challenging. To overcome this, we developed the dynamic coupling-off technique. For CZ gates, dynamic coupling-off pulses are applied on idle couplers to dynamically suppress unwanted coupling induced by frequency shifts. During read-out, dynamic coupling-off waveforms synchronously turn off residual coupling paths, substantially reducing measurement-induced correlations. This technique provides a practical and scalable solution, effectively suppressing key sources of correlated noise in cluster state preparation.

The quantum circuit for cluster state preparation begins with a layer of $\sqrt{Y}$ gates to initialize all qubits to the $|+\rangle$ state, followed by layers of parallel CZ gates applied sequentially to entangle the qubits (Fig.~\ref{fig1}b). 
To mitigate measurement errors, we rely on read-out error mitigation methods that operate on the experimentally calibrated read-out statistics obtained from our benchmark dataset.
Specifically, we measured the outcomes of 21,500 randomized initial states on the system involving 95 qubits, with each state subjected to 3,000 single-shot measurements, providing a detailed characterization of the system read-out errors. The dataset enabled us to estimate two-qubit correlated noise and calculate the correlation coefficients for selected qubit pairs, identifying key sources of error correlations. 
As exemplified in Fig.~\ref{fig2}d, for the correlated read-out error $1_j1_k$ to $0_j0_k$, which is defined as the covariance of events that $1_j$ and $1_k$ flips, the median correlation is \( 2.09 \times 10^{-6} \), with a maximum of  \( 2.30 \times 10^{-3} \), suggesting that most correlated noise between qubit pairs is relatively weak. 
This observation supports modelling read-out flip errors as approximately independent, which allows both independent and correlated read-out errors to be mitigated by the tensor product (TP) and continuous-time Markov process (CTMP) schemes, respectively~\cite{Bravyi2021}.

\begin{figure}[htbp]
    \centering
    \includegraphics[width=\linewidth,height=0.48\textheight,keepaspectratio]{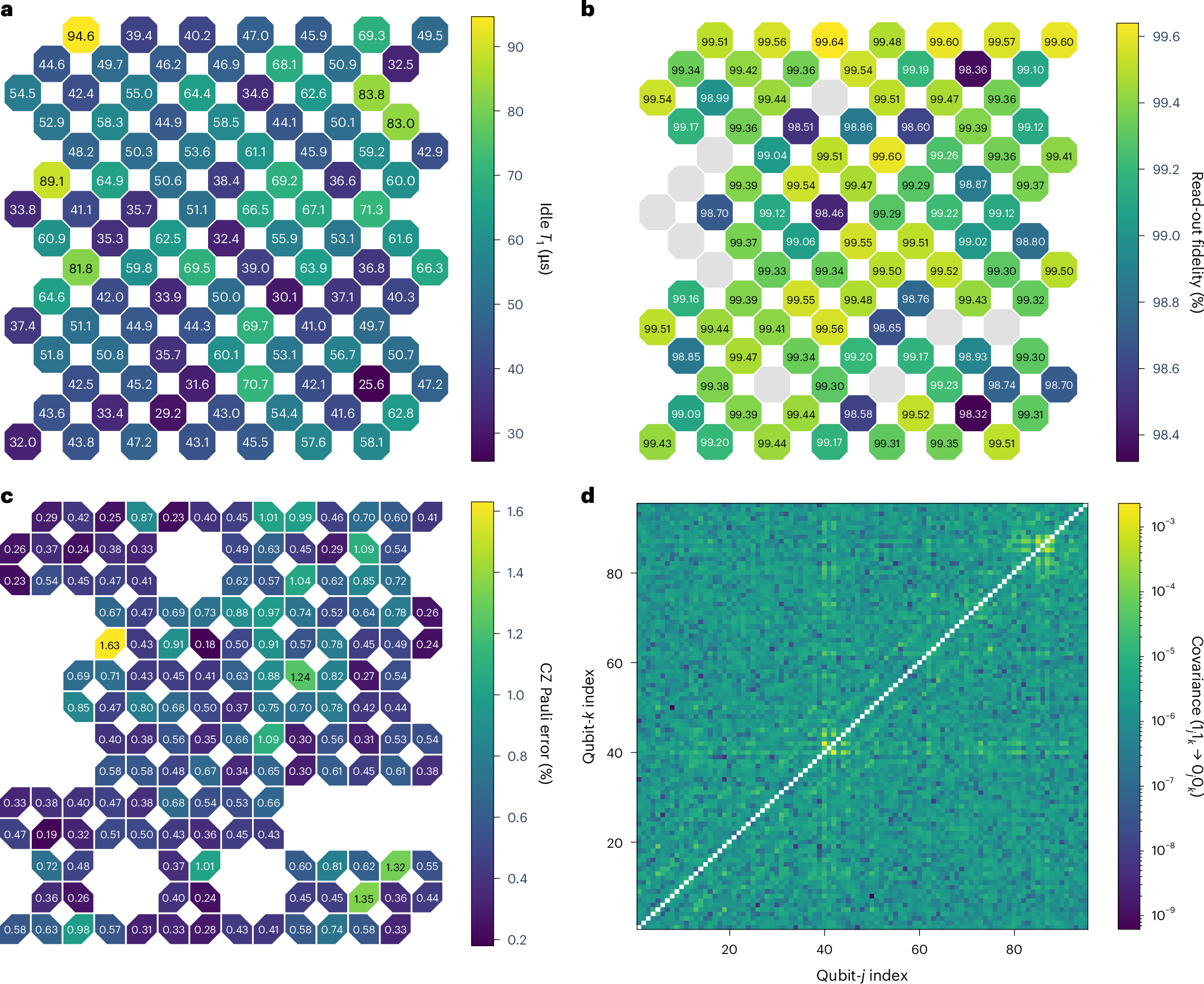}
    \caption{Key performances of the processor. \textbf{a}, Energy relaxation time $T_1$: the median $T_1$ of all 105 qubits at idle frequencies is 49.7~$\mu$s with a standard deviation of 13.6~$\mu$s. \textbf{b}, Read-out fidelity: the median read-out fidelity of 95 measured qubits using random state preparation schemes\cite{Cao2023} is 99.35\% with a standard deviation of 0.31\%. \textbf{c}, Two-qubit gate fidelity: the median fidelity of 150 calibrated CZ gates is 99.50\% with a standard deviation of 0.25\%, obtained via parallel-pattern cross-entropy benchmarking (XEB) and reported as Pauli fidelities. \textbf{d}, Covariance matrix of the read-out errors $1_j1_k$ to $0_j0_k$ for the 95 measured qubits in the 1D topology. Off-diagonal elements indicate correlated errors, with the colour intensity on a logarithmic scale reflecting the magnitude of the covariance.}
    \label{fig2}
\end{figure}

For 1D cluster states, we used the topology depicted in Fig.~\ref{fig3}a. To evaluate the fidelity of the 95-qubit 1D cluster state, we uniformly sampled \(M = 3{,}490\) stabilizers and performed \(K = 3{,}000\) single-shot measurements per stabilizer. 
The experiments achieve a data acquisition rate of $\sim$1,700 shots per second, with a total measurement duration of approximately 26 h. This includes entanglement generation and verification across 5- to 95-qubit systems, along with correlated read-out characterization. 
Under the conventional independent read-out error mitigation model, implemented through TP mitigation, the fidelity of the 95-qubit 1D cluster state is certified as \( 0.5603 \pm 0.0084 \) with a confidence level exceeding 99.7\% (Supplementary Section~IV).

After analysing the fidelity results for the stabilizers (Supplementary Fig. 14a), we observed that the fidelity distribution follows a composite Gaussian-type form with a standard deviation of 0.064. This results in an overall average uncertainty of 0.0032 within a 99.7\% confidence interval. These findings confirm the stability of our system, as the uncertainty, accounting for both systematic and statistical errors, remains well below the theoretically estimated error. 

To further evaluate the effect of read-out correlations, we utilize the CTMP scheme~\cite{Bravyi2021} to address correlated errors between qubits. 
Unlike the independent read-out error model, the CTMP framework incorporates noise generators $G_i$, which correspond to specific error types, including single-qubit bit-flip errors and two-qubit correlated bit-flip errors. 
The evolution of noise is governed by the matrix exponential $\Lambda = \exp(G)$, where $G=\sum_i r_iG_i$ is the sum of these generators, weighted by error rate coefficients $r_i$. 
For noise inversion, we use a sampling-based method that avoids the computational complexity of direct matrix inversion. Instead, we sample from a Poisson distribution to compute the weighted sum, simulating the inverse effects of noise interactions. This approach enables efficient noise simulation, prioritizing dominant correlated errors, and scales well for larger quantum systems.
After applying the conventional correlated read-out error mitigation, we obtain a fidelity of $0.5639 \pm 0.0453$ for the 95-qubit case, with a confidence level exceeding 99.7\%. A more detailed analysis of correlated read-out errors and the associated error mitigation overhead is provided in Supplementary Section IV-A,B. Figure~\ref{fig3}d compares the fidelities obtained from both TP and CTMP methods, from 5 to 95 qubits, illustrating the negligible effect of read-out correlations in our system. 

For 2D cluster states, we implemented two distinct topologies: the three-pattern (`sparse') topology and the four-pattern (`full') topology (Fig.~\ref{fig3}b,c, respectively). The key difference between these configurations is that the three-pattern topology can be derived from the four-pattern topology by removing the fourth layer of CZ gates, which allows for a larger 2D entanglement system and achieves comparable fidelity in practice. 

For the 72-qubit 2D three-pattern cluster state, we uniformly sampled \(M = 2{,}393\) stabilizers, with each stabilizer undergoing \(K = 2{,}000\) single-shot measurements. For the 57-qubit 2D four-pattern cluster state, we sampled \(M = 2{,}002\) stabilizers, with the same number of single-shot measurements per stabilizer. 
We estimated the fidelities and associated errors for both topologies across a range of system sizes, from 6 qubits to their maxima of 72 qubits (three patterns) and 57 qubits (four patterns) (Fig.~\ref{fig3}e,f). The total measurement times were 16 h and 6.5 h, respectively. 
In particular, under the conventional independent error mitigation, the fidelities for the 72-qubit 2D three-pattern and 57-qubit 2D four-pattern cluster states are \(0.5519 \pm 0.0054\) and \(0.5104 \pm 0.0045\), respectively, with a confidence level exceeding 99.7\%.
By contrast, under the conventional correlated read-out error mitigation, the fidelities for the 72-qubit and 57-qubit cluster states are \(0.5549 \pm 0.0127\) and \(0.5128 \pm 0.0082\), respectively. The comparison of the fidelities under both error mitigation schemes shows the negligible effect of read-out correlation on the fidelity results for these systems. 

\begin{figure}[htbp]
    \centering
    \includegraphics[width=\linewidth,height=0.48\textheight,keepaspectratio]{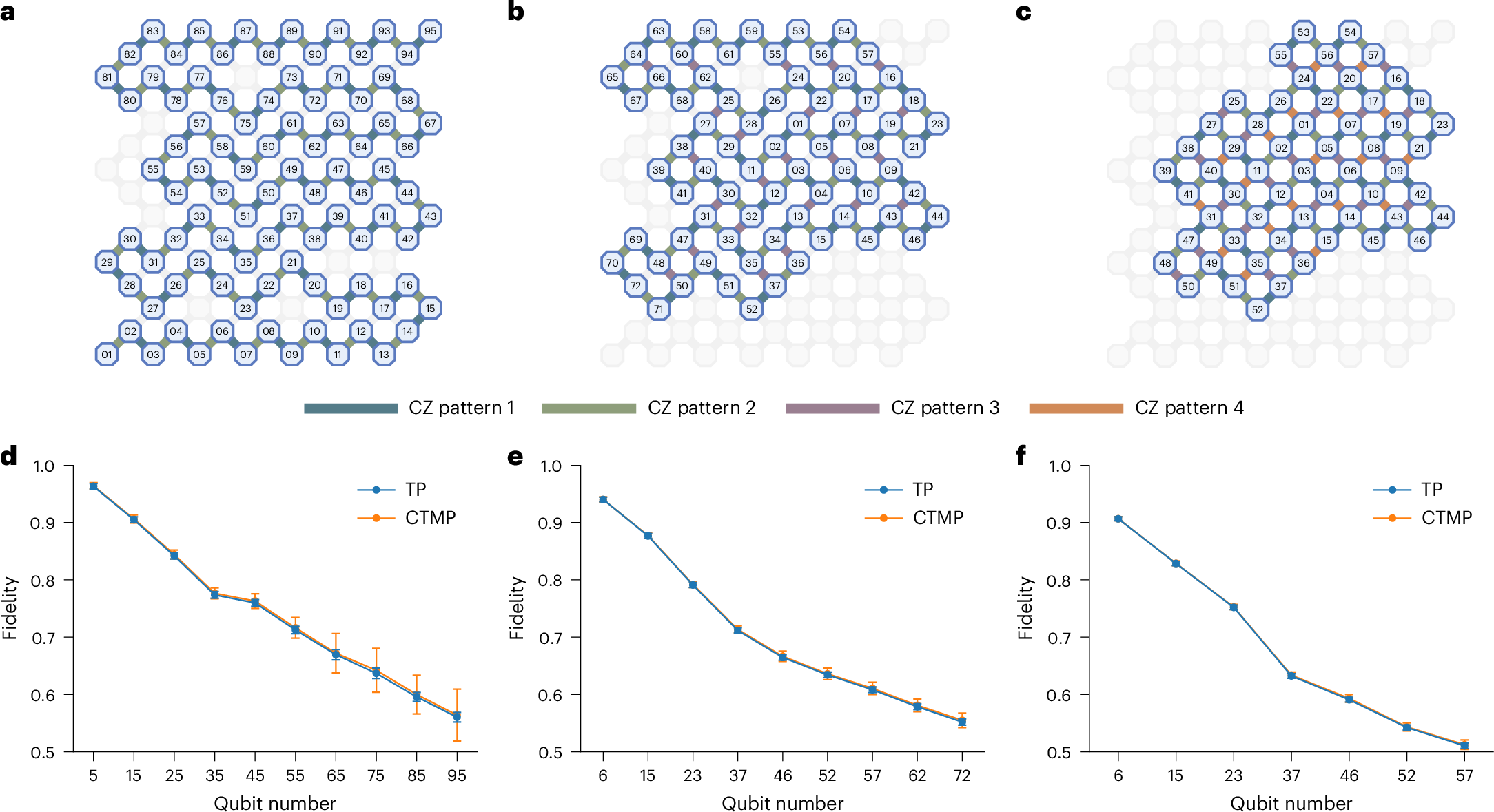}
    \caption{Layouts and fidelities of 1D and 2D cluster states. \textbf{a}--\textbf{c}, Layouts of cluster state geometries: 1D cluster state with 95 qubits and two CZ patterns (\textbf{a}), 2D sparse cluster state with 72 qubits and three CZ patterns (\textbf{b}) and 2D full cluster state with 57 qubits and four CZ patterns (\textbf{c}). \textbf{d}--\textbf{f}, Cluster state fidelity versus qubit number under two read-out error mitigation strategies (TP and CTMP), shown for 1D (\textbf{d}), 2D sparse (\textbf{e}) and 2D full (\textbf{f}) cluster states. The error bars denote the 99.7\% confidence upper bounds derived from the Hoeffding inequality. The sampling configuration, including the number of randomly sampled stabilizer circuits and the number of repetitions per circuit, is reported in Supplementary Tables IV--VI.}
    \label{fig3}
\end{figure}

\section{Simulation of SPT phases}
\label{sec:SPT}

Cluster states on a bipartite lattice, namely, two-colourable structures such as the line and square lattices used in our experiments, realize a non-trivial SPT phase under the on-site $\mathbb{Z}_2\times\mathbb{Z}_2$ symmetry. This SPT order can enhance the robustness of MBQC~\cite{miller2015resource,chirame2025stable,raussendorf2019computationally,stephen2017computational}, making the verification of $\mathbb{Z}_2\times\mathbb{Z}_2$ symmetry a meaningful step in assessing the computational utility of large-scale cluster states. With the availability of high-fidelity cluster states, we are now well positioned to verify and investigate their SPT phase through teleportation-based protocols~\cite{Azses2020,Paszko2024}.

A 1D cluster state satisfies the \( \mathbb{Z}_2 \times \mathbb{Z}_2 \) symmetry, which is characterized by the odd-parity operator \( P_{\text{odd}} = \prod_{i \in \text{odd}} S_i \) and the even-parity operator \( P_{\text{even}} = \prod_{i \in \text{even}} S_i \), and serves as the key resource for MBQC.
The SPT phase of the 1D cluster states can be certified by quantum teleportation~\cite{symmetry2012}, a special case of MBQC. Specifically, we teleport an input state $|\psi_{\rm in}\rangle$ using a 1D cluster state by entangling them and measuring the intermediate qubits in the Pauli $X$ basis. The input state is then transferred to the output qubit, up to a unitary correction applied in post-processing (Fig.~\ref{fig4}a).
The symmetry-protected nature of the 1D cluster states implies that quantum teleportation remains robust as long as the underlying symmetry is preserved. This robustness can be leveraged to detect SPT orders and, hence, benchmark the feasibility of MBQC. 

In our experiment, we first conducted a series of teleportation tasks with different input states and 1D cluster states with up to 94 nodes (Fig.~\ref{fig4}b). 
The fidelity \( F = |\langle \psi_{\text{in}} | \tilde{\psi}_{\text{out}} \rangle|^2 \), where \( |\tilde{\psi}_{\text{out}}\rangle = X^{\sum_{i \in \text{odd}} m_i} Z^{\sum_{i \in \text{even}} m_i+m_0} \ket{\psi_{\text{out}}} \), where \( Z \) and \( X \) are Pauli operators, and \( m_i \) denotes the measurement outcomes (0 or 1) for the qubits indexed starting from 1. These Pauli corrections are implemented in post-processing. 
The teleportation fidelity consistently exceeded 50\% for all cases, certifying the SPT phase of the 1D cluster states. In particular, the fidelities for the \( |+\rangle \), \( |-\rangle \), \( |0\rangle \) and \( |1\rangle \) states consistently outperformed those for the \( |+i\rangle \) and \( |-i\rangle \) states.
This indicates that the teleportation circuits could be affected differently by the noises for different input states.

We further investigated the SPT phases by introducing three symmetry-checking unitaries: $U_{\rm S}(\alpha, \beta) = e^{i \beta Z_1 X_2 Z_3} e^{i \alpha X_3}$, $U_{\rm SB, odd}(\alpha, \beta) = e^{i \beta Z_1 X_2 Z_3} e^{i \alpha Y_3}$ and  $U_{\rm SB, even}(\alpha, \beta) = e^{i \beta Z_1 X_2 Z_3} e^{i \alpha Y_2}$, corresponding to symmetry-preserving, odd-parity symmetry-breaking and even-parity symmetry-breaking perturbations, respectively. As shown in Fig.~\ref{fig4}c, for the input states $|0\rangle$ and $|1\rangle$, the fidelity of teleportation was examined as a function of phase $\alpha$, revealing the effects of these perturbations. Pronounced fidelity oscillations were observed under odd-parity symmetry-breaking perturbations, whereas the fidelity remained robust for the other unitaries. 

For the  $|+\rangle$  and  $|-\rangle$  input states (Fig.~\ref{fig4}d), the fidelity was influenced only by even-parity symmetry-breaking perturbations $U_{\rm SB, even}$. Interestingly, for the  $|+i\rangle$  and  $|-i\rangle$  states, fidelity oscillations were observed under both odd- and even-parity symmetry-breaking perturbations (Fig.~\ref{fig4}e). 

This intriguing phenomenon arises from the fact that successful teleportation requires the process to preserve certain $\mathbb{Z}_2$ symmetries associated with the input stabilizer states---odd parity (for example, $\ket{0}$, $\ket{1}$), even parity (for example, $\ket{\pm}$) or both (for example, $\ket{\pm i}$). A detailed explanation is provided in Supplementary Section V-D. 
In particular, the sensitivity to both odd- and even-parity symmetry-breaking perturbations aligns with the lower fidelities observed for the  $|+i\rangle$  and  $|-i\rangle$  states (Fig.~\ref{fig4}b), suggesting that such states are more susceptible to symmetry-breaking noise. 
This observed response reflects the presence of an operational parity structure. Although our experiment does not directly implement error correction, this structure highlights a potential inherent mechanism for parity-based fault tolerance in MBQC architectures. These findings offer valuable insights into how symmetry-protected features manifest in quantum computation and their potential utility in mitigating noise.

\begin{figure}[htbp]
    \centering
    \includegraphics[width=\linewidth,height=0.48\textheight,keepaspectratio]{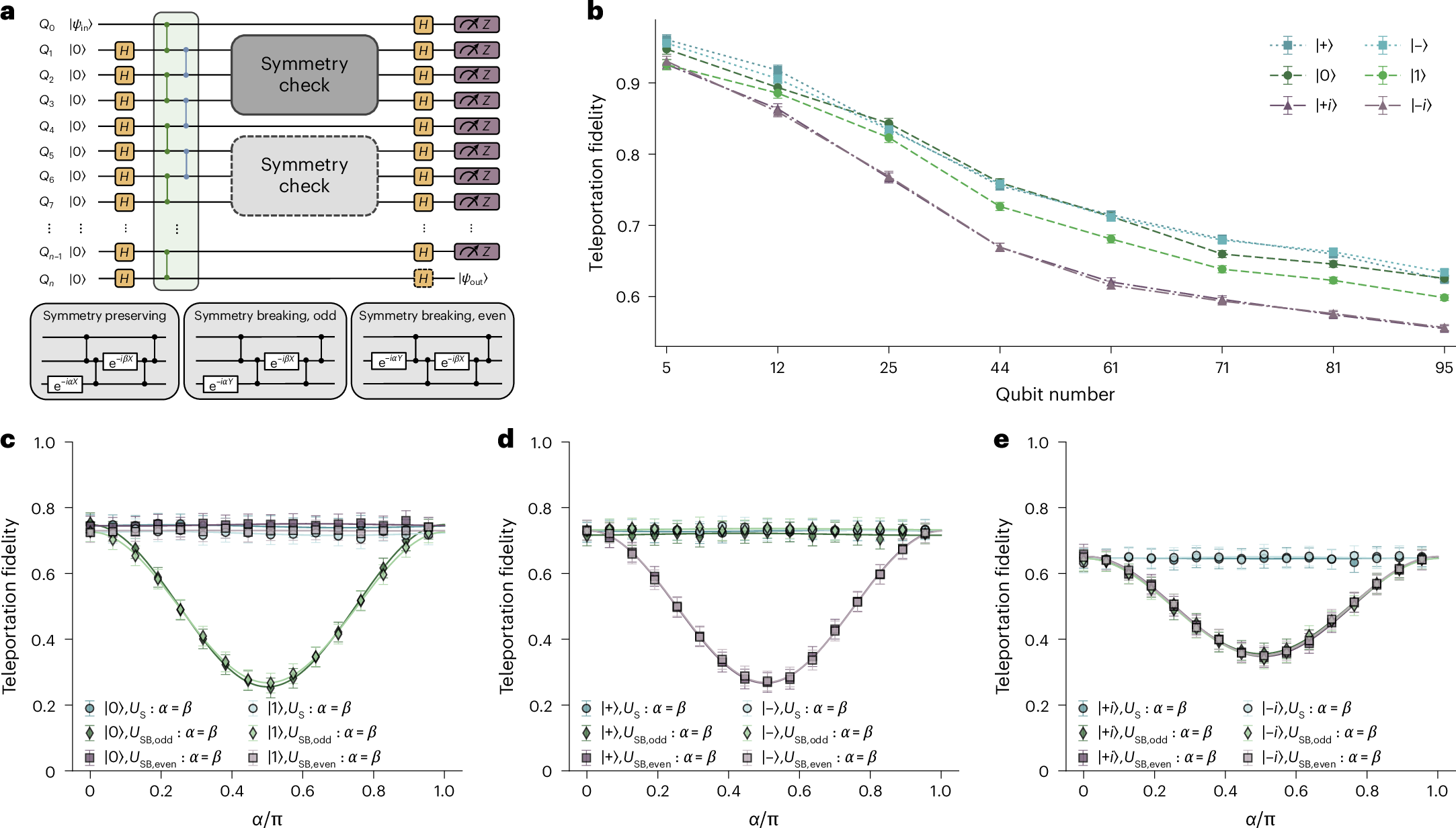}
    \caption{Experimental simulation of SPT phases in quantum teleportation. \textbf{a}, Quantum teleportation circuit for SPT phase experiments, where teleportation is realized through measurement after 1D entanglement. The experiments investigate fidelity oscillations in quantum teleportation under three symmetry perturbations. \textbf{b}, Teleportation fidelity as a function of the total qubit number for various input states (\(| 0\rangle\), \(| 1\rangle\), \(| +\rangle\), \(| -\rangle\), \(| +i\rangle\), \(| -i\rangle\)). The repeated single-shot counts (from 5 to 95 qubits) are 1.8 $\times$ $10^{5}$, 1.8 $\times$ $10^{5}$, 1.8 $\times$ $10^{5}$, 2.7 $\times$ $10^{5}$, 2.7 $\times$ $10^{5}$, 3.6 $\times$ $10^{5}$, 4.5 $\times$ $10^{5}$ and 5.4 $\times$ $10^{5}$. \textbf{c}--\textbf{e}, Teleportation fidelity versus the phase parameter $\alpha$ for the input states \(| 0\rangle\) and \(| 1\rangle\) (\textbf{c}), \(| +\rangle\) and \(| -\rangle\) (\textbf{d}) and \(| +i\rangle\) and \(| -i\rangle\) (\textbf{e}), measured using a 24-qubit cluster state under symmetry-preserving (circles), odd-parity symmetry-breaking (diamonds) and even-parity symmetry-breaking (squares) perturbations. Each data point represents the mean fidelity, and the error bars indicate the 99.7\% confidence upper bounds derived from the Hoeffding inequality. Each fidelity in \textbf{c}--\textbf{e} is obtained from 10,000 repeated single-shot measurements.}
    \label{fig4}
\end{figure}

\section{Implementing the DJ algorithm}
\label{sec:DJ}

Last, we implemented the DJ algorithm~\cite{Deutsch1992} using large-scale 2D cluster states within the MBQC framework~\cite{Raussendorf2001,raussendorf2003measurement,Briegel2009,Greganti2016}. In MBQC, preprepared entangled states act as computational resources, with algorithmic logic embedded in local measurement patterns. This model allows the circuit depth to be fixed before measurement, independent of the specific algorithm, and has the potential to surpass traditional circuit-based quantum computation for certain tasks~\cite{broadbent2009parallelizing,lee2022measurement,de2024quantum,kaldenbach2024mapping}. 

The DJ algorithm determines whether a black-box (oracle) Boolean function $f: \{0, 1\}^n \to \{0, 1\}$ is constant or balanced, with the promise that one of the two cases holds. In classical computation, distinguishing these cases in the worst scenario requires up to $2^{n-1} + 1$ queries. By contrast, the DJ algorithm leverages quantum superposition to make this determination with only a single query, offering exponential speed-up in theory and serving as a foundational module in various quantum algorithmic constructions.

Traditionally, the DJ algorithm has been demonstrated on circuit-based quantum devices for small $n$, constrained by limited qubit resources and circuit depth~\cite{Walther2005, Tame2007}. Here we demonstrate an MBQC realization by mapping the balanced function $f(x) = \bigoplus x_i$ (that is, the bitwise XOR of all input bits) onto a 2D cluster state (Fig.~\ref{fig5}a). Logical qubits---representing inputs and ancillary variables---are embedded into the cluster, and quantum operations are implemented by applying specific measurement patterns to the corresponding physical qubits in carefully chosen Pauli bases. This implementation involves up to 65 qubits (Fig.~\ref{fig5}b) and input size up to $n = 20$ bits. 

We benchmarked the MBQC implementation by comparing its output fidelity against both a circuit-based realization and the fidelity of the underlying cluster states (Fig.~\ref{fig5}c). 
For constant functions, the trivial oracle corresponds to an identity operation, resulting in consistently high output fidelities across all input sizes. In the balanced case, however, the circuit-based model exhibited a faster decline in fidelity as $n$ increased, primarily due to the growing circuit depth. By contrast, the MBQC approach maintained more stable output performance, achieving a fidelity of $0.5643 \pm 0.0095$ at $n = 20$, compared with $0.4645 \pm 0.0030$ for the circuit-based approach, despite the latter using an optimally selected subset of CZ gates with a lower median error rate of $3.3 \times 10^{-3}$. 

Moreover, the MBQC output fidelity consistently surpassed that of the corresponding cluster states, with the gap widening as the input size increased (Fig.~\ref{fig5}c). This trend illustrates a key feature of the MBQC framework: the logical performance is not strictly limited by the global fidelity of the resource state. Instead, it benefits from the symmetry-protected structure of the cluster state, which confers robustness against local imperfections and enables reliable algorithmic execution even when the entangled resource is moderately noisy.

To complement the DJ algorithm analysis, we further investigated the number of oracle queries required to distinguish constant and balanced functions. Adopting a hypothesis testing strategy (Supplementary Section VI-D) based on the probability of measuring the all-zeros output state, we estimated the minimum number of queries needed to achieve a 99.9\% confidence level (Fig.~\ref{fig5}d). Although quantum imperfections in real devices preclude a perfect one-query success rate, the MBQC implementation required fewer queries than a randomized classical algorithm across all tested input sizes. In addition, the MBQC and circuit-based quantum models exhibited similar query requirements at the chosen confidence level. These results consistently demonstrate the efficiency of MBQC in query times relative to classical approaches, further supporting the algorithmic completeness of our implementation. More broadly, the observed stability and scalability in output fidelity highlight the practical viability of MBQC as an alternative computational paradigm and underscore the utility of cluster state resources for executing multiqubit quantum algorithms within the entanglement--SPT--MBQC framework developed in this work.

\begin{figure}[htbp]
    \centering
    \includegraphics[width=\linewidth,height=0.48\textheight,keepaspectratio]{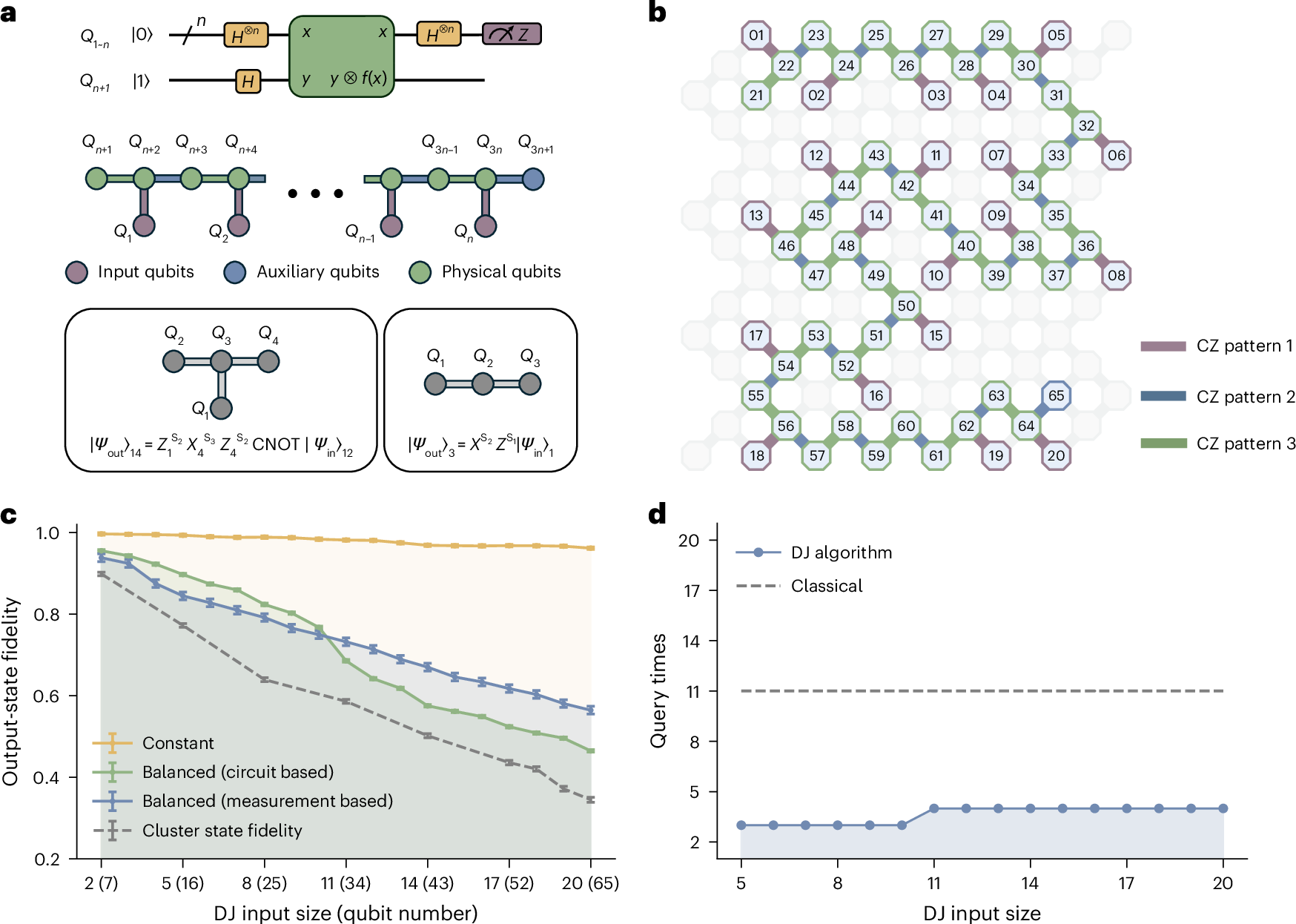}
    \caption{Demonstration of the DJ algorithm. \textbf{a}, Comparison between circuit-based and MBQC implementations. The circuit-based version applies Hadamard gates and an oracle \textit{f}(\textit{x}); for the representative balanced function $f(x) = \bigoplus x_i$, the oracle is realized by applying CNOT gates from each input qubit to the auxiliary qubit. The MBQC version maps the oracle to measurement patterns on a 2D cluster state; nodes of different types represent logical (input and auxiliary) or physical (measured MBQC) qubits, as indicated. The bottom boxes illustrate single- and two-qubit gate implementations in MBQC (Supplementary Section VI-B), with Pauli corrections applied via post-processing. \textbf{b}, Layout of the 2D cluster state used for the MBQC implementation of the balanced function. \textbf{c}, Output-state fidelity versus input size \textit{n} for constant and balanced functions, without applying read-out error mitigation. All data points represent the mean fidelities, with error bars denoting the 99.7\% confidence upper bounds derived from the Hoeffding inequality. The constant and balanced (circuit-based) data are obtained from $10^{6}$ samples, the balanced (measurement-based) data from $10^{5}$ samples and the sampling configuration for the cluster state fidelity is provided in Supplementary Table VIII. \textbf{d}, Query times required to distinguish constant from balanced functions, comparing the MBQC and classical approaches at 99.9\% confidence.}
    \label{fig5}
\end{figure}

\section{Conclusion}
\label{sec:conclusion}
In this work, we demonstrated the scalable generation and rigorous verification of genuine multipartite entangled cluster states involving up to 95 qubits on a superconducting quantum processor under the conventional read-out error model, establishing a new milestone for experimentally certified entanglement. Leveraging these high-fidelity cluster states, we experimentally characterized SPT phases in the 1D cluster state, revealing their inherent robustness and an operational parity structure through teleportation under symmetry-breaking perturbations. Furthermore, we implemented the DJ algorithm within the MBQC framework on a 2D cluster state, demonstrating the feasibility and scalability of MBQC on near-term quantum hardware. 

Leveraging large-scale entangled resource states and local single-qubit measurements, MBQC offers unique structural and operational advantages across a broad spectrum of computational tasks~\cite{broadbent2009parallelizing,lee2022measurement,de2024quantum,kaldenbach2024mapping}.
State-of-the-art superconducting platforms already integrate several capabilities essential for fault-tolerant MBQC, including multiround measurement, qubit reset and entanglement regeneration. These capabilities support fault-tolerant MBQC in 2D architectures via cyclic, layered operations that yield an effective three-dimensional resource state~\cite{Raussendorf_2007}.
Looking ahead, realizing effective three-dimensional architectures and generating large-scale 3D cluster states will mark pivotal milestones towards fault-tolerant MBQC on superconducting platforms~\cite{Raussendorf_2006,bourassa2021blueprint,brown2016fault,campbell2017roads}. Achieving these goals will require sustained advances in hardware design, alongside continuous improvements in device performance and control techniques, which are essential steps towards universal fault-tolerant MBQC systems.

\begin{acknowledgments}
We thank R. Jozsa and D. Azses for valuable discussions, the University of Science and Technology of China (USTC) Center for Micro- and Nanoscale Research and Fabrication for supporting the sample fabrication and QuantumCTek Co., Ltd. for the manufacture and maintenance of room-temperature electronics equipment. This research was supported by the Quantum Science and Technology-National Science and Technology Major Project (grant numbers 2021ZD0300200 and 2023ZD0300200), Anhui Initiative in Quantum Information Technologies, special funds from Jinan Science and Technology Bureau and Jinan High Tech Zone Management Committee, Shanghai Municipal Science and Technology Major Project (grant number 2019SHZDZX01), National Natural Science Foundation of China (grant numbers 92476203, 92476001, 12175003 and 12361161602), NSAF (grant number U2330201), Shandong Provincial Natural Science Foundation (grant number ZR2022LLZ008), Cultivation Project of Shanghai Research Center for Quantum Sciences (grant number LZPY2024), Key-Area Research and Development Program of Guangdong Province (2020B0303060001). M.G. was sponsored by National Natural Science Foundation of China (grant numbers T2322024 and 12474495), Shanghai Rising-Star Program (grant number 23QA1410000) and the Youth Innovation Promotion Association of CAS (grant number 2022460). X. Zhu acknowledges support from the New Cornerstone Science Foundation through the Xplorer Prize and the Taishan Scholars Program.
\end{acknowledgments}

\section*{Author contributions}
X. Yuan, M.G., X. Zhu and J.-W.P. conceived the research and designed the experiment. Development of the experimental system, device fabrication, experimentation, data analysis, and the writing and revision of the manuscript were carried out collaboratively by all authors.

\section*{Competing interests}
The authors declare no competing interests.

\section*{Data availability}
All source data supporting the findings of this study are available via figshare at \url{https://doi.org/10.6084/m9.figshare.30774326}. Source data are provided with this paper.
\section*{Code availability}
The code used in this study is available from the corresponding authors upon reasonable request.

\clearpage
\input{references.tex}

\clearpage
\section*{Supplementary Materials for ``One- and two-dimensional cluster states for topological phase simulation and measurement-based quantum computation''}
\input{SM_attached.tex}
\end{document}

%% file: references.tex

%% file: SM_attached.tex
\begingroup
\setcounter{section}{0}
\setcounter{subsection}{0}
\setcounter{figure}{0}
\setcounter{table}{0}
\setcounter{equation}{0}
\renewcommand{\thesection}{\Roman{section}}
\renewcommand{\thesubsection}{\Alph{subsection}}
\renewcommand{\theequation}{\arabic{equation}}
\renewcommand{\thefigure}{S\arabic{figure}}
\renewcommand{\thetable}{\Roman{table}}
\renewcommand{\theHsection}{supp.\arabic{section}}
\renewcommand{\theHsubsection}{supp.\arabic{section}.\arabic{subsection}}
\renewcommand{\theHfigure}{supp.\arabic{figure}}
\renewcommand{\theHtable}{supp.\arabic{table}}
\renewcommand{\theHequation}{supp.\arabic{equation}}
\newcommand{\yk}[1]{\textcolor{red}{#1}}
\newcommand{\comments}[1]{{\color{comment} #1}}
\renewcommand{\tr}[1]{\text{Tr}\left( #1 \right)}
\newcommand{\Amat}{\mathcal{A}}
\newcommand{\Cmat}{\mathcal{C}}
\newcommand{\Acal}{\mathcal{A}}
\newcommand{\Ccal}{\mathcal{C}}
\newcommand{\Dcal}{\mathcal{D}}
\newcommand{\Kcal}{\mathcal{K}}
\newcommand{\Ecal}{\mathcal{E}}
\newcommand{\Mcal}{M}
\newcommand{\Xcal}{\mathcal{X}}
\newcommand{\Ucal}{\mathcal{U}}
\newcommand{\Qcal}{\mathcal{Q}}
\newcommand{\Rcal}{\mathcal{R}}
\newcommand{\Scal}{\mathcal{S}}
\newcommand{\Rbb}{\mathbb{R}}
\newcommand{\Cbb}{\mathbb{C}}
\newcommand{\Ebb}{\mathbb{E}}
\newcommand{\Fbb}{\mathbb{F}}
\newcommand{\Hcal}{\mathcal{H}}
\newcommand{\Ubb}{\mathbb{U}}
\newcommand{\Ord}[1]{\mathcal{O}\left( #1 \right)}
\renewcommand{\abs}[1]{\left| #1 \right|}
\newcommand{\cbra}[1]{\left\{ #1 \right\}}
\newcommand{\vabs}[1]{\left\| #1 \right\|}
\newcommand{\pbra}[1]{\left( #1 \right)}
\newcommand{\sbra}[1]{\left[ #1 \right]}
\newcommand{\abra}[1]{\left\langle #1 \right\rangle}
\newcommand{\poly}{\mathrm{poly}}
\newcommand{\NP}{\mathrm{NP}}
\section{Experimental system setup}

\subsection{Qubit and readout characterization of the \textit{Zuchongzhi 3.1} processor}

Based on the \textit{Zuchongzhi 2.0} quantum processor, we developed the \textit{Zuchongzhi 3.1} processor with larger scale and improved performance by optimizing qubit topology design, materials, fabrication processes, and so on. The processor consists of 105 frequency-tunable transmon qubits and 182 tunable couplers, constructed using capacitively coupled transmon structures. The couplers are designed with higher frequencies than both the qubits and readout resonators. The device layout and key parameters are summarized in Fig.~\ref{si:fig:105Q}, including the qubits' maximum \( f_{01} \) frequencies, idle \( T_1 \) times, readout frequencies, resonator linewidths \( \kappa_r \), qubit-resonator coupling strengths \( g \), average readout photon numbers, and readout efficiencies~\cite{sank2025system} . The 105 qubits are divided into 15 rows for readout, each comprising seven qubits sharing a single readout Purcell filter. 

\begin{figure}[htbp]
\renewcommand{\baselinestretch}{1}\selectfont
    \centering
    \includegraphics[width=\linewidth,height=0.48\textheight,keepaspectratio]{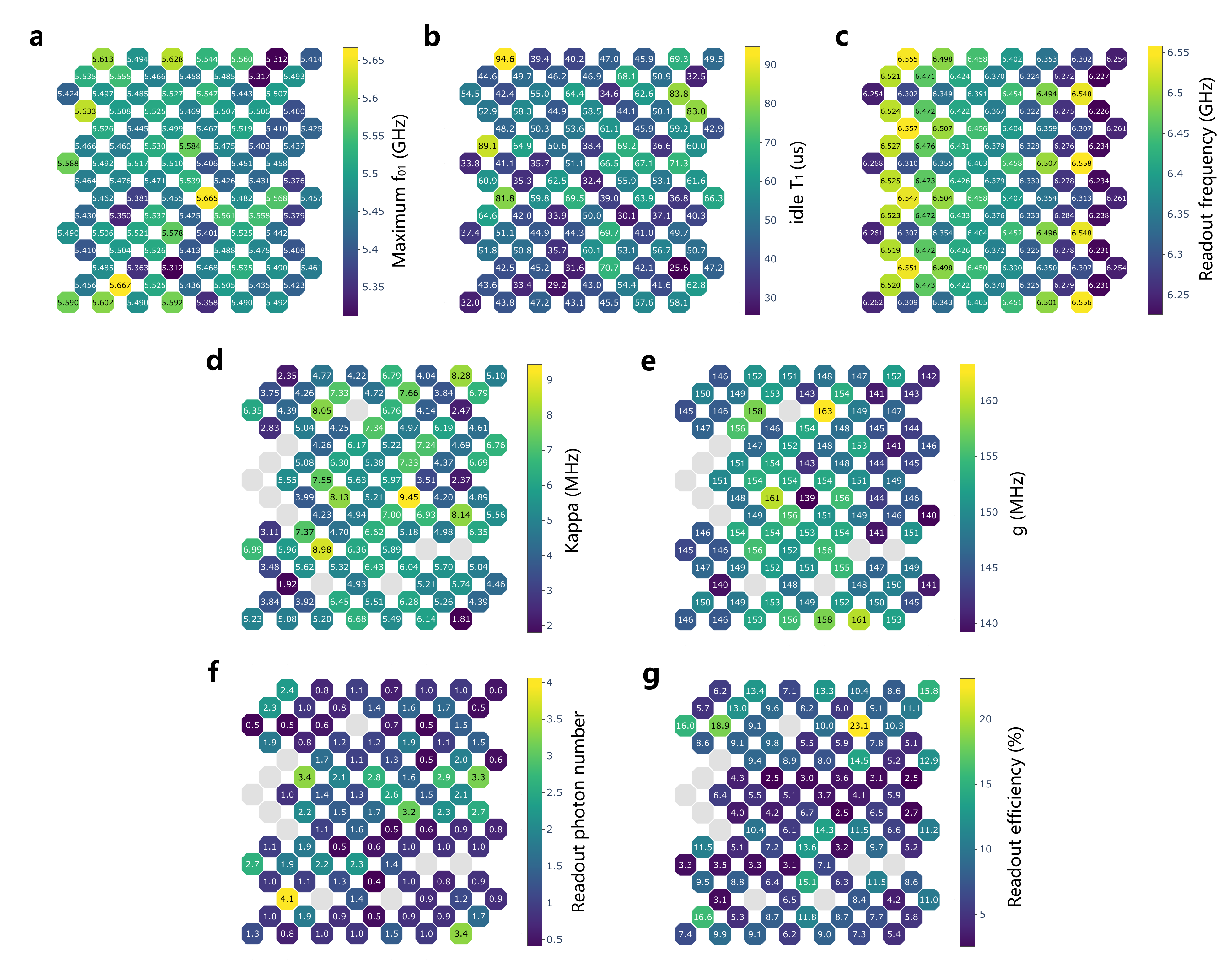}
    \caption{Characterization of qubit and readout parameters in the processor. (a-c) Parameters for all 105 qubits: (a) maximum \( f_{01} \) frequencies, (b) idle \( T_1 \) times, and (c) readout frequencies. (d-g) Readout-related parameters for the 95 qubits used: (d) resonator linewidths $\kappa_r$ (median $5.26 \mathrm{MHz}$, s.d. $1.52 \mathrm{MHz}$), (e) qubit-resonator coupling strengths $g$ (median $149.0 \mathrm{MHz}$, s.d. $5.0 \mathrm{MHz}$), (f) average photon number during the $2~\mu\mathrm{s}$ readout pulse (median $1.41$, s.d. $0.78$), and (g) readout efficiencies (median $8.1\%$, s.d. $3.9\%$).}
    \label{si:fig:105Q}
\end{figure}

\begin{figure}[htbp]
\renewcommand{\baselinestretch}{1}\selectfont
    \centering
    \includegraphics[width=\linewidth,height=0.48\textheight,keepaspectratio]{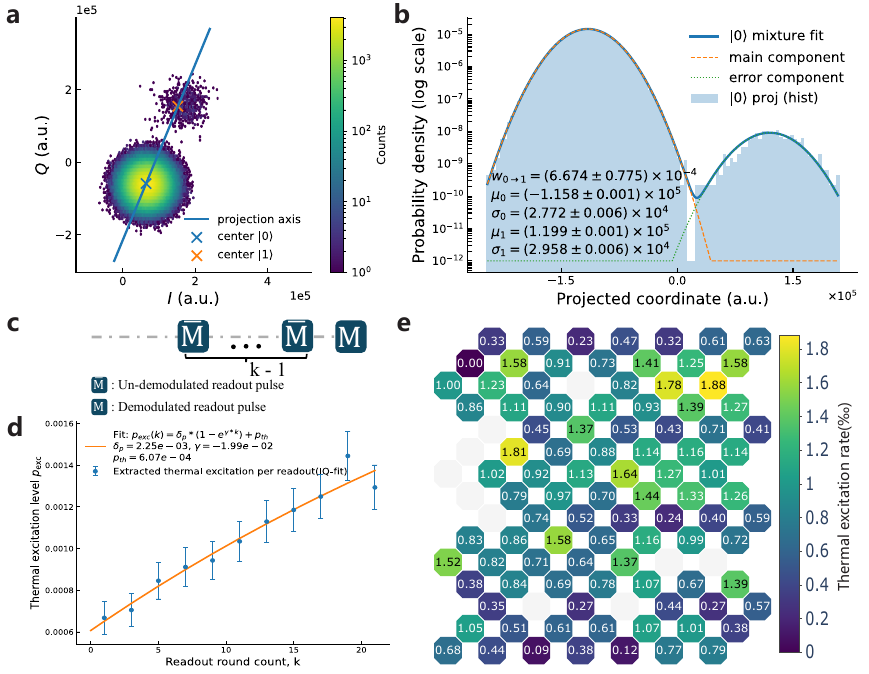}
    \caption{Extraction of on-chip thermal excitation levels using multi-round readout.
    (a) Two-dimensional IQ distribution for nominal \(|0\rangle\) preparation, obtained from one million measurement shots, together with the projection axis used for the analysis.
    (b) Log-weighted Gaussian mixture fit of the projected \(|0\rangle\) histogram, resolving the main and error components and yielding the excited-state weight \(w_{0\rightarrow 1}\). The uncertainties of the fitted parameters are obtained from the inverse Hessian of the likelihood function and reported as three times the standard deviations of the parameter estimates, consistent with panel~(d). 
    (c) Multi-round readout sequence composed of \(k-1\) un-demodulated pulses and a final demodulated pulse. 
    (d) Extracted excited-state population \(p_{exc}(k)\) as a function of the readout-round index \(k\), together with an exponential fit used to infer the intrinsic thermal excitation level \(p_{th}\). 
    (e) Spatial distribution of the extracted thermal excitation levels across the 95 qubits used in this work, with an average value of $0.85\text{\textperthousand}$ and a standard deviation of $0.42\text{\textperthousand}$.
    }
    \label{si:fig:thermal_excit}
\end{figure}

To quantify the intrinsic thermal excitation of the qubits prior to any measurement backaction, we employ a multi-round readout protocol combined with log-weighted Gaussian mixture fitting of the \(|0\rangle\)-state IQ distributions. As illustrated in Fig.~\ref{si:fig:thermal_excit}a, for each qubit we collect the IQ histogram corresponding to nominal \(|0\rangle\) preparation under a given readout-round index \(k\). The projected IQ coordinate is modeled by a two-component Gaussian mixture,
\[
(1-w_{0\rightarrow 1})\,\mathcal{N}(\mu_0,\sigma_0)
+ w_{0\rightarrow 1}\,\mathcal{N}(\mu_1,\sigma_1),
\]
where \(w_{0\rightarrow 1}\) represents the probability that the event originates from the excited-state component. Performing the fit on the logarithm of the binned probability density enables robust extraction of small excited-state fractions down to the \(10^{-4}\) level. An example of the log-scale mixture fit is shown in Fig.~\ref{si:fig:thermal_excit}b.

The experimental sequence used to probe measurement-induced heating is shown in Fig.~\ref{si:fig:thermal_excit}c. Each sequence consists of \(k-1\) un-demodulated readout pulses followed by a final demodulated readout pulse whose IQ distribution is recorded. The un-demodulated pulses induce the same backaction and residual photon population as a standard readout but do not contribute to the recorded IQ histogram, ensuring that only the final round is used for the mixture analysis. Repeating this procedure for different \(k\) yields a set of extracted excited-state probabilities \(p_{exc}(k)\), as summarized in Fig.~\ref{si:fig:thermal_excit}d.

Since repeated readout progressively increases qubit excitation through measurement-induced heating, the sequence \(p_{exc}(k)\) approaches a steady-state value. We therefore fit the extracted probabilities using
\[
p_{exc}(k) = \delta_p \left(1-e^{\gamma k}\right) + p_{th},
\]
where \(p_{th}\) denotes the intrinsic thermal excitation level in the absence of measurement, \(\delta_p\) captures the cumulative excitation induced by the readout pulses, and \(\gamma\) characterizes the per-round heating rate. Extrapolating the model to \(k=0\) yields the on-chip thermal population \(p_{th}\).

Applying this procedure to all 95 qubits used in this work yields the spatial distribution shown in Fig.~\ref{si:fig:thermal_excit}e. The extracted thermal excitation levels exhibit a mean of \(0.85\times10^{-3}\) and a standard deviation of \(0.42\times10^{-3}\). These values reflect the intrinsic qubit temperature of the system prior to any measurement and indicate that the processor operates at a low thermal excitation error level.

\subsection{Engineering advances in the processor}

We summarize the state and performance parameters of the processor, as detailed in Fig.~\ref{si:fig:perfrom_compare}. Compared to the \textit{Zuchongzhi 2.0} processor, the enhancements in scale and performance of the \textit{Zuchongzhi 3.1} processor are attributed to the following improvements:

\begin{figure}[htbp]
\renewcommand{\baselinestretch}{1}\selectfont
    \centering
    \includegraphics[width=\linewidth,height=0.48\textheight,keepaspectratio]{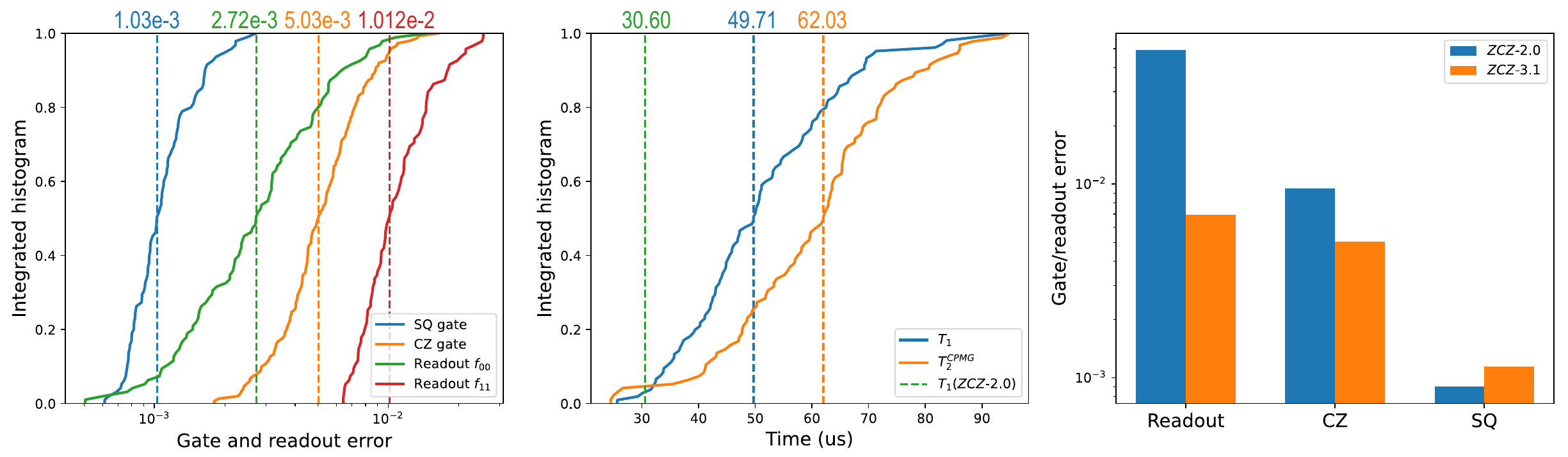}
    \caption{
    Performance metrics for \textit{Zuchongzhi 3.1} processor.
    (a) Integrated histograms of gate and measurement errors, showing error distributions for single-qubit(SQ) gate ($0.103\%$), CZ gate ($0.503\%$), and readout operations ($f_{00}$: $0.272\%$, $f_{11}$: $1.012\%$). SQ gates exhibit the lowest error rates, while CZ gates and measurement errors dominate overall error contributions.
    (b) Relaxation time $T_1$ ($49.71~\mu$s) and Carr-Purcell-Meiboom-Gill(CPMG) based coherence time $T_2^{\text{CPMG}}$ ($62.03~\mu$s) for superconducting qubits, alongside the median $T_1$ of the \textit{Zuchongzhi 2.0} device ($30.60~\mu$s). 
    (c) Comparison of gate and measurement errors between two system configurations (\textit{Zuchongzhi 2.0} and \textit{Zuchongzhi 3.1}). \textit{Zuchongzhi 3.1} (orange) demonstrates reduced errors for readout and CZ gate operations compared to \textit{Zuchongzhi 2.0} (blue), indicating improved system performance. SQ gate errors remain consistently low across both systems.
    }
    \label{si:fig:perfrom_compare}
\end{figure}

\textbf{Qubit design:}  
In the design of the quantum processor, we first optimized the geometric dimensions of the qubits and couplers, reducing the electric field participation ratio (EPR) of them. The simulation results show that the surface dielectric loss of this processor is 1.6 times lower than that of the \textit{Zuchongzhi 2.0} processor.  
While optimizing the geometric dimensions of the qubits, we also considered the requirements for scalable circuit design. In the current processor, the coupling between qubits consists of two components: direct capacitive coupling and indirect capacitive coupling. By tuning the frequency of the coupler between two qubits, the coupling strength can be quickly and precisely adjusted between approximately $+3\,\mathrm{MHz}$ and $-50\,\mathrm{MHz}$.  
For the readout design, a multi-stage Purcell filter was implemented, which significantly suppresses the Purcell effect and enhances readout multiplexing and performance. Additionally, we optimized the coupling strength between the processor and external components, as well as the sample box design, further improving the coherence and control performance of the processor.  

\textbf{Materials and fabrication process:}
The \textit{Zuchongzhi 3.1} processor employs flip-chip bonding technology. The top chip contains only the qubits and couplers, while all readout and control lines are placed on the bottom chip. To suppress signal crosstalk, all control lines are covered with full-encapsulation crossovers. The readout transmission lines and resonators use isolated crossovers to suppress slotline modes.
For the top chips, a high quality 200\,nm $\alpha$ -tantalum thin film was fabricated on a sapphire substrate using magnetron sputtering. Before deposition, the sapphire substrate underwent piranha solution cleaning, followed by high-temperature annealing and deposition.  
On the high-quality $\alpha$-tantalum thin film, the coplanar waveguide layer was fabricated using reactive ion etching (RIE). For the bottom chip, the crossover structure was prepared using a pattern transfer and deposition method to sequentially fabricate the silicon dioxide dielectric layer and aluminum bridges. For the top chip, the coupler and qubit Josephson junctions were fabricated using a double-angle evaporation technique.  
After indium bump deposition was completed on the top and bottom chips, the chips were diced, flip-chip bonded, and packaged.  

\begin{figure}[htbp]
\renewcommand{\baselinestretch}{1}\selectfont
    \centering
    \includegraphics[width=\linewidth,height=0.48\textheight,keepaspectratio]{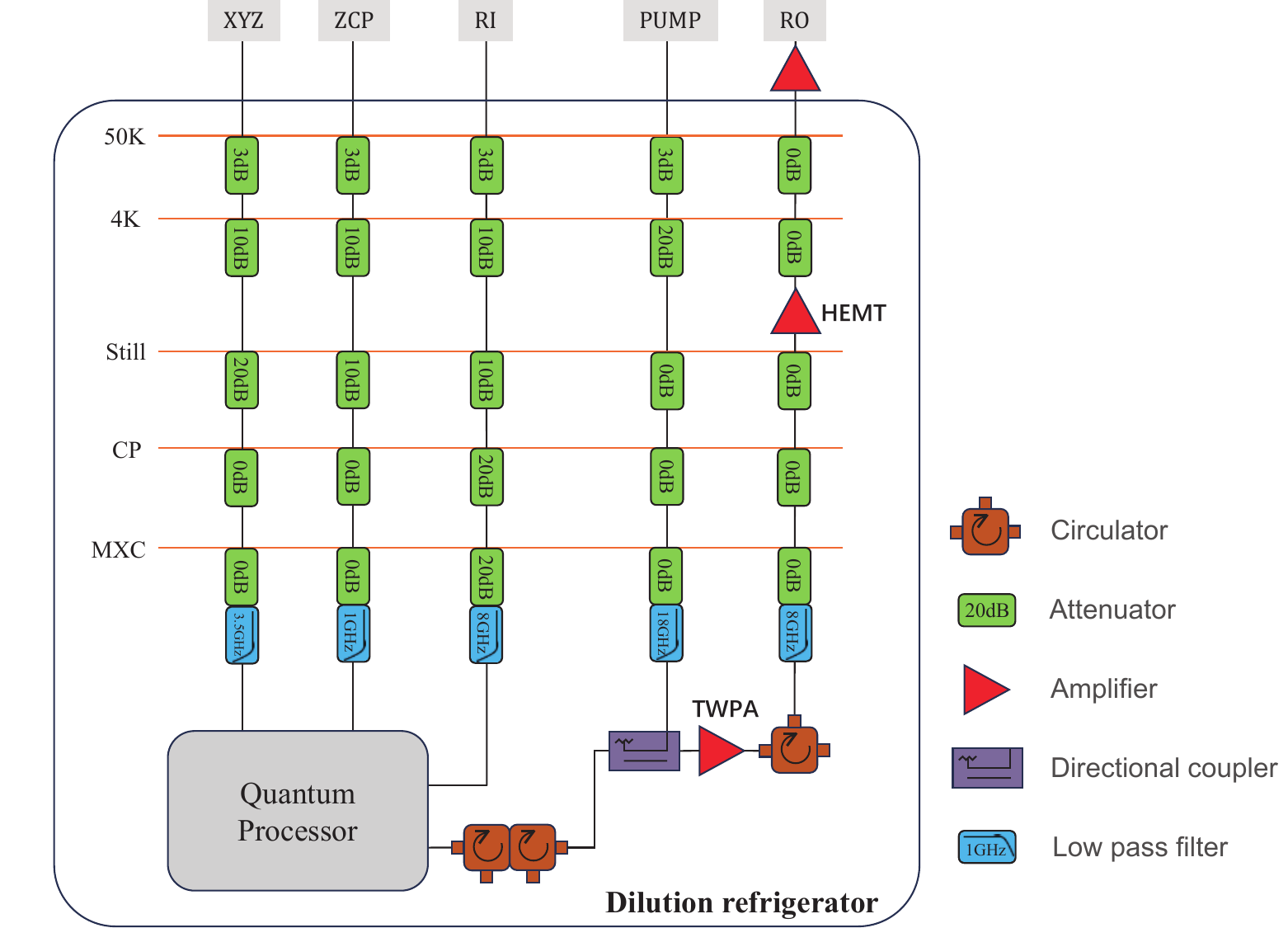}
    \caption{
        The schematic diagram of control electronics and wiring in refrigerator. Each qubit has corresponding XY and Z control lines, which are combined at room temperature stage using a bias tee before entering the dilution refrigerator. Similarly, the control lines for couplers are independently routed into the refrigerator due to different filtering requirements. In the dilution refrigerator, attenuators and filters are installed at various stages to reduce noise. Traveling wave parametric amplifiers (TWPA) , high electron mobility transistors (HEMT) and room-temperature amplifiers are used to amplify the readout signals. At room temperature, the integrated room-temperature electronics mentioned later provides XY/Z control signals as well as RI, RO, and Pump signals. Before being fed into the TWPA, the pump signal is combined with the chip's readout signal through a directional coupler. The readout signals amplified by the room-temperature amplifiers are digitized and demodulated by ADC modules. 
    }
    \label{si:fig:refrigerator}
\end{figure}

\textbf{Transmission system:}  
After chip fabrication, the processor is connected and secured to a PCB in a gold-plated copper sample box by wire bonding. The fully packaged processor is then mounted on the MXC stage of a dilution refrigerator and connected to room-temperature electronics via microwave transmission lines. The experimental wiring setup for qubit and coupler controls and frequency-multiplexed readouts is shown in  Fig.~\ref{si:fig:refrigerator}. 
To suppress thermal noise from room-temperature electronics, as well as signal attenuation and shielding, appropriate attenuators and low-pass filters are installed at different temperature stages of the dilution refrigerator.  
Compared to the signal system of the \textit{Zuchongzhi 2.0} processor, we optimized the integration and thermal load of the signal transmission lines in this system. Additionally, in the readout chain, we replaced the Josephson parametric amplifier(JPA) with an in-house developed traveling wave parametric amplifier(TWPA) as the first-stage amplifier for readout signals, further enhancing the signal-to-noise ratio and multiplexing capability of the readout system.  

\begin{figure}[htbp]
\renewcommand{\baselinestretch}{1}\selectfont
    \centering
    \includegraphics[width=\linewidth,height=0.48\textheight,keepaspectratio]{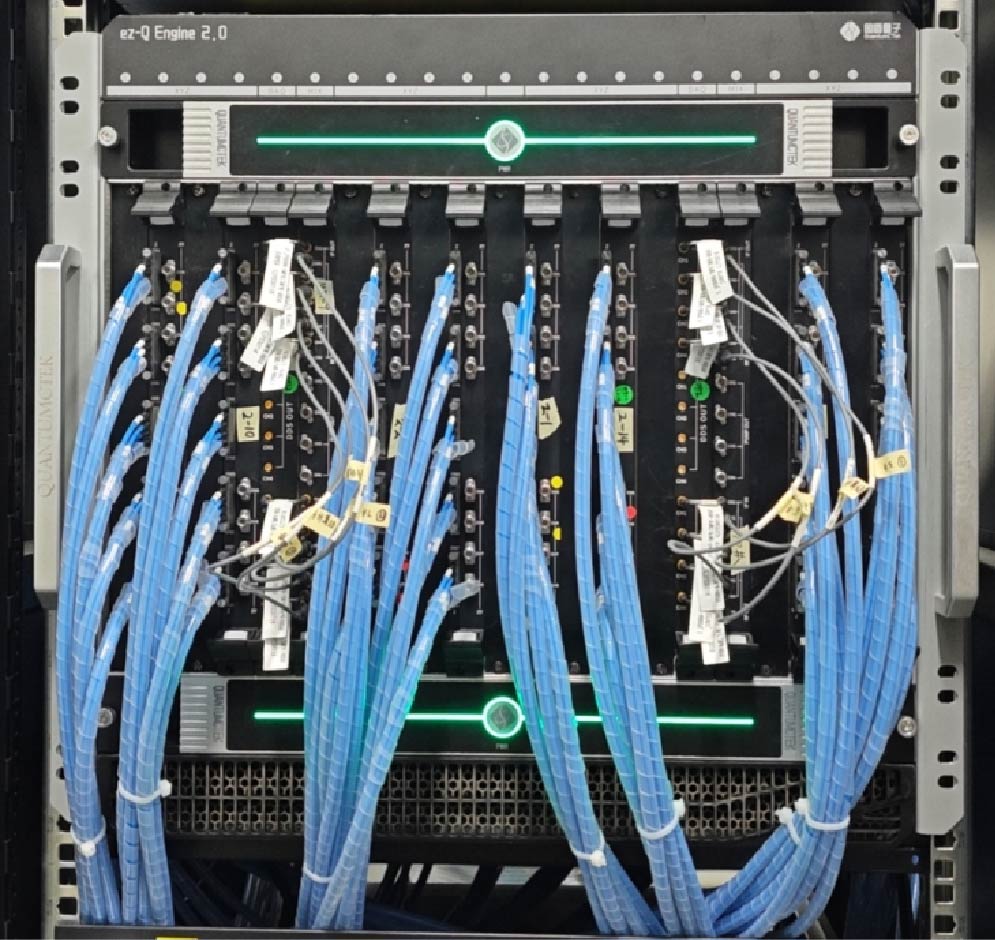}
    \caption{
        The room-temperature electronics devices used in our experiment. 
    }
    \label{si:fig:cabinet}
\end{figure}

\textbf{Room-temperature electronics System:}  
We designed and developed a high-integration room-temperature electronics system based on digital mixing and direct RF sampling, achieving an integration density of 10 qubits/Unit. The system uses a PXIE chassis as its foundational architecture and includes control modules, readout modules, communication modules, and support for multi-chassis synchronization and feedback cascading.  
As shown in the Fig.~\ref{si:fig:cabinet}, the single chassis accommodates 16 control cards, 4 readout cards, 1 communication card, 1 system synchronization card, and 1 clock trigger synchronization card in its front and rear slots. The control module generates both XY signals and Z signals. Z signals are produced using a conventional 1\,GSPS 4-channel DAC chip, while the XY signals are generated directly by an RF DAC chip. The 16 control cards offer 128 XY channels and 320 Z channels. Considering quantum chips with couplers, the control cards in a single chassis can manipulate at least 107 qubits. 
Besides, we have integrated digital quadrature modulator and NCO module in the RF DAC, so that issues such as signal leakage caused by analog IQ mixers are effectively avoided. Additionally, the high data rate of the high-sampling-rate chips utilizes the JESD 204B/204C standard, requiring only a minimal number of high-speed GT interfaces on the FPGA, which facilitates high integration.  
The readout module includes readout excitations, readout acquisitions, and readout pumps. Each of the 4 readout cards provides 8 RO, 8 RI, and 8 PUMP signals. The readout acquisition employs a single-sideband mixing scheme to reduce the usage of ADC/DAC channels, achieving a bandwidth of 1.8\,GHz, which enables real-time demodulation of 16 qubit states.  
The readout excitation uses the same RF DAC scheme as the control XY signals, generating 1.8\,GHz bandwidth signals that are upconverted through mixing to control the readout excitation of the qubits. The readout module also integrates PUMP signals that support 7-9\,GHz and 12-14\,GHz frequencies to meet the driving requirements of both TWPA and JPA.  

In addition, we made corresponding improvements in noise control and electromagnetic shielding. Through the optimizations and improvements mentioned above, the \textit{Zuchongzhi 3.1} processor achieved the integration of 105 qubits and 182 couplers, totaling 287 transmons. This processor also demonstrated significant improvements in two-qubit gate fidelity and measurement performance.  

Considering factors such as qubit measurement performance, coherence, and gate fidelity, we selected up to 95 qubits from the 105 operational qubits for our experiments based on the actual experimental requirements.  

\section{Processor calibration}

The processor calibration process in superconducting quantum computing typically involves frequency arrangement strategy, single-qubit gate calibration, two-qubit gate calibration, and readout calibration. These steps are crucial for ensuring qubits can reliably perform quantum gate operations with precise and accurate readout. To enable large-scale cluster state preparation, we have developed a more efficient processor calibration process. 

\subsection{Frequency arrangement strategy}

As the scale of superconducting quantum processors rapidly expands, determining how to effectively allocate the frequencies of qubits, couplers, and readout has become one of the key challenges in quantum processor measurement and control. The frequency arrangement strategy must take into account various factors, including direct and parasitic couplings, signal crosstalk, noise propagation, and so on. 

\begin{figure}[htbp]
\renewcommand{\baselinestretch}{1}\selectfont
    \centering
    \includegraphics[width=\linewidth,height=0.48\textheight,keepaspectratio]{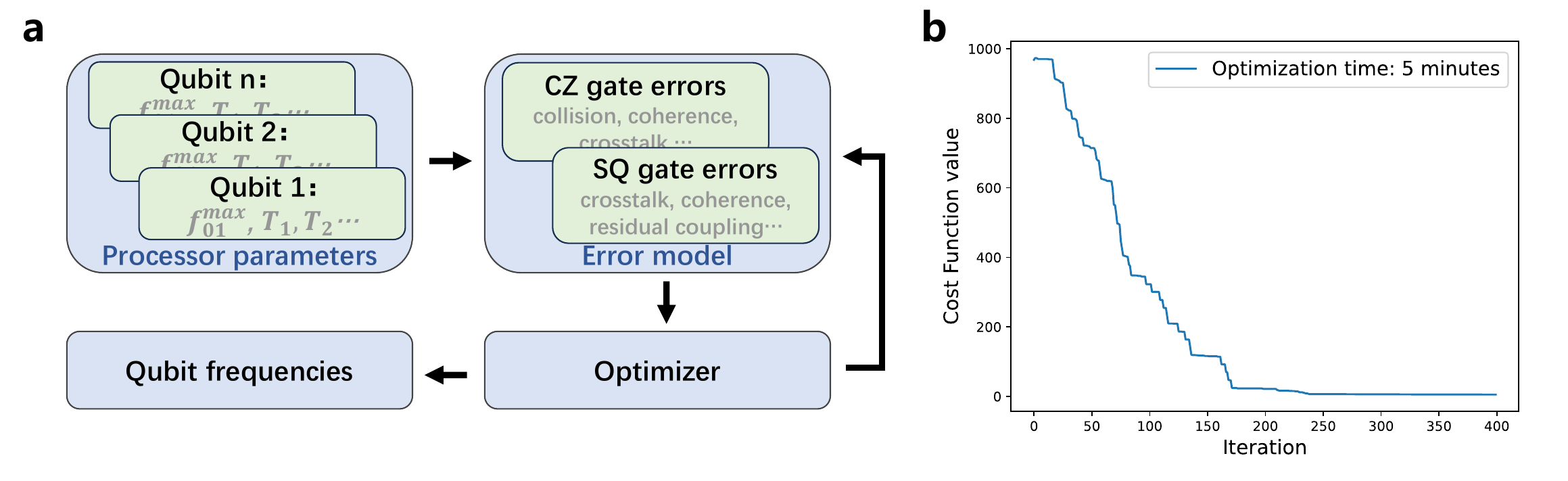}
    \caption{
     Frequency arrangement strategy for reducing errors in parallel gate operations. (a) Flowchart illustrating Frequency arrangement strategy. (b) Error rate convergence curve of the frequency arrangement optimizer, demonstrating convergence within five minutes.
    }
    \label{si:fig:Frequency_Arrangement}
\end{figure}

Based on the error model of quantum gates, we developed a feasible frequency arrangement strategy to reduce errors during parallel gate operations, shown in Fig.~\ref{si:fig:Frequency_Arrangement}a. In the quantum gate error model, the error model was instantiated using processor parameters and various error model functions. During the processor's initial calibration, we obtained state data from calibration experiments, including the qubit frequency-flux bias spectrum, qubit coherence performance and its frequency dependence, residual coupling parameters, distortion response parameters, XY signal crosstalk parameters, and other relevant data.

The error model functions were divided into two categories based on single-qubit gate Pauli errors and two-qubit gate Pauli errors. The single-qubit gate error model accounted for factors such as residual coupling during gate operation, qubit coherence performance, and XY signal crosstalk. The two-qubit gate error model function, specifically addressing CZ gate errors in this experiment, primarily considered frequency collisions along the qubit frequency trajectories during gate operation, coherence performance, and its frequency dependence, as well as crosstalk from signal distortion in the CZ gate. 

Additionally, to comply with the requirements of our room-temperature electronics and ensure signal phase synchronization, we introduced supplementary frequency allocation rules to constrain the frequency optimization range.

After instantiating the error model, we define the optimizer's objective function: 
\begin{equation}
Cost Function=Opt(Freqs, Patterns, Model)
\end{equation}
where $Freqs$ is a one-dimensional frequency list containing the qubit idle frequencies and the CZ gate interaction frequencies, $Patterns$ is a set of coupler groups, each used for the parallel execution of CZ gates in the quantum circuit, and $Model$ is the instantiation based on the processor parameters and error model functions. The optimizer employed a commercial, closed-source tool based on a tree search metaheuristic for discrete space optimization. 

To accelerate convergence and reduce the optimization space, we imposed tight constraints on the frequency arrangement, allowing each optimization cycle to converge within five minutes and yield a solution, shown in Fig.~\ref{si:fig:Frequency_Arrangement}b. However, the objective function for frequency arrangement is highly non-convex and often exhibits multiple local optima. After performing several optimization iterations, we selected the arrangement with the lowest error rate among the identified local optima as the final result. 

\subsection{Quantum gate optimization}

In this experiment, we updated the quantum gate optimization strategy to achieve high-fidelity single-qubit and CZ gate operations at larger scales. Addressing signal crosstalk is essential for ensuring high-fidelity parallel operations. To mitigate this, we effectively reduced control signal crosstalk on the processor using a full-encapsulation crossover process. However, XY signal crosstalk is observed for certain qubits, particularly those with multiple control lines crossing beneath them. 

\begin{figure}[htbp]
\renewcommand{\baselinestretch}{1}\selectfont
    \centering
    \includegraphics[width=\linewidth,height=0.48\textheight,keepaspectratio]{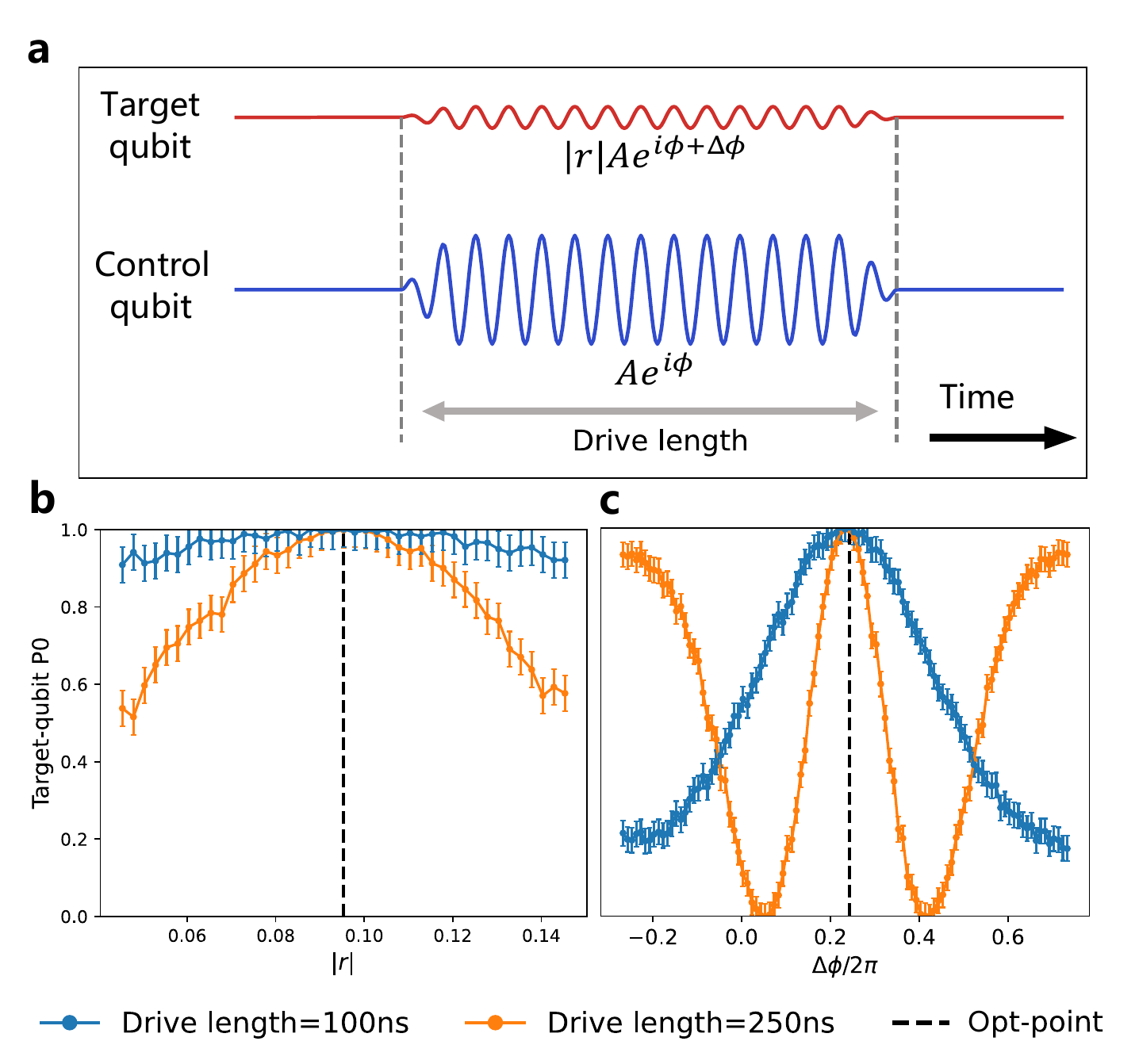}
    \caption{
         Illustration of the active microwave crosstalk correction technology and the XY crosstalk calibration results. 
         (a) Active microwave crosstalk correction applied to the affected qubits. 
         (b) and (c) Calibration results showing the impact of microwave amplitude and phase adjustments on the target qubit. Each data point in (b) is the mean of 500 repeated measurements and in (c) the mean of 1000 repeated measurements, with error bars given by $\Gamma/\sqrt{N}$, where $\Gamma$ accounts for the variance amplification introduced by the readout-mitigation operation $\Lambda^{-1}$. A detailed discussion of $\Gamma$ is provided in Supplementary Section~\ref{si:sec:readout_error_mitiga}. The blue line represents a drive length of $100 ns$, while the orange line corresponds to a drive length of $250 ns$. 
    }
    \label{si:fig:Single_Qubit_Gate_Cal}
\end{figure}

To improve the fidelity of parallel single-qubit gates, we implemented active microwave crosstalk correction technology, as illustrated in the Fig.~\ref{si:fig:Single_Qubit_Gate_Cal}a. Specifically, we apply a microwave signal with the same frequency and duration as the crosstalk signal to the affected qubits to mitigate its impact. The results from our XY crosstalk calibration experiment, also shown in the Fig.~\ref{si:fig:Single_Qubit_Gate_Cal}b and c, involve applying a drive signal of specific strength to the control qubit while systematically scanning the microwave amplitude and phase applied to the target qubit, enabling active compensation of the crosstalk waveform.

\begin{figure}[htbp]
\renewcommand{\baselinestretch}{1}\selectfont
    \centering
    \includegraphics[width=\linewidth,height=0.48\textheight,keepaspectratio]{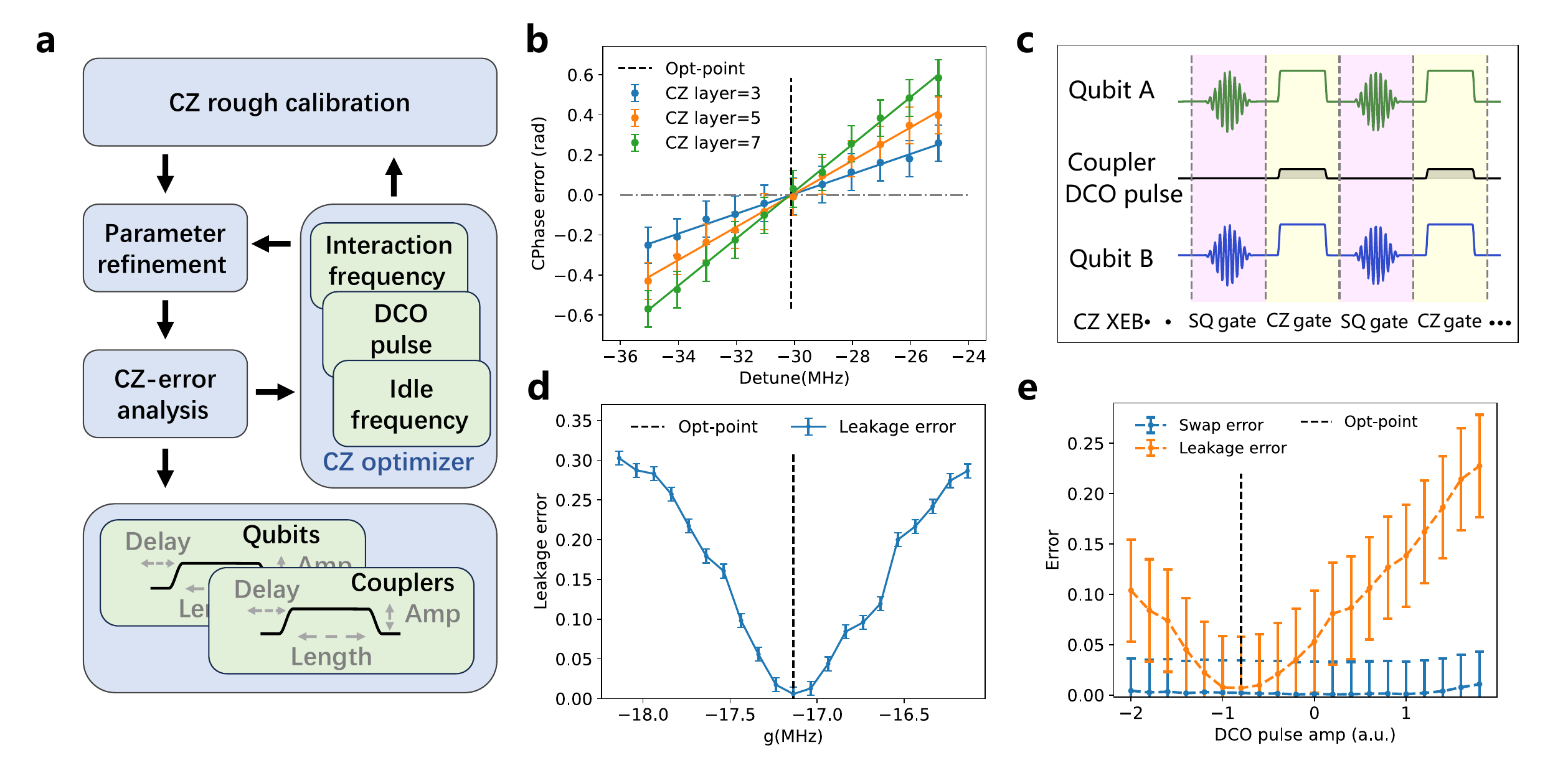}
    \caption{
    CZ gate optimization. 
    (a) Flowchart illustrating the CZ gate optimization strategy. 
    (b) Dependence of c-phase error on the detuning parameter of the CZ gate. 
    (c) Schematic of the DCO technique, demonstrating the application of DCO pulses to couplers to mitigate re-coupling. 
    (d) Dependence of leakage error on the coupling strength of the CZ gate coupler.
    (e) Swap and leakage errors during DCO optimization.
    Data in (b) and (e) are mean values from $N=1000$ repeated measurements, whereas data in (d) are mean values from $N=30000$ repeated measurements, with error bars representing the statistical uncertainty due to finite sampling. 
    In panel~(b), the uncertainty is given by $2\sqrt{2}\,\Gamma_q /\sqrt{N}$; 
    in panel~(d), by $\Gamma_{q_1}\Gamma_{q_2}/\sqrt{N}$; 
    and in panel~(e), by $\Gamma_{q_1}\Gamma_{q_2}/\sqrt{N}$ for the leakage error curve and $\Gamma_{q}/\sqrt{N}$ for the swap error curve.  
    Here $q_1$ and $q_2$ denote the two qubits involved in the CZ gate, and $\Gamma_{q_1}$ and $\Gamma_{q_2}$ are their respective variance amplification factors arising from readout mitigation.  
    In panels~(b) and (e), the factor $\Gamma_{q}$ corresponds to the qubit chosen for measurement.
    }
    \label{si:fig:CZ_Optimize}
\end{figure}

Prior to calibration, we had already determined a low-error-rate frequency arrangement based on the strategy outlined earlier, which included both idle and CZ interaction frequencies. Using this arrangement, we initially calibrated the CZ gate parameters and then performed distortion response calibration for the qubits and couplers during the CZ gate operation. To further ensure high-fidelity parallel CZ gate performance, we implemented an efficient calibration optimization procedure for the CZ parameters, as illustrated in the Fig.~\ref{si:fig:CZ_Optimize}a, with the detailed process described as follows:

\begin{enumerate}
    \item \textbf{Parameter refinement}: CZ gate errors were amplified using multi-layer CZ gate circuits, enabling precise optimization of qubit detuning frequencies and coupler coupling strengths. The optimization targeted cphase errors and leakage errors during CZ gate operations, shown in Fig.~\ref{si:fig:CZ_Optimize}b and c.
    \item \textbf{Interaction point optimization}: Signal crosstalk and frequency collisions in large-scale parallel CZ gate operations can significantly degrade the fidelity of certain gates. In this experiment, we not only targeted CZ SPB errors but also considered additional factors such as swap and leakage errors, which result from transitions of qubit states involving the $|1\rangle$ and $|2\rangle$ states of neighboring qubits. Two-qubit XEB~\cite{arute2019quantum} circuits were employed to optimize the CZ gate interaction frequencies, with a focus on minimizing these errors to improve overall gate fidelity.
    \item \textbf{Local frequency adjustment}: The variation in TLS locations on the processor can introduce significant decoherence errors in the quantum gates of certain qubits. To mitigate this, we optimized the frequencies of the affected qubits and CZ gate interaction points. These adjustments, guided by the frequency arrangement model, aimed to reduce local coherence errors. 
    \item \textbf{Dynamic Coupling Off technology}: During CZ gate operations, qubit detuning can cause unintended re-coupling with neighboring qubits. To mitigate this, we introduced Dynamic Coupling Off(DCO) technology, applying DCO pulses of the same duration as the CZ gate waveform to couplers outside the pattern around the target qubit. The pulse amplitude was optimized to minimize swap and leakage errors involving transitions with the $|1\rangle$ and $|2\rangle$ states of neighboring qubits, shown in Fig.~\ref{si:fig:CZ_Optimize}d and e, and the optimal coupling off point was selected based on these error metrics.
\end{enumerate}

By employing the above optimization techniques, we achieved parallel CZ gates with a fidelity up to 99.50\% and single-qubit gates with a fidelity exceeding 99.89\% in this experiment. 

\subsection{Readout calibration}

To achieve the preparation and witnessing of large-scale entangled states, it is crucial to simultaneously enhance parallel readout fidelities and minimize correlated measurement errors. In superconducting quantum systems, readout operations are typically the most error-prone. As the processor scale increases, the issue of readout correlation errors caused by qubit frequency collisions, leading to state transitions and leakage, becomes more prominent. This effect becomes especially significant as readout fidelity improves, further amplifying the impact of correlation errors.

In this experiment, we integrated the model of readout correlation errors with experimental data to develop an efficient strategy for optimizing parallel readout. This approach combines theoretical insights and experimental findings to enhance parallel readout fidelity while minimizing correlation errors. The optimization process is as follows:

\begin{figure}[htbp]
\renewcommand{\baselinestretch}{1}\selectfont
    \centering
    \includegraphics[width=\linewidth,height=0.48\textheight,keepaspectratio]{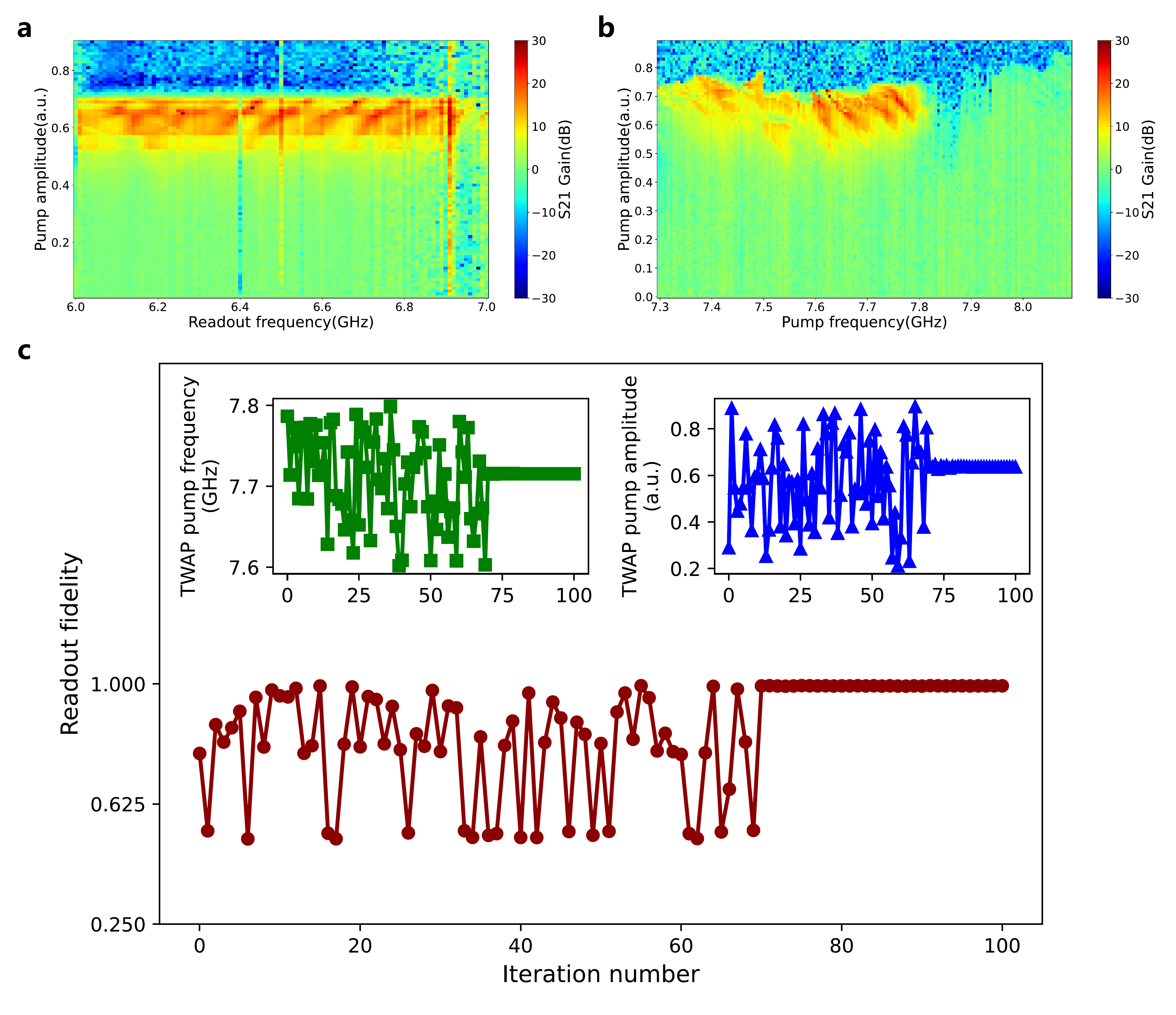}
    \caption{
    Performance and optimization process of the traveling wave parametric amplifier (TWPA). (a) The TWPA gain spectrum with a fixed pump frequency 7.7 GHz, where the x-axis represents the readout frequency and the y-axis represents the pump amplitude. (b) The TWPA gain spectrum with a fixed readout frequency 6.55GHz, where the x-axis represents the pump frequency and the y-axis represents the pump amplitude. (c) Experimental data showing the optimization of the TWPA pump frequency and amplitude using the NM algorithm, with readout fidelity as the optimization target. The x-axis represents the number of iterations. Each data point in (a) and (b) is the mean of 1000 repeated measurements, and in (c) is the mean of 8000 measurements. 
    }
    \label{si:fig:TWPA_opt}
\end{figure}

\begin{figure}[htbp]
\renewcommand{\baselinestretch}{1}\selectfont
    \centering
    \includegraphics[width=\linewidth,height=0.48\textheight,keepaspectratio]{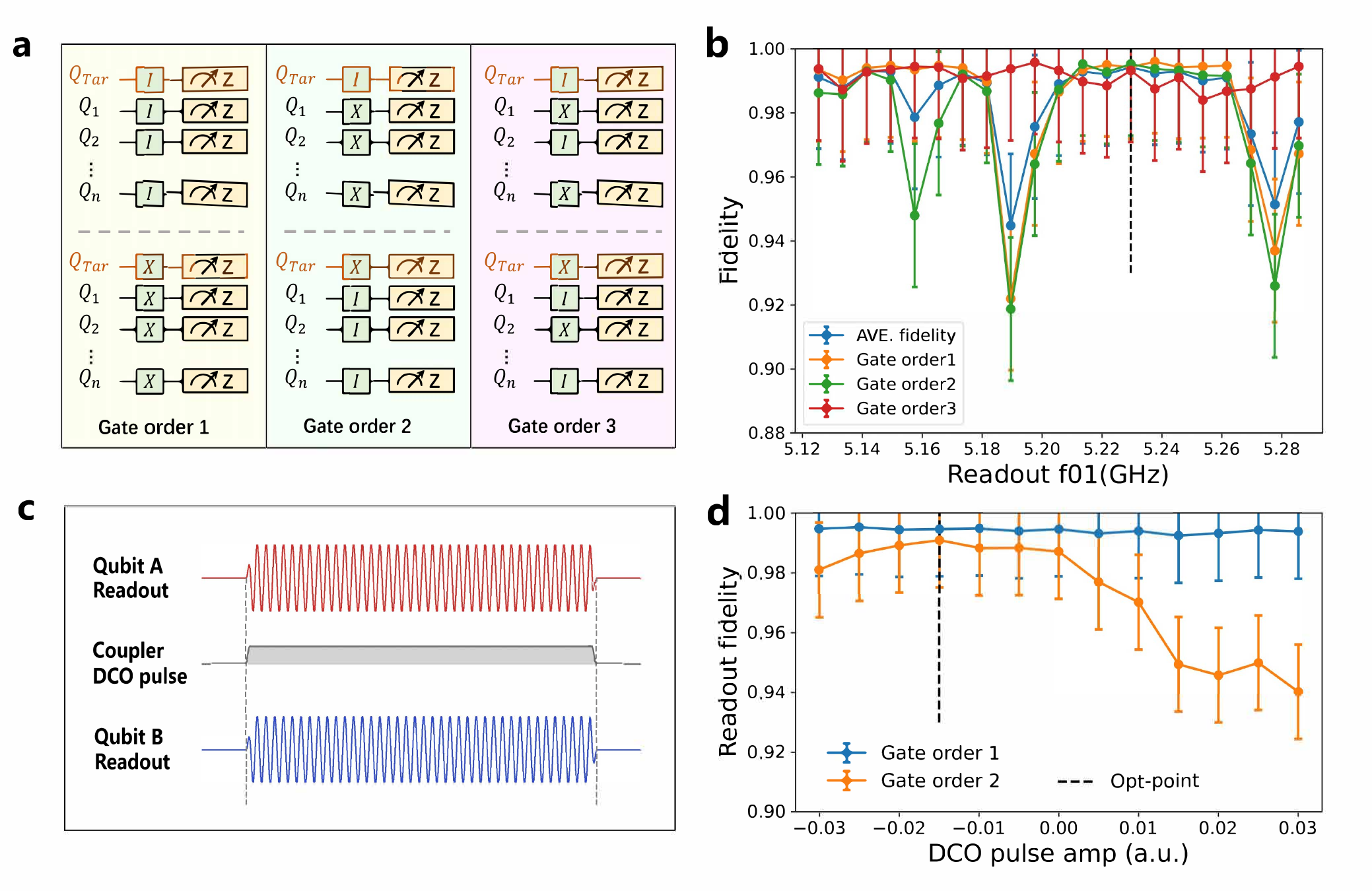}
    \caption{Readout optimization. 
    (a) Three representative circuits used for quick detection of readout correlation errors. 
    (b) Experimental data showing the optimization of correlated readout by scanning the qubit frequency during readout. 
    (c) Schematic of the DCO technique applied during readout.  
    (d) Experimental data showing the results of the DCO parameter scan, demonstrating the reduction in readout correlation errors.
    Each data point in (b) and (d) represents a mean value from 2000 and 4000 repeated measurements, respectively, with error bars showing the statistical uncertainty due to finite sampling ($1/\sqrt(N)$). 
    }
    \label{si:fig:Readout_Optimize}
\end{figure}

\begin{enumerate}
    \item \textbf{Initial readout parameter calibration}: Determine the initial readout parameters, including frequency, amplitude, and length, based on dispersive shift and other readout calibration experiments.
    \item \textbf{TWPA parameter optimization}: To determine the optimal operating point of the TWPA, the pump power and pump frequency are scanned to obtain the TWPA gain spectrum, as shown in Fig.~\ref{si:fig:TWPA_opt}a and b. Based on this spectrum, an optimal reference point is selected as the starting point for the optimization process. Subsequently, as shown in Fig.~\ref{si:fig:TWPA_opt}c, the NM algorithm is employed to maximize the signal-to-noise ratio (SNR) and fidelity of the readout, ultimately identifying the best operating point for the TWPA. 
    \item\textbf{Parallel readout optimization}: Readout parameters, such as power, length, and frequency, are optimized to maximize the fidelity of parallel qubit readout.
    \item \textbf{Crosstalk readout optimization}: In superconducting quantum processors using dispersive readout, the AC Stark effect can induce qubit frequency shifts during readout, potentially causing frequency collisions and significant correlation readout errors as the processor scale increases. In the experiment, three representative circuits were selected for quick detection and, shown in Fig.~\ref{si:fig:Readout_Optimize}a. By considering the weighted readout fidelity and differences among these circuits, parameters such as readout frequency, qubit frequency, and readout power were optimized, shown in Fig.~\ref{si:fig:Readout_Optimize}b.
    \item \textbf{Dynamic Coupling Off technology}: Qubit frequency shifts during readout can result in re-coupling, leading to readout correlation errors. To mitigate this, DCO technology is also applied here, shown in Fig.~\ref{si:fig:Readout_Optimize}c and d. By introducing a DCO waveform with the same duration as the readout waveform, coupling is dynamically turned off, thereby reducing readout correlation errors. 
\end{enumerate}

Through the optimization process described above, we reduced the readout errors to a level comparable to the two-qubit gate error rate, achieving a parallel 02 readout fidelity of up to 99.35\%. The improvement in readout fidelity significantly enhanced our experimental efficiency, enabling us to complete entanglement witness tasks with fewer random stabilizer circuits. In previous work, $4.3992\times 10^8$ single measurements were required to estimate the fidelity of a 51-qubit cluster state. In contrast, in this experiment, only $1.0470\times 10^7$ samples were needed to effectively estimate the fidelity of a 95-qubit cluster state, with a lower fidelity error.

\section{Correlations of the measurement error}
In this section, we will demonstrate that our processor exhibits very small measurement correlation errors. Experimentally, we characterize the correlated measurement error by covariance
\begin{equation}
	{\mathbf {cov}}[E_i,E_j]=\mathbb{E}[E_i{\wedge}E_j]-\mathbb{E}[E_i]\mathbb{E}[E_j], 
\label{si:cov}
\end{equation}
where $E_{i(j)}$ denote the event that the measurement error appears in $Q_{i(j)}$, $E_i{\wedge}E_j$ denote the event of joint error with $E_i$ and $E_j$.

\begin{table}[!ht]
\renewcommand{\baselinestretch}{1}\selectfont
	\centering
	\begin{adjustbox}{max width=\linewidth}
\begin{tabular}{|l|l|l|l|}
		\hline
		\thead{Exepetiment} & \thead{Number of qubits} & \thead{Number of random states } & \thead{Number of repeat measure } \\ \hline
		1D cluster state & 95 & 21500 & 3000 \\ \hline
		2D cluster state (sparse) & 72 & 12200 & 3000 \\ \hline
		2D cluster state (full) & 57 & 7200 & 3000 \\ \hline
		DJ algorithm & 65 & 11600 & 3000 \\ \hline
	\end{tabular}
\end{adjustbox}
	\caption{The numbers of random initial state and repeat measurement in our experiments for calibrated correlated error.}
	\label{si:random_table}
\end{table}

We calibrated the correlated measurement error by the random state preparation scheme~\cite{Cao2023}. The experiment parameters for characterizing correlated measurement errors in cluster state preparations and the Deutsch-Jozsa (DJ) algorithm demonstration  are shown in Table~\ref{si:random_table}. For the 95Q 1D cluster state preparation experiment, we uniformly randomly generated 21,500 initial input states in $\{0,1\}^{95}$, and independently repeated the measurement 3,000 times for each state. The results of the covariance for the measurement error of all qubit pairs are shown in Fig.~\ref{si:fig:corralateM_fig}a. According to the topological relationships of qubit pairs on the processor, which can be classified into three categories: adjacent pairs (coupled by couplers, a total of 150 pairs), share-readout-line pairs (with the same Purcell-filter and reading-signal line, a total of 257 pairs), and other non-local pairs (remaining cases, a total of 4,058 pairs). In Fig.~\ref{si:fig:cov_bar_fig}, we present the proportion of qubit pairs with any type of covariance exceeding $0.1\%$, relative to the total number of measured pairs.

\begin{figure}[htbp]
\renewcommand{\baselinestretch}{1}\selectfont
    \centering
    \includegraphics[width=\linewidth,height=0.48\textheight,keepaspectratio]{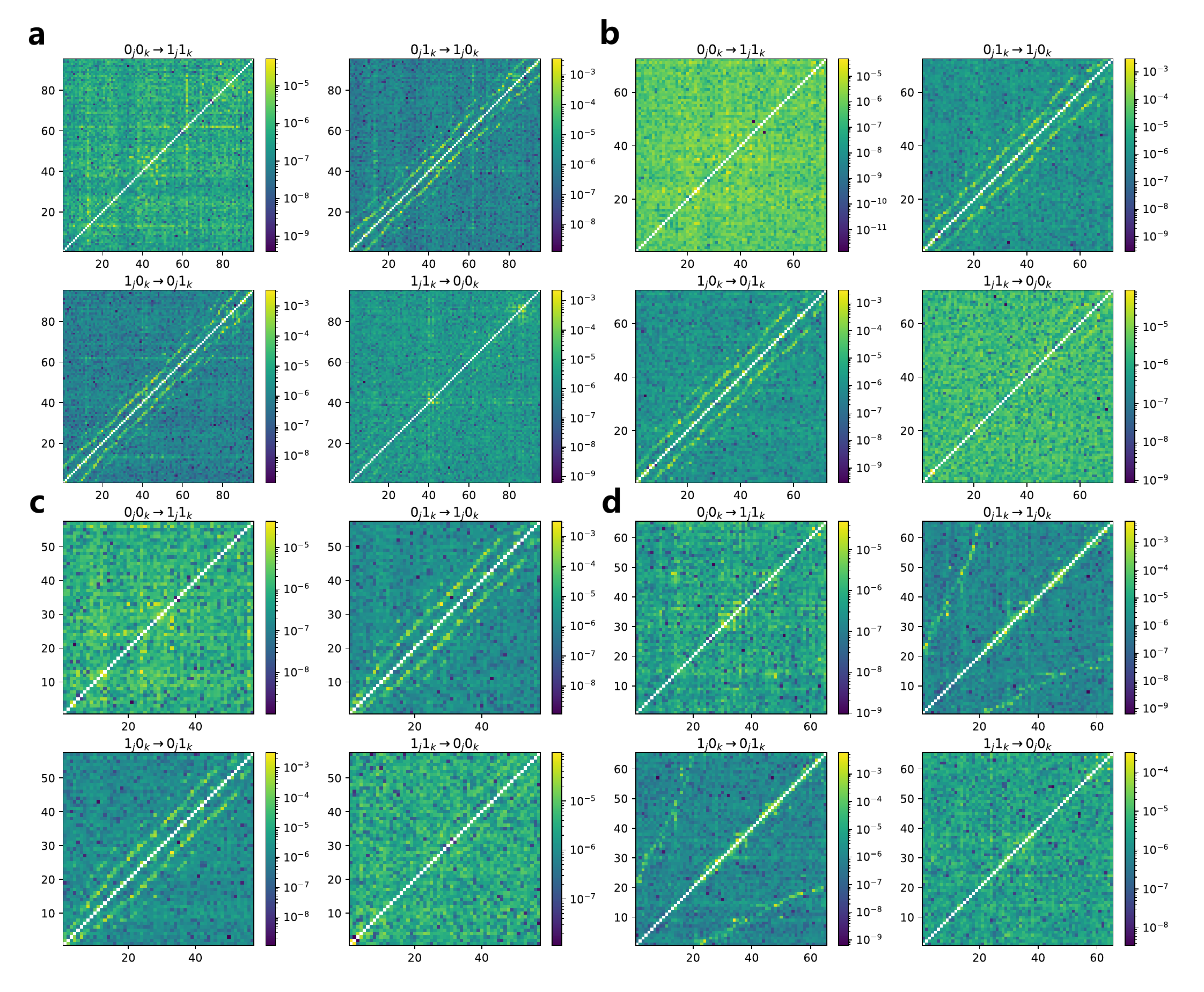}
    \caption{The covariance of all flip errors for qubit pairs in (a) the 95Q 1D cluster state, (b) the 72Q sparse-2D cluster state, (c) the 57Q full-2D cluster state, and (d) the DJ algorithm demonstration with 65Q 2D cluster state. 
    The $jk \rightarrow \bar{j}\bar{k}$ denotes the covariance of two qubits with input state $|j, k\rangle$ and output state $|\bar{j},\bar{k}\rangle$. }
    \label{si:fig:corralateM_fig}
\end{figure}

The experimental results indicate that the correlated measurement error mainly arises from adjacent pairs. This is because the AC-Stark effect of the readout waveform will cause a shift in qubit frequencies, leading to frequency collisions between qubits. Such collisions result in state leakage, which in turn introduces measurement correlation. 

To further analyze the impact of correlated measurement errors on the experimental results, we presented the statistical distribution of two-qubit joint error and covariance from the random state preparation data of the one-dimensional cluster state preparation experiment, as shown in Fig.~\ref{si:fig:cov_joint_fig}. The results show that the median covariance is at least one order less than the joint measurement error. Therefore, we can calculate the cluster state fidelity using the measurement error mitigation by the Tensor Product(TP) scheme while neglecting the correlated measurement errors. Additionally, in the following chapters, we provide a more detailed discussion of the impact of correlated errors. For comparison, using the Continuous Time Markov Processes (CTMP) scheme, we present the cluster fidelity results obtained after considering the main correlated errors, which are consistent with the results by the tensor-product scheme. 

\begin{figure}[htbp]
\renewcommand{\baselinestretch}{1}\selectfont
    \centering
    \includegraphics[width=\linewidth,height=0.48\textheight,keepaspectratio]{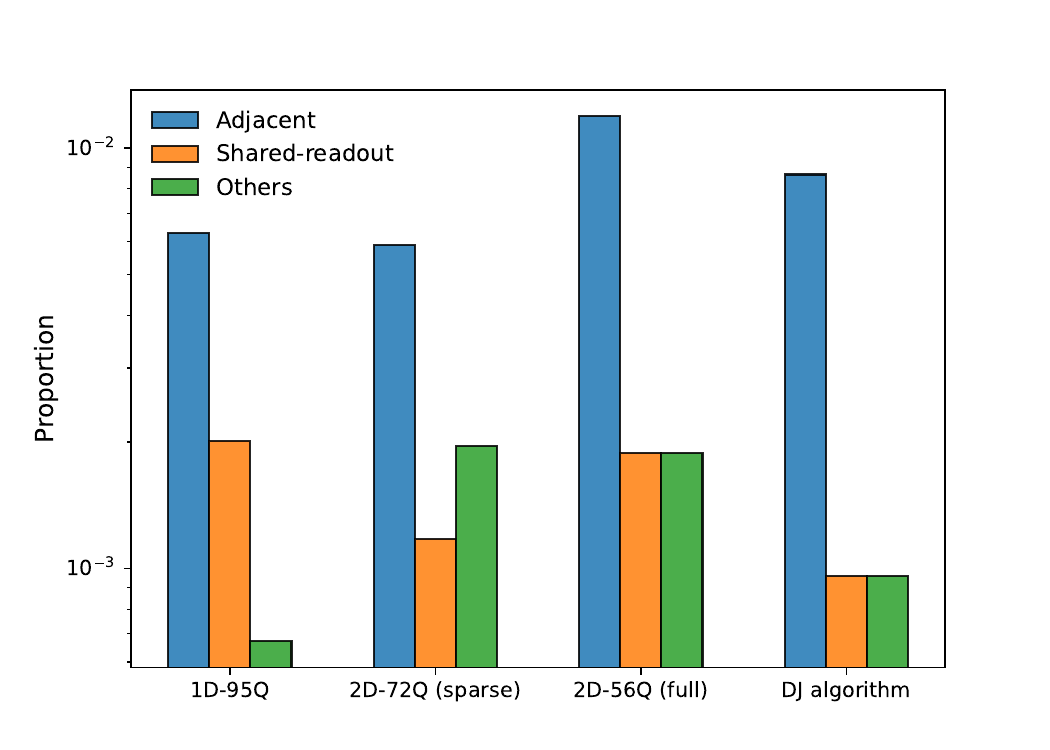}
    \caption{Proportion of qubit pairs with covariance exceeding $0.1\%$ across the three qubit pair categories: adjacent, shared-readout, and others. The proportions are calculated relative to the total number of measured pairs for each experiment.}
    \label{si:fig:cov_bar_fig}
\end{figure}

\begin{figure}[htbp]
\renewcommand{\baselinestretch}{1}\selectfont
    \centering
    \includegraphics[width=\linewidth,height=0.48\textheight,keepaspectratio]{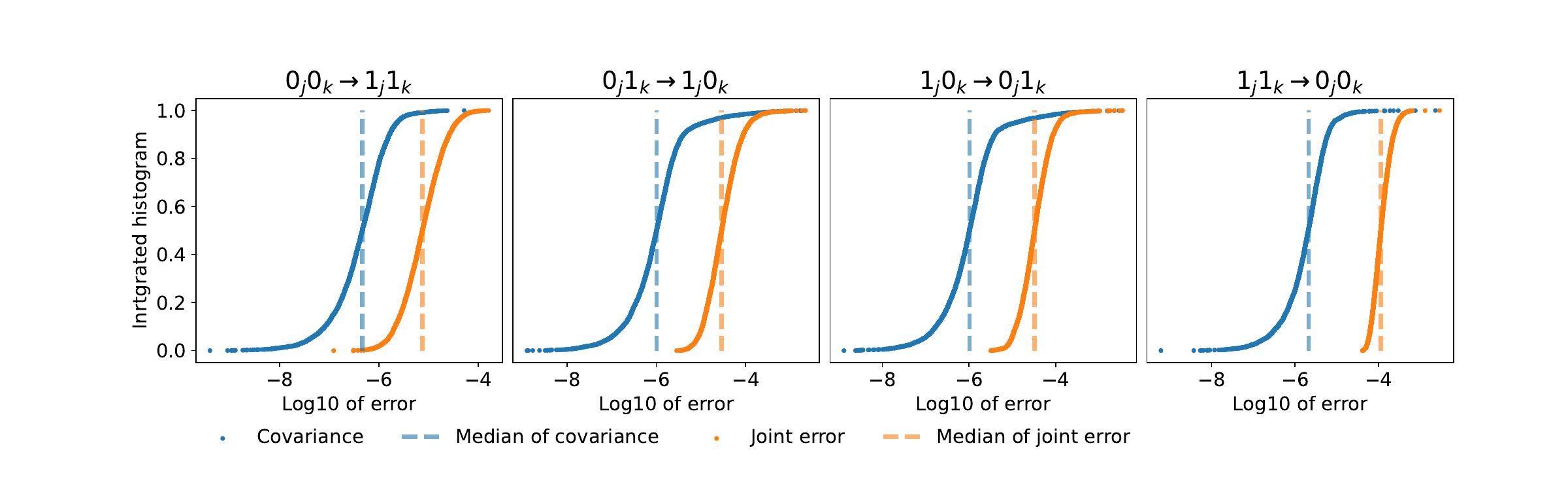}
    \caption{Distribution of covariance and joint-error for the qubit pairs having the same device as 95Q 1D cluster state, where $jk \rightarrow \bar{j}\bar{k}$ denotes the covariance of two qubits with input state $|j, k\rangle$ and output state $|\bar{j},\bar{k}\rangle$. }
    \label{si:fig:cov_joint_fig}
\end{figure}

\section{Correction of measurement errors}
\label{si:sec:readout_error_mitiga}
\subsection{Theoretical framework}
Measurement errors play a critical role in the accuracy of quantum device readout. Accurate calibration of these errors is crucial for improving the overall performance of quantum computations, particularly in enhancing fidelity and enabling reliable entanglement witnessing in cluster states. 

In the absence of measurement errors, the estimation \( v \) can be obtained by aggregating the results of \( M \) distinct measurement settings \( P_i \) (on $n$ qubits), with each setting being repeated \( K \) times to gather sufficient statistics. The estimation formula for \( v \) is given by
\begin{equation}
    v = \frac{1}{MK}\sum_{m=1}^M \operatorname{sign}(P_m)\sum_{k=1}^K O(S_{mk}),
\end{equation}
where \( S_{mk} \) denotes the outcome string obtained from the \( k \)-th repetition of the \( m \)-th measurement setting, and \( O(S) \) is a function that maps the outcome string \( S \) to a scalar value. Specifically, \( O(S) \) is defined as 
\begin{equation}
    O(S) = \prod_{i=1}^n \mu(s_i):=\prod_{i=1}^n \left(1 - 2s_i\right),
\end{equation}
where \( s_i \) represents the binary outcome for the \( i \)-th qubit. 

To account for the impact of readout errors, we introduce the $2^n\times 2^n$ readout noise matrix $\Lambda$, which represents the cumulative noise affecting the measurement outcomes of a quantum state. 
This matrix models the deviation between the ideal measurement outcomes and the actual, noisy observations. 
To correct for these deviations, we compute the inverse of the noise matrix, $\Lambda^{-1}$, and apply it to the observed measurement statistics. The calibrated estimation becomes
\begin{equation}
    v_c = \frac{1}{MK}\sum_{m=1}^M \operatorname{sign}(P_m)\sum_{k=1}^K \sum_x O(x)\langle x|\Lambda^{-1}|S_{mk}\rangle.
    \label{si:eq:calibration}
\end{equation}
This process effectively ``undoes" the influence of the readout noise, allowing us to recover a more accurate representation of the underlying quantum state.

The error in estimating the expectation of a random variable using the mean of a finite number of samples is governed by Hoeffding's inequality. 
For a random variable $v$, which is the average of measured outcomes, Hoeffding's inequality provides the following bound on the probability of deviation from the true expectation $\mathrm{E}[v]$
\begin{equation}
    \mathrm{Pr}\left(\left|v - \mathrm{E}[v]\right| \geq \epsilon\right) \leq 2 \exp\left(-\frac{2 N^2 \epsilon^2}{\sum_{m=1}^N (b_m - a_m)^2}\right),
    \label{si:eq:error_Hoeffding}
\end{equation}
where $N = MK$ is the total number of measurements, and $b_m$ and $a_m$ are the upper and lower bounds of the $m$-th random variable, respectively.
For our framework, the measurement outcomes are bounded between $-1$ and 1, so $b_m = 1$ and $a_m = -1$. 
Substituting these bounds, the denominator simplifies to $4N$, and Hoeffding's inequality becomes $\mathrm{Pr}\left(\left|v - \mathrm{E}[v]\right| \geq \epsilon\right) \leq 2 \exp\left(- N \epsilon^2/2\right)$.
To estimate the error $\epsilon$ with a confidence level of $1-\delta$, we solve for $\epsilon$ such that the probability of error is $\delta$
\begin{equation}
    \epsilon = \sqrt{\frac{2\log(2/\delta)}{N}}.
\end{equation}
Alternatively, for a desired error $\epsilon$, the total number of measurements $N$ required to achieve a confidence level $1-\delta$ can be expressed as
\begin{equation}
    N \geq \left\lceil \frac{2\log(2/\delta)}{\epsilon^2} \right\rceil.
\end{equation}

Since the total number of measurements is distributed as $N=MK$, the number of independent measurement settings can be derived accordingly. 

Through a straightforward derivation, it becomes evident that the presence of bit-flip errors effectively reduces the average value of $v$. 
To counteract this influence, we apply the inverse of the noise operation, denoted as $\Lambda^{-1}$, which restores the original scale by ``dilating''  $v$ back. 
However, this correction process introduces a trade-off: while it compensates for the shrinkage, it simultaneously broadens the range of the random variables, thereby amplifying the associated statistical uncertainty. 
As a result, the error in the final estimation increases, necessitating a higher number of measurements for accurate calibration.
To formalize this intuition, we define the \textit{overhead} of the error mitigation method as
\begin{equation}
    \Gamma = \max _S \sum_x\left|\langle x| \Lambda^{-1}| S\rangle \right|.
    \label{si:eq:gamma_eq}
\end{equation}
This quantity measures the extent to which the error mitigation procedure amplifies the variance. 
Its impact on the error analysis modifies the bounds on the estimation error and the required sample size as follows
\begin{equation}
    \epsilon = \Gamma\sqrt{\frac{2\log(2/\delta)}{N}}, \quad N \geq \left\lceil \frac{2\Gamma^2\log(2/\delta)}{\epsilon^2} \right\rceil.
    \label{si:eq:errorbar}
\end{equation}
Here, $\epsilon$ denotes the estimation error, $\delta$ the confidence interval, and $N$ the number of measurements.
Furthermore, if we take into account the assumption of the normal distribution and $\Gamma$ as the upper bound estimate of the variance, we can have a tighter error estimation, namely
\begin{equation}
    \Pr \left( \left| \bar{v}_c - \mathbb{E}(v_c) \right| \leq \frac{3\Gamma}{\sqrt{N}} \right) \geq 99.7\%,
\end{equation}
and thus the estimation error and the required sample size are
\begin{equation}
    \epsilon = \frac{3\Gamma}{\sqrt{N}}, \quad N \geq \left\lceil \frac{9\Gamma^2}{\epsilon^2} \right\rceil.
\end{equation}

We further introduce the concept of \textit{noise strength}, defined as
\begin{equation}
    \gamma = -\min_{x\in\{0,1\}^n}\langle x|G|x\rangle,
\end{equation}
where $G = \log(\Lambda)$. 
The relationship between $\Gamma$ and $\gamma$ is approximately given by $\Gamma \approx e^{2\gamma}$.

Due to the enormous size of $\Lambda$ ($2^n\times 2^n$), we usually cannot compute it nor its inverse, which urges us to find a simpler model to decompose it into smaller subspaces. 
In this section, we explore two distinct noise models to address readout error mitigation: the qubit-independent noise model and the correlated Markovian noise model. 

\subsubsection{Independent readout model and TP correction}
Under the qubit-independent noise model, we assume that readout errors act independently on each qubit in the quantum system. 
This assumption is often valid for devices where noise sources, such as measurement crosstalk, are minimal or negligible. 
By treating each qubit's noise contribution as independent, the overall noise matrix $\Lambda$ can be expressed as a tensor product of individual qubit noise matrices
\begin{equation}
    \Lambda = \bigotimes_{i=1}^n \Lambda_i,
    \label{si:eq:TP}
\end{equation}
where $\Lambda_i$ represents the noise matrix corresponding to the $i$-th qubit's readout process. Thus 
Each $\Lambda_i$ is typically a \(2 \times 2\) stochastic matrix,
\begin{equation}
    \Lambda_i=\left[\begin{array}{ll}
    p(0 | 0) & p(0 | 1) \\
    p(1 | 0) & p(1 | 1)
    \end{array}\right]:=\left[\begin{array}{cc}
    1-\epsilon_i & \eta_i \\
    \epsilon_i & 1-\eta_i
    \end{array}\right],
    \label{si:eq:TP_form}
\end{equation}
which captures the probabilities of correctly or incorrectly reading out the two computational basis states, \( \ket{0} \) and \( \ket{1} \).

The independence assumption significantly simplifies the calibration process by enabling the characterization of each qubit's readout noise independently. 
This decouples the calibration of individual qubits from the collective behavior of the entire system, thus beneficial for scalability.
In practice, we prepare the quantum device with \(2n^2\) randomly chosen 01-bit strings as input states, to gain insights into two-qubit crosstalk effects (though this model may not strictly require it). 
By repeating this process for preparing each initial state multiple times and measuring the outcomes, we collect sufficient statistics to estimate the probabilities associated with each computational basis state on individual qubits. 
Specifically, we estimate the probabilities of state transitions (bit-flip errors) for each qubit as
\begin{equation}
    \begin{aligned}
    \hat{\epsilon}_i & = \frac{\sum_{x, y} m(y, x)\left\langle 1 | y_i\right\rangle\left\langle x_i | 0\right\rangle}{\sum_{x, y} m(y, x)\left\langle x_i | 0\right\rangle}, \\
    \hat{\eta}_i & = \frac{\sum_{x, y} m(y, x)\left\langle 0 | y_i\right\rangle\left\langle x_i | 1\right\rangle}{\sum_{x, y} m(y, x)\left\langle x_i | 1\right\rangle},
\end{aligned}
\end{equation}
where \(m(y, x)\) represents the number of occurrences of the measured output state \(y\) when the input state was \(x\). 

For the tensor product (TP) form matrix \(\Lambda=\bigotimes \Lambda_i\), one of its key advantages is the ability to perform operations efficiently at the level of individual qubits. 
Specifically, the inverse of \(\Lambda\) can be computed directly as \(\Lambda^{-1}=\bigotimes \Lambda_i^{-1}\), leveraging the independence of the qubits in the tensor product structure. 
This decomposition is particularly useful when addressing readout errors, as it allows for these errors to be mitigated independently for each qubit without requiring a computationally expensive inversion of the full matrix \(\Lambda\).
In this framework, each \(2 \times 2\) sub-matrix \(\Lambda_i\) from Eq.~(\ref{si:eq:TP_form}) has a well-defined analytical inverse given by  
\begin{equation}
    \hat{\Lambda}_i^{-1}= \frac{1}{1-\hat{\epsilon}_i-\hat{\eta}_i}\left[
    \begin{array}{cc}
    1-\hat{\eta}_i & -\hat{\eta}_i \\
    -\hat{\epsilon}_i & 1-\hat{\epsilon}_i
    \end{array}\right],
    \label{si:eq:inverse}
\end{equation}
from the estimated probabilities of bit-flip errors derived from prior characterization of the measurement process. 
This explicit form enables an efficient correction of measurement errors, provided the condition \(1-\hat{\epsilon}_i-\hat{\eta}_i > 0\) is satisfied, which ensures the inverse matrix remains well-defined.

Substituting this correction into the calibrated estimation formula in Eq.~(\ref{si:eq:calibration}), and incorporating the tensor product structure as shown in Eq.~(\ref{si:eq:TP}), we obtain
\begin{equation}
    v_c = \frac{1}{MK}\sum_{m=1}^M \operatorname{sign}(P_m)\sum_{k=1}^K \prod_{i=1}^n \sum_{x_i\in\{0,1\}} \mu(x_i)\langle x_i|\Lambda_i^{-1}|s_i\rangle.
    \label{si:eq:TP_final}
\end{equation}
This expression highlights that the full \(2^n\)-dimensional Hilbert space can be reconstructed by combining the \(n\) qubit-wise subspaces, enabling efficient computation in systems with large numbers of qubits.

Leveraging the inverse matrix provided in Eq.~(\ref{si:eq:inverse}), we derive the overhead and noise strength as
\begin{equation}
    \Gamma_{\text{TP}} = \prod_{i=1}^n \frac{1+\left|\epsilon_i-\eta_i\right|}{1-\epsilon_i-\eta_i}, \quad \gamma_{\text{TP}} = \sum_{i=1}^n \max\{\epsilon_i, \eta_i\}.
    \label{si:eq:TP_overhead}
\end{equation}
Here, $\epsilon_i$ and $\eta_i$ represent the parameters characterizing the noise affecting the $i$-th qubit. 
These expressions illustrate the cumulative nature of noise effects in a system with $n$ qubits.
The computational complexity for estimating quantities governed by Eq.~(\ref{si:eq:TP_final}) scales as $O(nMK) = O\left(n\epsilon^{-2} \Gamma_{\text{TP}}^2 \right)$, or equivalently $O\left(n\epsilon^{-2} \exp(4\gamma_{\text{TP}})\right)$, where $M$ and $K$ are constants related to the algorithm's iterative structure.

\subsubsection{Correlated readout model and CTMP correction}
Although the qubit-independent noise model is both economical and often sufficiently accurate, it fails to capture correlations between two-qubit interactions. 
Experimental evidence shows that the state-switching probability of a qubit is frequently influenced by the states of its neighboring qubits. 
This necessitates a more comprehensive model that accounts for such correlations.

To address this, Bravyi et al.~\cite{Bravyi2021} introduced a noise model incorporating two-qubit correlations within the framework of Continuous-Time Markov Processes (CTMP).
This model provides a systematic way to include both one-qubit and two-qubit noise interactions.

The correlated Markovian noise model is mathematically defined as
\begin{equation}
    \Lambda = \exp(G), \quad G =\sum_{i=1}^{2n^2}r_i G_i
    \label{si:eq:exp_gen}
\end{equation}
where $\exp(G)=\sum_{\alpha=0}^{\infty} G^\alpha / \alpha!$ is the matrix exponential, $r_i$ is the error rate coefficient which can be calibrated by system measurement outcomes, and $G_i$ are noise generators that encompass various types of one-qubit and two-qubit bit-flip noise interactions.
The complete set of generators $G_i$, along with their corresponding noise processes and multiplicities, is detailed in Table~\ref{si:tab:CTMP_gen}, including single-qubit bit-flip errors and two-qubit correlated errors.
\begin{table}[h!]
\renewcommand{\baselinestretch}{1}\selectfont
    \centering
    \begin{adjustbox}{max width=\linewidth}
\begin{tabular}{c|c|c}
    \hline
       Generator $G_i$  &  Readout error & Number of $G_k$\\\hline
       $|1\rangle\langle0|_{j}-|0\rangle\langle0|_j$  &  $0_j\rightarrow 1_j$ & $n$\\\hline
       $|0\rangle\langle1|_{j}-|1\rangle\langle1|_j$  &  $1_j\rightarrow 0_j$ & $n$\\\hline
       $|10\rangle\langle01|_{jk}-|01\rangle\langle01|_{jk}$  &  $01_{jk}\rightarrow 10_{jk}$ & $n(n-1)/2$\\\hline
       $|01\rangle\langle10|_{jk}-|10\rangle\langle10|_{jk}$  &  $10_{jk}\rightarrow 01_{jk}$ & $n(n-1)/2$\\\hline
       $|11\rangle\langle00|_{jk}-|00\rangle\langle00|_{jk}$  &  $00_{jk}\rightarrow 11_{jk}$ & $n(n-1)/2$\\\hline
       $|00\rangle\langle11|_{jk}-|11\rangle\langle11|_{jk}$  &  $11_{jk}\rightarrow 00_{jk}$ & $n(n-1)/2$\\\hline
    \end{tabular}
\end{adjustbox}
    \caption{One-qubit and two-qubit generators of correlated Markovian noise, where $j, k \in [n]$ and $j < k$. The total number of generators is $2n^2$.}
    \label{si:tab:CTMP_gen}
\end{table}

To describe the correlated readout error between the $j$-th and $k$-th qubits, we use a $4 \times 4$ matrix $\Lambda(j, k)$, defined as
\begin{equation}
    \langle u| \Lambda(j, k)|v\rangle=\operatorname{Pr}\left[y_{j, k}=u \mid x_{j, k}=v, y_{-j,-k}=x_{-j,-k}\right],
    \label{si:eq:conditional}
\end{equation}
where $u, v \in \{00, 01, 10, 11\}$, and $x_{j,k}$ represents the two-bit binary measurement outcomes for the $(j,k)$ qubits, while $x_{-j,-k}$ represents the measurement outcomes for all other $n-2$ qubits.

Given that the readout error rates $r_i$ are small, we approximate the $(j, k)$-th term of $G$ as:
\begin{equation}
    \hat{G}(j,k) = \mathcal{P}(\log(\hat{\Lambda}(j,k))), \quad j<k,
    \label{si:eq:log_gen}
\end{equation}
where $\mathcal{P}$ is a channel that maps all negative off-diagonal elements in the input matrix to zero, and $\log$ denotes the matrix logarithm, with the branch chosen such that $\log(\hat{\Lambda}(j, k)) = 0$ if $\hat{\Lambda}(j, k) = I$. 
This approximation is not exact, as the matrix exponential and logarithm are usually affected by commutation. 
The higher-order terms arising from commutation are given by
\[
\begin{aligned}
\exp(G) = \exp\left(\sum_{i=1}^{2n^2} r_i G_i\right)
&= \left(\prod_{i=1}^{2n^2} \exp(r_i G_i)\right) \cdot \exp(\text{higher-order commutators}) \\
&= \left(\prod_{j,k \in [n]} \exp\left[G(j,k)\right]\right) \cdot \left(1 + O(r^2)\right) \\
&= \left(\prod_{j,k \in [n]} I \otimes \cdots \Tilde{\Lambda}(j,k) \cdots \otimes I\right) \cdot \left(1 + O(r^2)\right),
\end{aligned}
\]
where $r=\max_i r_i$. 
However, \(\Tilde{\Lambda}(j,k)\) is not equivalent to $\Lambda(j,k)$ due to residual commutation effects.
The relationship between $\Lambda(j,k)$ and \(\Tilde{\Lambda}(j,k)\) can be expressed as
\[
\Lambda(j,k)=\operatorname{Tr}_{\{-j,-k\}}\left(\prod_{j,k\in [n]} I \bigotimes \cdots \Tilde{\Lambda}(j,k)\cdots\bigotimes I\right)
\]
where the partial trace removes all qubits except $j$ and $k$.
Because different \(\Tilde{\Lambda}(j,k)\) terms interact through the non-commutation, the resulting $\Lambda(j,k)$ depends on contributions from the broader system, and it becomes clear that using $\Lambda(j,k)$ to approximate $\Tilde{\Lambda}(j,k)$ is inherently a first-order, linear approximation. 
Fortunately, the error rates are usually small enough to support our approximation.

To estimate the error rates and the stochastic matrix $\Lambda(j, k)$ associated with the quantum processor, we generate $2n^2$ randomly chosen 01-bit strings as input states. 
Each initial state is prepared and measured multiple times to obtain statistically robust results. 
This procedure ensures sufficient sampling across the Hilbert space for accurate characterization of the system's noise properties.

From the definition in Eq.~(\ref{si:eq:conditional}), the elements of the estimated stochastic matrix $\Lambda(j, k)$ can be expressed as
\begin{equation}
    \langle u| \hat{\Lambda}(j, k)|v\rangle = \frac{\sum_{x, y} m(y, x) \left\langle w | y_{j,k} \right\rangle \left\langle x_{j,k} | v \right\rangle \left\langle x_{-j,-k} | y_{-j,-k} \right\rangle}{\sum_{x, y} m(y, x) \left\langle x_{j,k} | v \right\rangle \left\langle x_{-j,-k} | y_{-j,-k} \right\rangle},
\end{equation}
where $m(y, x)$ represents the number of counts of the measured output state $y$ for a given input state $x$. 
This ratio quantifies the conditional probabilities governing the stochastic process for qubits $j$ and $k$ while marginalizing over the remaining $n-2$ qubits. 
Here, $x_{j,k}$ and $y_{j,k}$ denote the states of the $j$-th and $k$-th qubits, while $x_{-j,-k}$ and $y_{-j,-k}$ capture the states of all other qubits.

Using Eq.~(\ref{si:eq:log_gen}), we compute the two-qubit error rates $r_i$, which describe transitions between specific qubit-pair states. 
For instance, the rate for the $00_{jk} \rightarrow 11_{jk}$ transition is estimated as
\[
\hat{r}{(00_{jk} \rightarrow 11_{jk})} = \langle 11 | \hat{G}(j, k) | 00 \rangle,
\]
where $\hat{G}(j, k)$ is the error generator matrix derived from the stochastic matrix $\hat{\Lambda}(j, k)$.
Similarly, one-qubit error rates can be inferred by summing over contributions from the relative two-qubit terms across all pairs involving the qubit of interest. 
For example, the bit-flip error rate for the $0_j \rightarrow 1_j$ transition is given by
\[
\hat{r}{(0_{j} \rightarrow 1_{j})} = \frac{1}{2(n-1)} \sum_{k \neq j} \left( \langle 10 |_{jk} \hat{G}({j, k}) | 00 \rangle_{jk} + \langle 11 |_{jk} \hat{G}({j, k}) | 01 \rangle_{jk} \right).
\]
Here, the factor $1/(2(n-1))$ normalizes the sum to account for all qubit pairs involving qubit $j$.

After calibrating the elements of the stochastic matrices, we are ready to compute the inverse of $\Lambda$. 
As detailed in Eq.~(\ref{si:eq:exp_gen}), the inverse is related to the generator $G$ through the matrix exponential, and using the decomposition $Q = I + \gamma^{-1} G$, we can rewrite $\Lambda^{-1}$ as
\begin{equation}
    \begin{aligned}
    \Lambda^{-1} = \exp(-G) &= \exp\left(\gamma I - \gamma Q\right) \\
    &= \exp(\gamma)\sum_{\alpha=0}^\infty \frac{(-\gamma)^\alpha Q^\alpha}{\alpha!}.
    \end{aligned}
    \label{si:eq:inverse_combinition}
\end{equation}
where the parameter $\gamma = -\min _{x \in \{0,1\}^n} \langle x| G |x\rangle$. 
Interestingly, $\gamma$ can be viewed as a classical Ising-like Hamiltonian with two-spin interactions. 
Heuristic methods such as simulated annealing or genetic algorithms can provide effective solutions in practice.

Instead of explicitly summing infinite terms in the series expansion, we propose a sampling-based method to estimate $\Lambda^{-1}$. 
From Eq.~(\ref{si:eq:inverse_combinition}), the inverse matrix can be expressed as
\[\Lambda^{-1} = \sum_{\alpha=0}^\infty c_\alpha S_\alpha,\]
where
\begin{equation}
    c_\alpha = \frac{e^\gamma (-\gamma)^\alpha}{\alpha!}, \quad S_\alpha = Q^\alpha.
\end{equation}

The coefficients $c_\alpha$ correspond to the terms in a weighted sum, and the normalization coefficient $\Gamma$ is given by
\begin{equation}
    \Gamma = \sum_{\alpha=0}^\infty |c_\alpha| = \sum_{\alpha=0}^\infty \frac{e^\gamma \gamma^\alpha}{\alpha!} = e^{2\gamma}.
    \label{si:eq:Gamma_exp}
\end{equation}
Thus, we can rewrite $\Lambda^{-1}$ as
\begin{equation}
    \Lambda^{-1} = \Gamma \sum_{\alpha=0}^\infty q_\alpha Q^\alpha,
\end{equation}
where $q_\alpha = c_\alpha / \Gamma = \frac{e^{-\gamma}(-\gamma)^\alpha}{\alpha!}$. 
Notably, $q_\alpha$ follows a Poisson distribution with mean $\gamma$.

Thanks to this structure, we can efficiently approximate $\Lambda^{-1}$ through sampling:
\begin{enumerate}
    \item Sample $\alpha$: Draw $\alpha$ from the Poisson distribution $q_\alpha$.
    \item Apply $Q^\alpha$: Iteratively apply the matrix $Q^\alpha$ to the measurement outcomes $|S\rangle$.
\end{enumerate}
In the second step, the matrix $Q = I + \gamma^{-1} G$ can be interpreted as a Markovian transfer matrix. 
Applying $Q^\alpha$ to the measurement outcomes $|S\rangle$ corresponds to performing a random walk with $Q$ for $\alpha$ steps. 
Since $Q$ is a large, potentially non-local matrix, we can further decompose it as
\[Q = I + \gamma^{-1} \sum_{j,k} G(j,k).\]
For each step of the random walk, we randomly sample from the components of $Q$, selecting either $I$ or $G(j,k)$ with probabilities proportional to their contributions in the decomposition.
For $\alpha$ steps, we sequentially apply the sampled components $G_{r_j}$ to the state $|S\rangle$.
This decomposition reduces the complexity of handling the full matrix $Q$, and also allows us to make good use of the estimates in Eq.~(\ref{si:eq:log_gen}).

Repeating the above sampling procedure $T$ times, we calculate the calibrated estimation for $M$ measurement settings and $K$ repetitions as
\begin{equation}
    \begin{aligned}
    v_c &= \frac{1}{MK} \sum_{i=1}^M \operatorname{sign}(P_i) \sum_{k=1}^K \sum_x O(x) \langle x | \Lambda^{-1} | S_{ik} \rangle \\
    &= \frac{\Gamma}{MKT} \sum_{i=1}^M \operatorname{sign}(P_i) \sum_{k=1}^K \sum_{t=1}^T \sum_x O(x) \langle x | \prod_{j=1}^{\alpha_t} G_{r_j} | S_{ik} \rangle.
    \end{aligned}
    \label{si:eq:CTMP_final}
\end{equation}
Here, $\Gamma$ serves as a normalization coefficient that effectively scales every outcome, reflecting the overhead introduced by the error mitigation method.
This expression uses the sampled $\alpha_t$ and the corresponding matrix sequence $\prod_{j=1}^{\alpha_t} G_{r_j}$ to approximate the application of $\Lambda^{-1}$ on the measurement outcomes. 
By averaging over $T$ independent trials, the estimator converges to the calibrated value of $v_c$.

The noise strength $\gamma_{\text{CTMP}}$ is determined via classical search algorithms. 
Using Eq.(\ref{si:eq:Gamma_exp}), the overhead is expressed as
\begin{equation}
    \Gamma_{\text{CTMP}} = \exp\left(2\gamma_{\text{CTMP}}\right).
\end{equation}

The computational complexity for evaluating Eq.~(\ref{si:eq:CTMP_final}) is given by
\begin{equation}
    O(n\gamma_{\text{CTMP}}MKT)
    = O\!\left(n\epsilon^{-2}\gamma_{\text{CTMP}}\Gamma_{\text{CTMP}}^{2}\right),
\end{equation}
or equivalently,
\begin{equation}
    O\!\left(n\epsilon^{-2}\gamma_{\text{CTMP}}\exp(4\gamma_{\text{CTMP}})\right).
\end{equation}

Upon careful examination, we observed that the CTMP overhead for the 1D 95-qubit cluster state becomes prohibitively large, leading to an impractically high number of required measurements. 
To address this issue, we propose customizing the summation set $G$ in Eq.~(\ref{si:eq:exp_gen}), allowing for greater flexibility in the CTMP method. The method can be tailored as follows:
\begin{enumerate}
\item When the summation is restricted to single-qubit generators, the CTMP method reduces to the TP method, as no multi-qubit correlations are included.
\item If only nearest-neighbor two-qubit bit-flip noise is considered, the computational overhead is significantly reduced, making the approach more practical for large-scale systems.
\item Including all two-qubit correlations across the system enables the method to capture a broader range of correlated errors.
\item By calculating the two-qubit correlation coefficients for all bit-flip errors, defined as
\begin{equation}
    r(E_1, E_2)=\frac{\operatorname{cov}(E_1, E_2)}{\sqrt{\operatorname{Var}(E_1) \operatorname{Var}(E_2)}},
\end{equation}
where $\operatorname{cov}(E_1, E_2)=\mathbb{E}(E_1, E_2)-\mathbb{E}(E_1)\mathbb{E}(E_2), \operatorname{Var}(E_1)=\mathbb{E}(E_1^2)-\mathbb{E}(E_1)^2$, we can identify qubit pairs with coefficients exceeding a threshold of 0.3 and selectively include them in the summation.
\end{enumerate}

The table below compares the relative noise strength ($\gamma$), overhead ($\Gamma$), and total number of measurements needed in 1D 95-qubit cluster case for the TP model, single-qubit CTMP model, nearest-neighbor CTMP model, full 2-qubit correlation CTMP model and most-correlated CTMP model.

\begin{table}[h!]
\renewcommand{\baselinestretch}{1}\selectfont
    \centering
    \renewcommand{\arraystretch}{1.2}
    \begin{adjustbox}{max width=\linewidth}
\begin{tabular}{|l|c|c|c|c|c|}
        \hline
        & \textbf{TP} & \multicolumn{4}{c|}{\textbf{CTMP}} \\ \hline
        & TP           & TP             & Nearest-neighbor       & Full-correlation & Most-correlated \\ \hline
        Noise strength ($\gamma$)    & 0.96559   & 0.95877   & 2.8609   & 91.093    & 1.7161   \\ \hline
        Overhead ($\Gamma$)    & 6.8976    & 6.8042   & 305.44   & $1.3266\times 10^{79}$       & 30.945            \\ \hline
        Number of measurements & $7.2326\times 10^6$    & $7.0381\times 10^6$    & $1.4183\times 10^{10}$   & $2.6752\times 10^{163}$      &  $1.4557\times 10^{8}$         \\ \hline
    \end{tabular}
\end{adjustbox}
    \caption{Comparison of TP and CTMP results.}
    \label{si:tab:results}
\end{table}

The results are presented in the Table~\ref{si:tab:results} highlights the trade-offs between error suppression performance and computational resources for different noise-mitigation methods.
The most correlated configuration offers a more balanced solution. By selecting qubit pairs with correlation coefficients above a threshold (e.g., 0.3), the method efficiently suppresses significant correlated errors without excessive resource overhead.
We thus use this most correlated configuration in the following numerical simulations for CTMP models. 
The complete experimental datasets documenting system size dependence are presented in Table~\ref{si:tab:1d_cluster_results} (1D chains),~\ref{si:tab:2d_sparse_results} (sparse 2D), and~\ref{si:tab:2d_full_results} (full 2D).

\begin{table}[ht]
\renewcommand{\baselinestretch}{1}\selectfont
\centering
\begin{adjustbox}{max width=\linewidth}
\begin{tabular}{|c|c|c|c|c|c|c|c|c|c|c|}
\hline
Qubits & 5 & 15 & 25 & 35 & 45 & 55 & 65 & 75 & 85 & 95 \\
\hline
$N$ ($\times10^5$) & 4.74 & 7.26 & 10.86 & 14.28 & 21.99 & 28.80 & 25.12 & 39.60 & 80.94 & 104.7 \\
\hline
Stabilizer circuits & 158 & 242 & 362 & 476 & 733 & 960 & 1256 & 1980 & 2698 & 3490 \\
\hline
Repeats per stabilizer & 3000 & 3000 & 3000 & 3000 & 3000 & 3000 & 2000 & 2000 & 3000 & 3000 \\
\hline
$\Gamma_{\mathrm{TP}}$ & 1.11 & 1.46 & 1.91 & 2.33 & 2.88 & 3.68 & 4.77 & 5.95 & 7.49 & 9.11 \\
\hline
Fidelity (\%) & 96.32 & 90.49 & 84.19 & 77.36 & 75.97 & 71.25 & 66.94 & 63.69 & 59.57 & 56.03 \\
\hline
Error bar (\%) & 0.49 & 0.51 & 0.55 & 0.58 & 0.59 & 0.65 & 0.90 & 0.90 & 0.79 & 0.84 \\
\hline
\end{tabular}
\end{adjustbox}
\caption{Experimental results for 1D cluster states}
\label{si:tab:1d_cluster_results}
\end{table}

\begin{table}[ht]
\renewcommand{\baselinestretch}{1}\selectfont
\centering
\begin{adjustbox}{max width=\linewidth}
\begin{tabular}{|c|c|c|c|c|c|c|c|c|c|}
\hline
Qubits & 6 & 15 & 23 & 37 & 46 & 52 & 57 & 62 & 72 \\
\hline
$N$ ($\times10^5$) & 6.80 & 9.18 & 12.00 & 18.86 & 24.02 & 28.66 & 32.56 & 35.46 & 47.86 \\
\hline
Stabilizer circuits & 340 & 459 & 600 & 943 & 1201 & 1433 & 1628 & 1773 & 2393 \\
\hline
Repeats per stabilizer & 2000 & 2000 & 2000 & 2000 & 2000 & 2000 & 2000 & 2000 & 2000 \\
\hline
$\Gamma_{\mathrm{TP}}$ & 1.13 & 1.36 & 1.60 & 2.10 & 2.51 & 2.77 & 3.01 & 3.23 & 3.93 \\
\hline
Fidelity (\%) & 94.03 & 87.66 & 79.10 & 71.16 & 66.46 & 63.42 & 60.82 & 57.85 & 55.19 \\
\hline
Error bar (\%) & 0.41 & 0.42 & 0.44 & 0.46 & 0.48 & 0.49 & 0.50 & 0.51 & 0.54 \\
\hline
\end{tabular}
\end{adjustbox}
\caption{Experimental results for 2D sparse cluster states}
\label{si:tab:2d_sparse_results}
\end{table}

\begin{table}[ht]
\renewcommand{\baselinestretch}{1}\selectfont
\centering
\begin{adjustbox}{max width=\linewidth}
\begin{tabular}{|c|c|c|c|c|c|c|c|}
\hline
Qubits & 6 & 15 & 23 & 37 & 46 & 52 & 57 \\
\hline
$N$ ($\times10^5$) & 9.10 & 11.74 & 15.12 & 23.18 & 29.42 & 34.36 & 40.04 \\
\hline
Stabilizer circuits & 455 & 587 & 756 & 1159 & 1471 & 1718 & 2002  \\
\hline
Repeats per stabilizer & 2000 & 2000 & 2000 & 2000 & 2000 & 2000 & 2000  \\
\hline
$\Gamma_{\mathrm{TP}}$ & 1.12 & 1.34 & 1.58 & 2.06 & 2.45 & 2.70 & 3.00 \\
\hline
Fidelity (\%) & 90.65 & 82.85 & 75.20 & 63.23 & 59.09 & 54.23 & 51.04 \\
\hline
Error bar (\%) & 0.35 & 0.37 & 0.38 & 0.41 & 0.43 & 0.44 & 0.45 \\
\hline
\end{tabular}
\end{adjustbox}
\caption{Experimental results for 2D full cluster states}
\label{si:tab:2d_full_results}
\end{table}

\subsection{Experimental performance under the conventional readout error model}
In this subsection, we evaluate the experimental performance using the conventional readout error model, where the readout correction matrix is based on the empirical matrix $\Lambda_{\mathrm{exp}}$ obtained directly from the calibrated probabilities $F_{00}$ and $F_{11}$ and subsequently applied in both TP and CTMP mitigation. This treatment follows the conventional methodology adopted in large-scale superconducting experiments.

We numerically compute the Tensor Product (TP) error-mitigated fidelity and its associated error for 1D and 2D cluster states, scaling up to 95 qubits, 72 qubits (3-pattern), and 57 qubits (4-pattern), respectively. 
These results are summarized in the figures below, demonstrating the scalability and effectiveness of our error mitigation approach. The main text presents the TP error-mitigated results for both 1D and 2D cluster states.

Beyond the theoretical bounds derived from error analysis, we also analyze the distribution of the average values for each measurement setting in the 1D 95-qubit, 2D 72-qubit, and 2D 57-qubit cases. This distribution provides insight into the variability across measurement outcomes and the reliability of the mitigated fidelity. Specifically, the variance across circuits for the 1D 95-qubit and 2D 72-qubit and 57-qubit cases is computed as $\sqrt{\mathrm{Var}(\overline{v})} = 0.064, 0.053, 0.051$, respectively, as shown in Fig.~\ref{si:fig:circuit_fidelity}. 
The relative variance of the total average estimation, given by $\epsilon = \sqrt{\mathrm{Var}(\overline{v})/M}$, serves as an empirical estimate of fidelity error, yielding values of $0.00108, 0.00108, 0.00114$, where $M$ represents the total number of circuits.

\begin{figure}
\renewcommand{\baselinestretch}{1}\selectfont
    \centering
    \includegraphics[width=\linewidth,height=0.48\textheight,keepaspectratio]{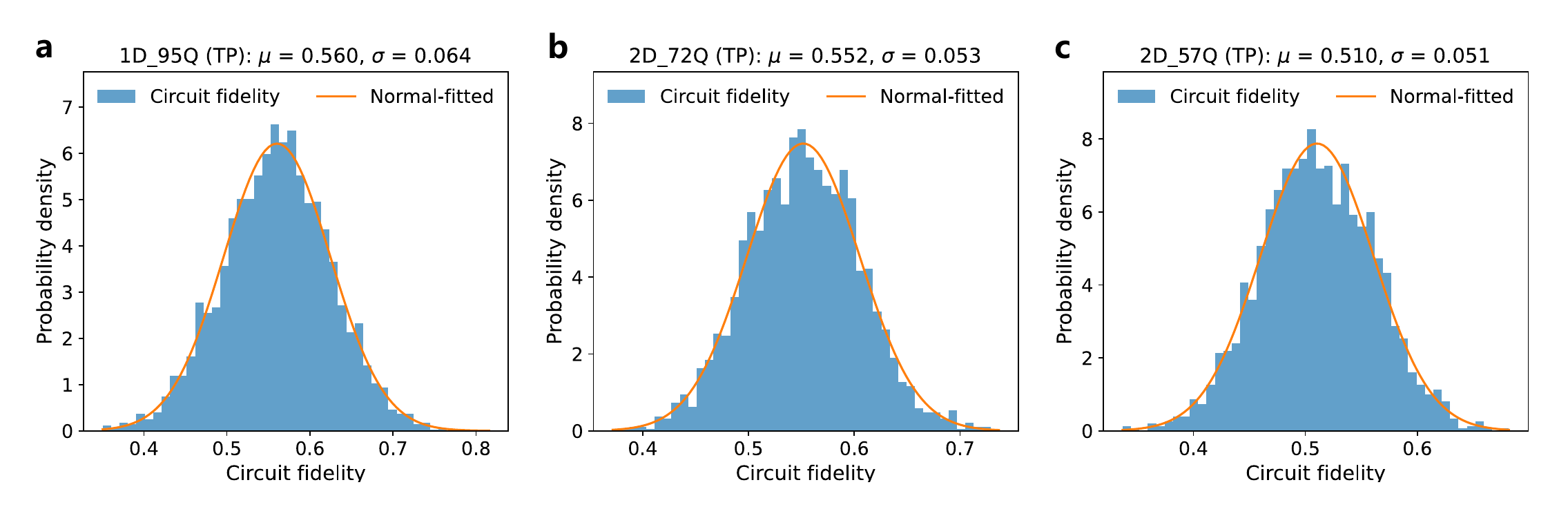}
    \caption{Circuit fidelity ditribution of 1D 95-qubit, 2D sparse 72-qubit and 2D full 57-qubit cases by TP method.}
    \label{si:fig:circuit_fidelity}
\end{figure}

\begin{figure}[htp]
\renewcommand{\baselinestretch}{1}\selectfont
    \centering
    \includegraphics[width=\linewidth,height=0.48\textheight,keepaspectratio]{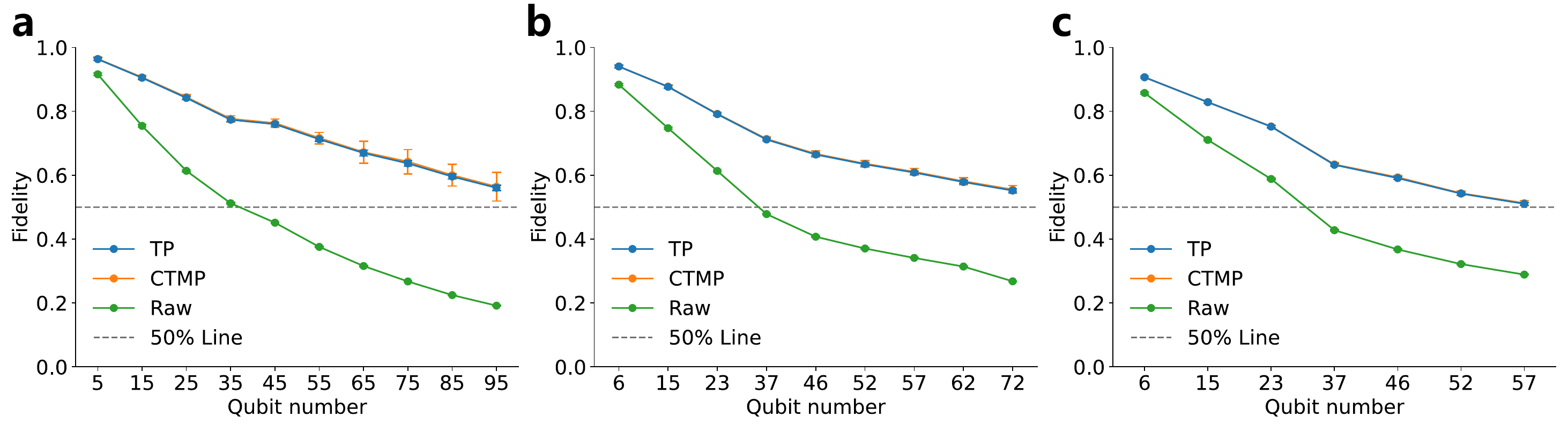}
    \caption{
    Comparison of fidelities obtained using the TP and CTMP readout strategies under the conventional readout error model, together with the unmitigated (``raw'') results, for 1D and 2D cluster states of varying sizes. 
    Fidelity values are computed as the mean stabilizer expectation over a large number of randomly sampled stabilizer circuits, and the error bars denote the 99.7\% confidence bounds derived from the Hoeffding inequality. 
    The sampling configurations for all stabilizer circuits are summarized in Supplementary Tables~\ref{si:tab:1d_cluster_results}, \ref{si:tab:2d_sparse_results}, and \ref{si:tab:2d_full_results}. 
    For the 1D 95-qubit cluster state, the TP, CTMP, and raw fidelities are $0.5603\pm0.0084$, $0.5639\pm0.0453$, and $0.1915\pm0.0009$, respectively. 
    For the 2D 72-qubit (three-pattern) cluster state, the corresponding values are $0.5519\pm0.0054$, $0.5549\pm0.0127$, and $0.2677\pm0.0014$. 
    For the 2D 57-qubit (four-pattern) cluster state, the fidelities are $0.5104\pm0.0045$, $0.5128\pm0.0082$, and $0.2887\pm0.0015$, respectively.
    }
    \label{si:fig:ctmp_compare}
\end{figure}

To mitigate fidelity errors for 1D and 2D cluster states, we also employ the Continuous-Time Markov Processes (CTMP) model. 
A comparison of the CTMP and TP methods, along with the raw cluster-state fidelities without readout error mitigation, is presented in Fig.~\ref{si:fig:ctmp_compare}. While the CTMP model considers the most-correlated qubit pairs to reduce the overhead compared to the full CTMP approach, the improvement in fidelity is marginal, and the required number of measurements remains significantly higher. These observations support the adoption of the Tensor Product method in the main text as a more practical alternative.
In addition to the error analysis, we present the distribution of average values for each measurement setting, as shown in Fig.~\ref{si:fig:circuit_CTMP}. For 1D 95-qubit, 2D 72-qubit, and 2D 57-qubit cluster states, the variance across circuits in this scenario is computed as $\sqrt{\mathrm{Var}(\overline{v})} = 0.118, 0.059, 0.053$. The relative variance of the total average estimation is $\epsilon = \sqrt{\mathrm{Var}(\overline{v})/M} = 0.00200, 0.00121, 0.00118$, where $M$ is the total number of circuits, serving as an empirical estimate of the fidelity error.

\begin{figure}
\renewcommand{\baselinestretch}{1}\selectfont
    \centering
    \includegraphics[width=\linewidth,height=0.48\textheight,keepaspectratio]{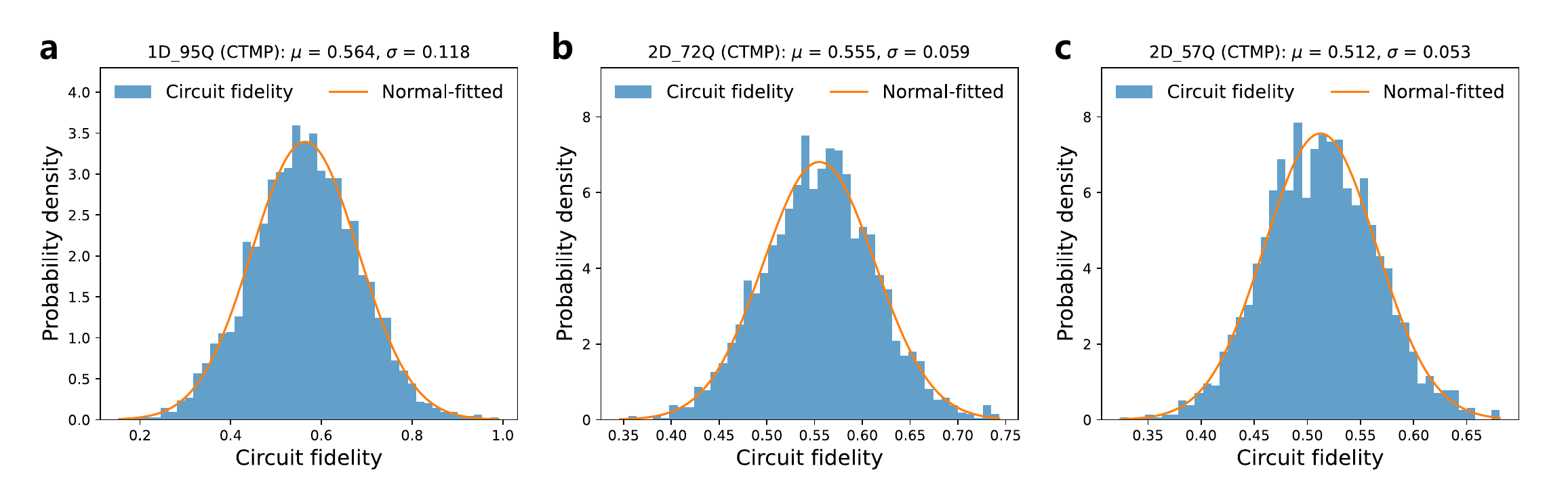}
    \caption{Circuit fidelity ditribution of 1D 95-qubit, 2D sparse 72-qubit and 2D full 57-qubit cases by CTMP method.}
    \label{si:fig:circuit_CTMP}
\end{figure}

If the variance of the underlying distribution is known-or can be tightly estimated from the data-the classical Hoeffding inequality can be significantly outperformed by the Maurer-Pontil empirical Bernstein bound~\cite{maurer2009empiricalbernsteinboundssample}.  In this form one shows that for \(M\) independent measurement settings \(v_1, \dots, v_M\), with sample mean
\[
\bar v = \frac{1}{M}\sum_{m=1}^M v_m
\]
and empirical variance \(\hat\sigma^2\), the following holds at confidence level \(1-\delta\):
\begin{equation}
    \Pr\bigl(|\bar v - \mathbb{E}[v]|\ge\epsilon\bigr)\;\le\;2\exp\!\Biggl(-\frac{M\,\epsilon^2}{2\bigl(\hat\sigma^2 + \frac{2\,\epsilon\sqrt{2\ln(3/\delta)}}{3\sqrt{M}}\bigr)}\Biggr).
\end{equation}
We have applied this bound to calculate the number of measurement settings \(M\) required to guarantee an estimation error \(\epsilon=0.01\) at 99.7\% confidence for various assumed values of \(\sigma\).  For example, when \(\sigma=0.05\) we found \(M\approx473\), and when \(\sigma=0.10\) we obtained \(M\approx1{,}388\).
As shown in Supplementary Fig.~\ref{si:fig:circuit_fidelity} and~\ref{si:fig:circuit_CTMP}, the resulting fidelity distributions closely follow a Gaussian profile, indicating that the number of measurement bases we have employed is sufficient to capture the relevant statistical fluctuations.
Moreover, the actual number of measurement settings used in our experiment largely exceeds these bounds, confirming that our sampling is fully adequate.

\subsection{Refined fidelity analysis with explicit thermal and single-qubit error separation}

In the main text, all cluster-state fidelities are estimated from random-stabilizer measurements using the conventional readout error-mitigation model, in which the readout correction matrix is constructed directly from the calibrated probabilities $F_{00}$ and $F_{11}$. This approach, which is widely employed in large-scale superconducting quantum experiments, effectively models measurement noise as classical readout flips and treats additional effects as part of the effective calibration.  

In practice, however, the calibration procedure also absorbs contributions from (i)~thermal population of the computational excited state prior to measurement and (ii)~single-qubit gate errors associated with measurement-basis rotations. 
These effects are therefore implicitly corrected when applying the empirical inverse readout matrix, which leads to a slight overestimation of the true cluster-state fidelity. 

In this section, we present a refined analysis that separately incorporates these physical contributions and quantifies their individual impact on the fidelity estimation. We note that this refined treatment is included for completeness: all fidelities presented in the main text continue to use the conventional readout error model, ensuring consistency with prior large-scale superconducting experiments.

\subsubsection{Qubit-initialization effects in the readout noise model}

For each qubit, let $q_0$ denote the thermally induced excited-state population before any operation. Even for an ideal preparation sequence, the nominal $\ket{0}$ state is realized as a mixed state
\begin{equation}
    \rho_{0}' = (1-q_0)\ket{0}\!\bra{0} + q_0\ket{1}\!\bra{1}
              = \begin{pmatrix} 1-q_0 & 0 \\[1mm] 0 & q_0 \end{pmatrix},
\end{equation}
where we model thermal excitation as a classical bit-flip process.

Preparation of $\ket{1}$ proceeds through an $X$ gate applied to $\ket{0}$. We assume that thermal excitation acts before the application of the $X$ gate, so that the input state to the $X$ operation is the thermally populated $\rho_{0}'$ introduced above. Applying the ideal $X$ gate to $\rho_{0}'$ yields
\begin{equation}
    X \rho_{0}' X = \begin{pmatrix} q_0 & 0 \\[1mm] 0 & 1-q_0 \end{pmatrix}.
\end{equation}

We then assume that the noise associated with the $X$ operation occurs after the ideal $X$ gate, and model this noisy gate by a Pauli channel with error probability $e_p$,
\begin{equation}
    \mathcal{E}_p(\rho)
    = (1-e_p)\rho 
    + \frac{e_p}{3}\left(X\rho X + Y\rho Y + Z\rho Z\right).
\end{equation}
Acting on $X \rho_{0}' X$ this yields the final prepared $\ket{1}$ state,
\begin{equation}
    \rho_{1}' = \mathcal{E}_p(X \rho_{0}' X)
    = \begin{pmatrix}
        (1-e_p)q_0 + \dfrac{e_p(2-q_0)}{3} & 0 \\[2mm]
        0 & (1-e_p)(1-q_0) + \dfrac{e_p(1+q_0)}{3}
      \end{pmatrix}
    = \begin{pmatrix} q_1 & 0 \\[1mm] 0 & 1-q_1 \end{pmatrix}.
\end{equation}
with an effective excited-state population
\begin{equation}
    q_1 = (1-e_p)q_0 + \frac{2}{3}e_p - \frac{4}{3}e_p q_0 .
\end{equation}
The same expression is obtained if one uses a depolarizing channel with parameter $e_d$ and identifies $e_p = 3e_d/4$ via the average gate infidelity. In the experiment, the Pauli error of the single-qubit primitive gate $\sqrt{X}$ is obtained from single-qubit XEB and denoted by $e_{\mathrm{xeb}}$. Since an $X$ gate is realized by two successive $\sqrt{X}$ rotations, its Pauli error is, to first order in the small-error limit, $e_p \simeq 2 e_{\mathrm{xeb}}$.

During the readout benchmarking experiment, the empirical $2\times2$ matrix $\Lambda_{\mathrm{exp}}$ is obtained from repeated measurements of nominal $\ket{0}$ and $\ket{1}$ preparations. Explicitly,
\begin{equation}
  \Lambda_{\mathrm{exp}} =
  \begin{pmatrix}
    F_{00} & 1-F_{11} \\
    1-F_{00} & F_{11}
  \end{pmatrix},
\end{equation}
where $F_{00}$ and $F_{11}$ denote the observed probabilities of correctly identifying $\ket{0}$ and $\ket{1}$, respectively. Because the actual input states contain thermal population and $X$-gate errors, they are given by $\rho_0'$ and $\rho_1'$ rather than the ideal computational basis states. This effect can be collected into the state-preparation matrix
\begin{equation}
  \mathrm{InI} =
  \begin{pmatrix}
    1-q_0 & q_1 \\
    q_0   & 1-q_1
  \end{pmatrix},
\end{equation}
which maps the nominal logical $\{\ket{0},\ket{1}\}$ ensemble to the actual 
populations entering the readout channel.

Let $\Lambda$ denote the intrinsic readout transform matrix describing genuine measurement flips,
\begin{equation}
  \Lambda =
  \begin{pmatrix}
    1-\lambda_0 & \lambda_1 \\
    \lambda_0   & 1-\lambda_1
  \end{pmatrix}.
\end{equation}
Here $\lambda_0$ is the probability of incorrectly reporting a measured $\ket{0}$ as $\ket{1}$, and $\lambda_1$ is the probability of incorrectly reporting a measured $\ket{1}$ as $\ket{0}$. These two parameters characterize the intrinsic readout noise of the device. 

Since readout acts after state preparation, the empirical readout noise matrix satisfies the decomposition
\begin{equation}
  \Lambda_{\mathrm{exp}} = \Lambda \, \mathrm{InI},
\end{equation}
where $\Lambda$ represents the intrinsic readout flips and $\mathrm{InI}$ captures the state-preparation imbalance arising from thermal excitation and single-qubit gate errors. Solving for $\Lambda$ gives
\begin{equation}
  \Lambda = \Lambda_{\mathrm{exp}} \, \mathrm{InI}^{-1},
\end{equation}
and therefore
\begin{equation}
  \Lambda^{-1} = \mathrm{InI} \, \Lambda_{\mathrm{exp}}^{-1}.
\end{equation}

This relationship shows that $\Lambda_{\mathrm{exp}}^{-1}$ reverses not only the intrinsic readout flips characterized by $\lambda_0$ and $\lambda_1$ but also the state-preparation imbalance contained in $\mathrm{InI}$. 
As a result, applying $\Lambda_{\mathrm{exp}}^{-1}$ to stabilizer measurements partially compensates the effects of thermal excitation and single-qubit gate errors, which causes a slight overestimation of the extracted cluster-state fidelity.

To express the effect of intrinsic readout correction explicitly, we write the readout-corrected stabilizer estimator in a form analogous to the random-stabilizer expression in Eq.~\ref{si:eq:calibration}. 
For a stabilizer $P$, measured $K$ times with outcome bitstrings $S_k = (x_{1,k}, \ldots, x_{n,k})$, the readout-corrected value is
\begin{equation}
    \langle P \rangle_{\mathrm{ro}}
    =
    \frac{1}{K}\,
    \operatorname{sign}(P)
    \sum_{k=1}^{K}
    \sum_{x}
    O(x)\,
    \langle x \,|\, \Lambda^{-1} \,|\, S_k \rangle ,
\end{equation}
where $O(x)$ encodes the deterministic eigenvalue of $P$ associated with the bitstring $x$, and $\Lambda^{-1}$ is the intrinsic readout correction matrix obtained above.

Averaging over all measured stabilizers gives the fidelity obtained after readout correction alone,
\begin{equation}
    F_{\mathrm{ro}}
    = \frac{1}{M} \sum_{P} \langle P \rangle_{\mathrm{ro}}.
\end{equation}
where $M$ is the number of random stabilizers measured. When using the intrinsic readout matrix $\Lambda^{-1}$, the corresponding overhead $\Gamma$ and the estimation error $\epsilon$ retain the same form as defined in Eqs.~\ref{si:eq:gamma_eq} and~\ref{si:eq:errorbar}, with $\Gamma$ evaluated from $\Lambda^{-1}$. 

Because $\Lambda^{-1}$ accounts only for intrinsic readout flips and does not reverse the state-preparation imbalance, the resulting quantity $F_{\mathrm{ro}}$ provides an unbiased readout-corrected fidelity that serves as the baseline for the analysis presented below.

\subsubsection{Single-qubit errors associated with measurement-basis rotations}

In the random-stabilizer protocol, each qubit is measured in one of the Pauli bases $X$, $Y$, or $Z$ (or left idle, corresponding to the identity operator in the stabilizer). A measurement in the $X$ or $Y$ basis requires a single-qubit $\pm\pi/2$ rotation, while a measurement in the $Z$ basis involves an idle period of duration $T_{\mathrm{gate}}$ to maintain circuit alignment. The associated Pauli error rate therefore depends on the measurement basis $B$ and is written as
\begin{equation}
    e_p(B) = 
    \begin{cases}
        e_{\mathrm{xeb}}, & B = X, Y, \\[2mm]
        1 - \exp(-T_{\mathrm{gate}}/T_1), & B = Z ,
    \end{cases}
\end{equation}
where $T_{\mathrm{gate}}$ is the primitive-gate duration. In our experiment, $T_{\mathrm{gate}} = 48~\mathrm{ns}$ and the typical energy relaxation time is $T_1 \approx 50~\mu\mathrm{s}$.

To model the effect of a single measurement-basis operation on qubit $i$, we use the depolarizing channel
\begin{equation}
    \mathcal{E}_{d,i}(\rho) = (1-e_{d,i})\rho + e_{d,i} \frac{I_i}{2}\otimes \mathrm{Tr}_i(\rho),
    \qquad e_{d,i}=\frac{4}{3}e_{p,i} .
\end{equation}
Here $e_{p,i} = e_p(B_i)$ is the Pauli error associated with the chosen measurement basis on qubit $i$.

The experimental cluster state entering the measurement layer can be written as
\begin{equation}
    \rho_{\mathrm{exp}} = F_{\mathrm{exp}}\ket{G}\!\bra{G}
    + (1-F_{\mathrm{exp}})\rho_\perp,
\end{equation}
where $\ket{G}$ is the ideal cluster state and $\rho_\perp$ lies in the orthogonal subspace. 
Applying the single-qubit depolarizing channel $\mathcal{E}_{d,i}$ gives 
\begin{equation}
    F_{\mathrm{out}}
    = \langle G|\mathcal{E}_{d,i}(\rho_{\mathrm{exp}})|G\rangle
    = F_{\mathrm{exp}}\,\alpha + (1-F_{\mathrm{exp}})\,\beta ,
\end{equation}
where
\begin{align}
  \alpha
  &= \langle G|\mathcal{E}_{d,i}(\ket{G}\!\bra{G})|G\rangle
   = (1-e_{d,i}) + \frac{e_{d,i}}{4}
   = 1 - e_{p,i}, \\[1.5mm]
  \beta
  &= \langle G|\mathcal{E}_{d,i}(\rho_{\perp})|G\rangle \approx 0 .
\end{align}
The approximation $\beta \approx 0$ follows from the fact that $\rho_{\perp}$ has support only in the orthogonal subspace of $\ket{G}$, and a single-qubit Pauli channel returns negligible population to $\ket{G}$.

For a stabilizer $P$, let $S(P)$ be the set of qubits on which $P$ acts nontrivially, and let $B_i$ denote the measurement basis used on each qubit $i \in S(P)$.  
Only qubits measured in the Pauli bases $X$, $Y$, or $Z$ contribute to the basis-dependent suppression. 
We define the corresponding suppression factor as
\begin{equation}
    \Gamma_{\mathrm{mb}}(P)
    = \prod_{i \in S(P)} \left(1 - e_p(B_i)\right),
\end{equation}
where $B_i \in \{X,Y,Z\}$ for all contributing qubits.  
The measured stabilizer expectation value can then be written as
\begin{equation}
  \langle P \rangle_{\mathrm{ro}}
  \approx
  \Gamma_{\mathrm{mb}}(P)\,\langle P \rangle_{\mathrm{mb}},
\end{equation}
and the underlying stabilizer value is obtained by inverting this suppression.
\begin{equation}
    \langle P \rangle_{\mathrm{mb}}
    = \frac{\langle P \rangle_{\mathrm{ro}}}{\Gamma_{\mathrm{mb}}(P)} .
\end{equation}

Averaging these corrected stabilizer values yields the fidelity
\begin{equation}
    F_{\mathrm{state}}
    = \frac{1}{M}\sum_{P} \langle P \rangle_{\mathrm{mb}},
\end{equation}
which serves as the estimate of the cluster-state fidelity after correcting for single-qubit errors associated with measurement-basis operations.
The corresponding estimation error is
\begin{equation}
    \epsilon
    = \frac{3\,\Gamma\,\max_{P}\!\left(1/\Gamma_{\mathrm{mb}}(P)\right)}{\sqrt{N}},
\end{equation}
where $N$ is the total number of measurement shots and $\Gamma$ is the overhead defined in Eq.~\ref{si:eq:gamma_eq} and evaluated using the intrinsic readout matrix $\Lambda^{-1}$.

\subsubsection{Thermal excitation and preparation fidelity}

Thermal excitation modifies the expectation values of stabilizers. For a stabilizer generator $S_i$ acting on a single qubit with thermal population $q_{0,i}$, the expectation value is reduced as
\begin{equation}
  \langle S_i\rangle \longrightarrow (1-2q_{0,i})\,\langle S_i\rangle.
\end{equation}

A general stabilizer $P$ is a product of several generators. Let $G(P)$ denote the set of generator indices that participate in the decomposition of $P$. The thermal suppression factor for $P$ is then
\begin{equation}
  \prod_{i \in G(P)} (1-2q_{0,i}) .
\end{equation}

To describe this effect at the state level, we write the ideal cluster state as
\begin{equation}
  \rho_{CL,\mathrm{ideal}} = \prod_i \frac{I + S_i}{2} ,
\end{equation}
and introduce the thermally modified reference state
\begin{equation}
  \rho_{CL,\mathrm{noisy}}
  = \prod_i \frac{I + (1 - 2 q_{0,i}) S_i}{2} .
\end{equation}
This state incorporates only the effect of thermal population on the 
stabilizers and no other error mechanisms.

The state fidelity extracted from random stabilizers and readout correction, after removing measurement-basis gate errors as described above, corresponds to the overlap with the ideal cluster state,
\begin{equation}
  F_{\mathrm{state}}
  = \mathrm{Tr}\!\left(\rho_{CL,\mathrm{ideal}}\,\rho_{\mathrm{exp}}\right) .
\end{equation}

Thermal excitation uniformly suppresses stabilizer values, so the preparation fidelity before thermal excitation can be obtained by dividing out the known suppression factor.  
For each stabilizer $P$, the corresponding thermal suppression factor is
\begin{equation}
  \Gamma_{\mathrm{prep}}(P) = \prod_{i \in G(P)} (1-2q_{0,i}) .
\end{equation}
To remove the effect of thermal excitation, we rescale the measured stabilizer value by this factor,
\begin{equation}
  \langle P \rangle_{\mathrm{prep}}
  = \frac{\langle P \rangle_{\mathrm{mb}}}{\Gamma_{\mathrm{prep}}(P)} .
\end{equation}
The preparation fidelity of the ideally initialized cluster state is then obtained by averaging these corrected stabilizer values, 
\begin{equation}
    F_{\mathrm{prep}}
    = \frac{1}{M} \sum_{P} \langle P \rangle_{\mathrm{prep}},
\end{equation}
where the sum runs over all random stabilizers measured in the experiment.  
The corresponding estimation error is
\begin{equation}
    \epsilon
    = \frac{3\,\Gamma\,\max_{P}\!\left(1/\Gamma_{\mathrm{prep}}(P) \cdot 1/\Gamma_{\mathrm{mb}}(P)\right)}{\sqrt{N}},
\end{equation}
where $N$ is the total number of measurement shots and $\Gamma$ is the overhead defined in Eq.~\ref{si:eq:gamma_eq} and evaluated using the intrinsic readout matrix $\Lambda^{-1}$.

\subsubsection{Final fidelity estimation for the 95-qubit chain}

Combining the effects of qubit initialization, measurement-basis dependent single-qubit errors, and thermal suppression of stabilizers, we obtain a more accurate estimate of the fidelity of the 95-qubit one-dimensional cluster state. The random-stabilizer protocol yields a collection of measured stabilizer expectation values. To extract the fidelity, these values are first corrected for readout and then adjusted to remove the suppression introduced by measurement-basis operations.

For each stabilizer $P$, this sequence is written as
\begin{equation}
    \langle P \rangle
    \longrightarrow
    \langle P \rangle_{\mathrm{ro}}
    \longrightarrow
    \langle P \rangle_{\mathrm{mb}},
\end{equation}
where $\langle P \rangle_{\mathrm{ro}}$ is obtained using the intrinsic inverse readout matrix $\Lambda^{-1}$ derived above, and $\langle P \rangle_{\mathrm{mb}}$ removes the basis-dependent factor $\prod_{i \in S(P)} (1 - e_p(B_i))$.  

Averaging the corrected values $\langle P \rangle_{\mathrm{mb}}$ over all measured stabilizers yields the directly observable state fidelity
\begin{equation}
    F_{\mathrm{state}}
    = \frac{1}{M}\sum_{P} \langle P \rangle_{\mathrm{mb}},
\end{equation}
which corresponds to the fidelity of the experimentally prepared cluster state and already includes the effect of thermal excitation.

In addition to this main analysis, we also estimate the preparation fidelity that would be obtained in the absence of thermal excitation. 
This is done by further rescaling each stabilizer value as
\begin{equation}
    \langle P \rangle_{\mathrm{prep}}
    = \frac{\langle P \rangle_{\mathrm{mb}}}{\Gamma_{\mathrm{th}}(P)},
\end{equation}
and averaging the corrected stabilizer values,
\begin{equation}
    F_{\mathrm{prep}}
    = \frac{1}{M}\sum_{P} \langle P \rangle_{\mathrm{prep}},
\end{equation}
which represents the fidelity of the cluster state under perfect qubit initialization but in the presence of all other experimental imperfections.

Applying this procedure to the 95-qubit chain, and using the measured values of $q_0$, $T_1$, $e_{\mathrm{xeb}}$, and the distribution of measurement bases across the random stabilizers, we obtain
\begin{align}
    F_{\mathrm{state}} &= 0.4791 \pm 0.0058,\\
    F_{\mathrm{prep}}  &= 0.5191 \pm 0.0065.
\end{align}
For comparison, applying the empirical readout correction matrix $\Lambda_{\mathrm{exp}}^{-1}$ without removing the effects of qubit initialization and measurement-basis errors yields the value reported in the main text,
\begin{equation}
    F_{\mathrm{state}}^{\mathrm{TP}} = 0.5603 \pm 0.0084 .
\end{equation}

The refined values $F_{\mathrm{state}}$ and $F_{\mathrm{prep}}$ therefore provide a more accurate characterization of the prepared many-body entangled state by explicitly separating the dominant physical error mechanisms. In particular, this refined analysis allows us to quantitatively separate the contributions of thermal excitation and single-qubit errors, including the basis-rotations used in stabilizer-witness measurements, which would otherwise be implicitly absorbed into the conventional readout calibration. 

We emphasize, however, that this treatment is built upon several simplifying assumptions: thermal excitation is modeled as a classical bit-flip process with a fixed population $q_0$; single-qubit gate imperfections are described by small Pauli or depolarizing channels acting independently on each qubit; the cumulative effect of entangling-gate imperfections is incorporated into an effective cluster-state density matrix $\rho_{\mathrm{exp}}$; and higher-order processes, such as leakage, non-Markovian relaxation, and multi-qubit correlated decoherence, are not explicitly included. Under these assumptions, the refined fidelity estimates $F_{\mathrm{state}}$ and $F_{\mathrm{prep}}$ provide a more faithful, but still approximate, characterization of the prepared many-body state, rather than a full microscopic reconstruction of the underlying density matrix. Developing experimentally efficient methods that capture all relevant error mechanisms in a unified and physically complete framework remains an open challenge, requiring continued advances in both experimental techniques and theoretical modeling.

\section{Simulation of Symmetry-Protected Topological Phases}
We discuss the simulation of symmetry-protected topological (SPT) phases by preparing a family of SPT states that are adiabatically connected to the cluster state. We then detect the SPT order based on the \emph{measurement-based wire protocol} as proposed in Ref.~\cite{Azses2020}. We introduce the cluster states as SPT states, the detection protocol, ways to tailor the protocol for compatibility with our experimental devices, and readout error mitigation for the method.

\subsection{Cluster states as symmetry-protected topological states}
In this section, we analyze the symmetry-protected topological (SPT) characteristics of cluster states.
Cluster states are widely recognized as the ground states of a specific parent Hamiltonian defined as
\begin{equation}
    H=-\sum_i h_i=-\sum_j X_j \prod_{k\in \mathcal{N}_j}Z_k,
\end{equation}
where $X$ and $Z$ are Pauli operators acting on the respective qubits, and $\mathcal{N}_j$ denotes the set of neighbors of qubit $j$. 
It is straightforward to verify that this Hamiltonian is the negative sum of the stabilizers of the cluster state.
This parent Hamiltonian exhibits several key properties:
\begin{enumerate}
    \item All terms in $H$ commute with each other.
    \item The Hamiltonian is gapped, with an integer-valued spectrum.
    \item The cluster state is the unique ground state of $H$.
\end{enumerate}
The SPT nature of the 1D cluster state arises from its invariance under a specific symmetry group, $\mathbb{Z}_2 \times \mathbb{Z}_2$.
This symmetry group corresponds to the conservation of the parity of all odd and even qubits, as defined by the operators
\begin{equation}
\begin{split}
    P_{\text {odd }}&=\prod_i h_{2 i+1}=\Pi_i X_{2 i+1},\\
    P_{\text {even }}&=\prod_i h_{2 i}=\Pi_i X_{2 i}.
\end{split}
\end{equation}
These symmetries ensure that the 1D cluster state remains within the same SPT phase under perturbations that respect $\mathbb{Z}_2 \times \mathbb{Z}_2$.

The cluster state's SPT nature is central to its robustness. Specifically, the entanglement structure persists under perturbations or noise that do not break the protecting symmetries. This robustness is particularly relevant in quantum information processing, where imperfections in hardware could otherwise destabilize the state. A defining feature of the non-trivial topology of the cluster state is the presence of symmetry-protected edge states. When a one-dimensional (1D) cluster chain is truncated, the edges of the chain host qubits that exhibit degenerate eigenstates are protected by the system's symmetries~\cite{Azses2020,symmetry2012}. These edge states serve as a clear indicator of the underlying topological order and play a critical role in the cluster state's suitability for measurement-based quantum computation.

\subsection{Measurement-based wire protocol}
It is known that the cluster state as a prototypical SPT state can serve as a computational resource for measurement-based quantum computing (MBQC). 
While the cluster state can only achieve universal MBQC in two dimensions, the cluster state can synthesize arbitrary $\rm{SU}(2)$ operations. This includes the identity gate: An input state $\ket{\psi_{\rm in}}$ is entangled with a $n$-qubit $1$D cluster state chain by a control-Z gate and by measuring the input and first $(n-1)$ qubits of the cluster state in the $X$ basis, we get the last output qubit in the state $\ket{\psi_{\rm out}}$. This process is illustrated in the main text {Fig.~4A}.
A final correction unitary $U_\Sigma$ is applied so that
\begin{equation}
    \ket{\psi_{\rm in}}=U_\Sigma\ket{\psi_{\rm out}}.
\end{equation}
This emphasizes that the identity gate is not a null operation but a quantum teleportation of the input state. Similar to the textbook quantum teleportation protocol, the correction unitary $U_\Sigma$ of the identity gate depends on the measurement results of the first $n$ qubits. 

It is known that the family of states possessing the same SPT order can all serve as resource states for MBQC. As such, as long as the $\mathbb{Z}_2 \times \mathbb{Z}_2$ symmetry is preserved, one can implement basic MBQC operations leveraging the state, such as the identity gate. This provides us with a method for detecting the SPT order by whether the quantum teleportation is faithfully implemented, independent of the input state. As such, the scheme following Ref.~\cite{Azses2020} is to introduce two sets of unitaries acting on the cluster state:
\begin{equation}
\begin{split}
    U_{\rm{S}}(\alpha, \beta)&=e^{i \beta Z_1 X_2 Z_3} e^{i \alpha X_3},\\
    U_{\rm{SB}}(\alpha, \beta)&=e^{i \beta Z_1 X_2 Z_3} e^{i \alpha Y_3}.
\end{split}
\end{equation}
One can readily check that $U_{\rm{S}}(\alpha, \beta)$ and $U_{\rm{SB}}(\alpha, \beta)$ preserves and breaks the $\mathbb{Z}_2 \times \mathbb{Z}_2$ symmetry. We also introduce $k=2,3,4,5,6$ parallel implementations of the unitaries. For $k=2$, the set of symmetry-preserving and breaking unitaries is
\begin{equation}
\begin{split}
    U_{\rm{S}}(\alpha, \beta)&=e^{i \beta Z_1 X_2 Z_3} e^{i \alpha X_3}e^{i \beta Z_5 X_6 Z_7} e^{i \alpha X_7},\\
    U_{\rm{SB}}(\alpha, \beta)&=e^{i \beta Z_1 X_2 Z_3} e^{i \alpha Y_3}e^{i \beta Z_5 X_6 Z_7} e^{i \alpha Y_7}.
\end{split}
\end{equation}
For larger $k$, the setting is similar, such that a one-qubit interval is set between each unitary. The presence of the topological order of the resource states $U_{\rm{S}}\ket{\rm{CL}}$ or $U_{\rm{SB}}\ket{\rm{CL}}$ are then determined by the fidelity of the teleported states. To compute the fidelity, one first applies the correction unitary $U_\Sigma$ and estimates the fidelity with the input state:
\begin{equation}
    \mathcal{F}(\ket{\psi_{\rm{in}}}\bra{\psi_{\rm{in}}}, \rho_{\rm{out}}^\prime)=\langle\psi_{\mathrm{in}}|\rho_{\rm{out}}^\prime|\psi_{\mathrm{in}}\rangle,
\end{equation}
where $\rho_{\rm{out}}^\prime=U_\Sigma\rho_{\rm{out}}U^{\dagger}_\Sigma$. This density matrix can be estimated through single-qubit tomography, i.e.~measuring in the Pauli basis:
$\rho_{\text {out}}^\prime=\left(1+\left\langle X_{\text {out}}\right\rangle \sigma^x+\left\langle Y_{\text {out}}\right\rangle \sigma^y+\left\langle Z_{\text {out}}\right\rangle \sigma^z\right) / 2$. Yet, we find the computational process unamenable to our quantum device as feedforward quantum operations are unavailable. Therefore, a possible solution may be to postselect on the measurement results such that the set of ideal $\rho_{\rm{out}}$ is identical, i.e.~have the same corresponding correction unitary. In the next section, we describe an even better scheme that exempts us from postselection.

\subsection{Postselection-free fidelity estimation}\label{si:sec:post-free}
In principle, feedforward quantum operations in MBQC protocols are necessary to successfully implement specific quantum gates. Yet, when the input state is a stabilizer state, we can circumvent this problem. This is in parallel with the fact that the MBQC protocol can implement Clifford circuits in a single round of measurement without feedforward operations. 

Consider a single qubit stabilizer state $\ket{\psi_{\rm in}}=C_{\rm in}\ket{0}$ with stabilizer $P_{\rm in}$, where $C_{\rm in}$ is a Clifford operation. Our protocol works by the following.
\begin{enumerate}
    \item Entangle the input state with the $n$-qubit cluster state with a control-Z gate.
    \item Measure the first $n$ qubits in the Pauli $X$ basis. For $n\in \mathbb{Z}_{\rm even }$, measure for the last qubit in the $\{C_{\rm in}\ket{0},~C_{\rm in}\ket{1}\}$ basis by acting the $C_{\rm in}$ gate before the standard basis measurement. Otherwise, measure for the last qubit in the $\{HC_{\rm in}\ket{0},~HC_{\rm in}\ket{1}\}$ basis, where $H$ is the Hadamard gate.
    \item Denote the measurement outcome for the first $n$ qubits as $\mathbf{s}$ and the measurement outcome of the last qubit as $\ket{i}$. The one-shot fidelity is computed as
    \begin{equation}
        \mathcal{F}(\rho_{\rm in},\rho_{\rm out})= (-1)^a \bra{i} C_{\rm in}P_{\rm in}C_{\rm in}^\dagger \ket{i}
    \end{equation}
    where $a$ is determined by the commuatation relationship between $P_{\rm in}$ and $U_\Sigma$ such that $P_{\rm in}U_\Sigma=(-1)^a U_\Sigma P_{\rm in}$. The correction unitary is given by $$U_\Sigma=Z^{s_{\rm even}+s_0}X^{s_{\rm odd}},$$
    where $s_{\rm even}=\sum_{i=1,i\in \mathbb{Z}_{\rm even}}^{n}\mathbf{s}_i$ and $s_{\rm odd}=\sum_{i=1,i\in \mathbb{Z}_{\rm odd}}^{n}\mathbf{s}_i$.
    \item Repeat Steps (1)-(3) sufficient times and take the empirical mean of the single-shot fidelities to get the final fidelity estimation, i.e., $\bar{\mathcal{F}}(\rho_{\rm in},\rho_{\rm out})=\frac{1}{M}\sum_{i=1}^M \mathcal{F}_i (\rho_{\rm in},\rho_{\rm out})$, where $M$ is the total number of repetitions.
\end{enumerate}
We now show the correctness of the protocol. Our protocol is built upon two insights: (i) For a stabilizer input state, the fidelity is estimated by estimating the expectation value of the stabilizer regarding the output state, and (ii) The correction unitary $U_\Sigma$ (i.e.~a Pauli) can either be commute or anti-commute with the stabilizer $P_{\rm in}$. In light of the two insights, we observe that the output density matrix satisfies
\begin{equation}
    \tr{\rho_{\rm out}P_{\rm in}}= \tr{U_\Sigma^\dagger \rho_{\rm in} U_\Sigma P_{\rm in}}=(-1)^a\tr{\rho_{\rm in}P_{\rm in}},
\end{equation}
by noting that $\rho_{\rm in}=U_\Sigma\rho_{\rm out}U_\Sigma^\dagger$. That is we can pull the correction unitary through the stabilizer up to a phase depending on the commutation relationship between the two operators. An effective way of measuring the fidelity of the teleported state thus yields without invoking feedforward quantum operations: We estimate the expectation value of the stabilizer $P_{\rm in}$ and multiply the phase $(-1)^a$, which gives us the estimation of $\tr{\rho_{\rm in}P_{\rm in}}$.
Finally, we note that for $n\in\mathbb{Z}_{\rm even}$ or $\mathbb{Z}_{\rm odd}$, we measure the last qubit in the different basis, which we will explain in Sec.~\ref{si:sec:mbqc_basics}.

\begin{figure}[htbp]
\renewcommand{\baselinestretch}{1}\selectfont
    \centering
    \includegraphics[width=\linewidth,height=0.48\textheight,keepaspectratio]{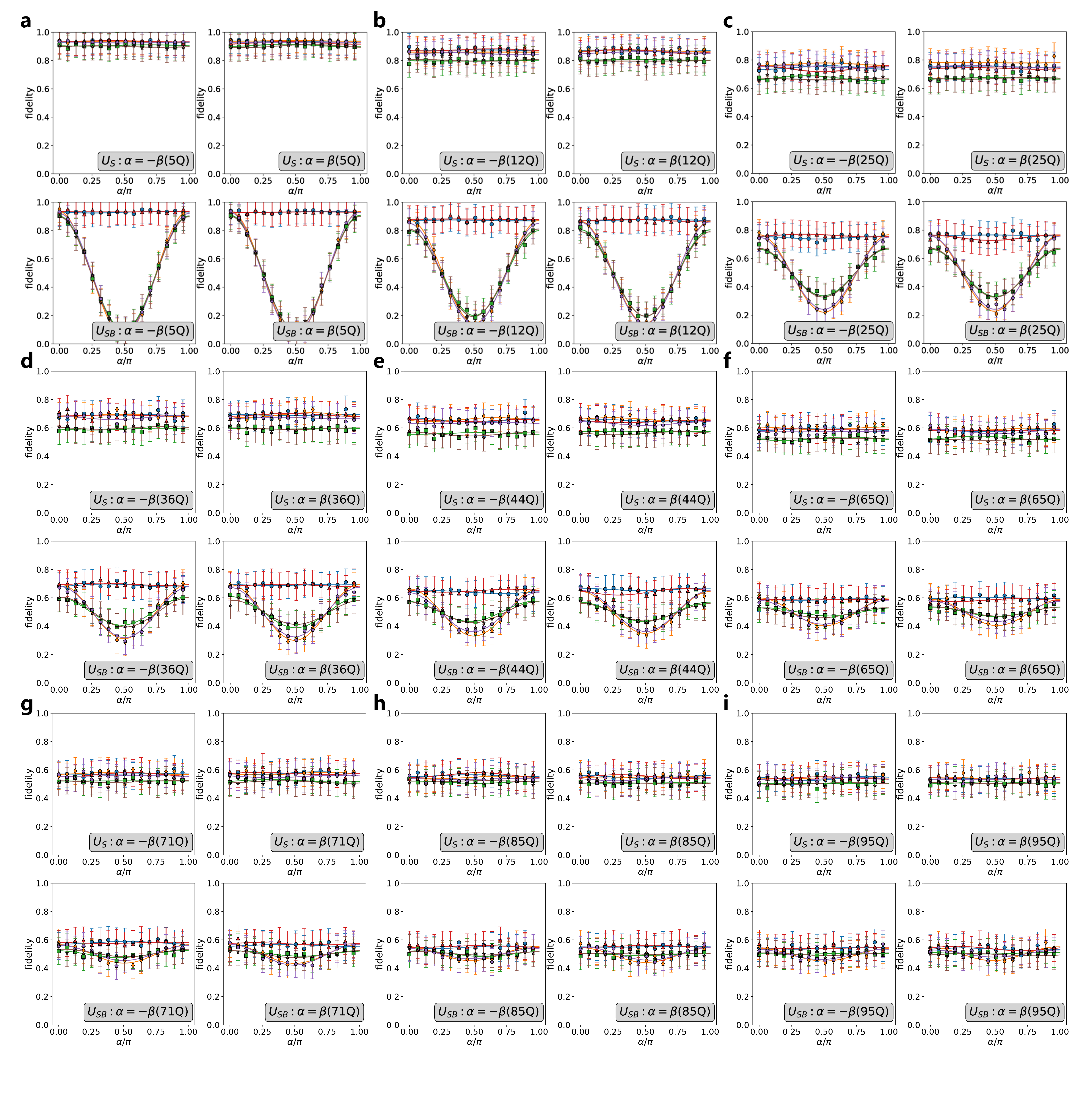}
    \caption{Experimental teleportation fidelity versus the phase parameter  $\alpha/\pi$  for different input states with 1 group of symmetry-preserving or odd-parity symmetry-breaking operation. The curves in the figure are derived from trigonometric fits to the data points, with the twice the amplitude of the fit defined as the fidelity oscillation. (a-i) contains different qubit numbers from 5Q to 95Q. 
    Each data point is the mean of 1000 repeated measurements, with error bars showing the Hoeffding 99.7\% confidence bounds.
    }
    \label{si:fig:spt-1group}
\end{figure}

\begin{figure}[htbp]
\renewcommand{\baselinestretch}{1}\selectfont
    \centering
    \includegraphics[width=\linewidth,height=0.48\textheight,keepaspectratio]{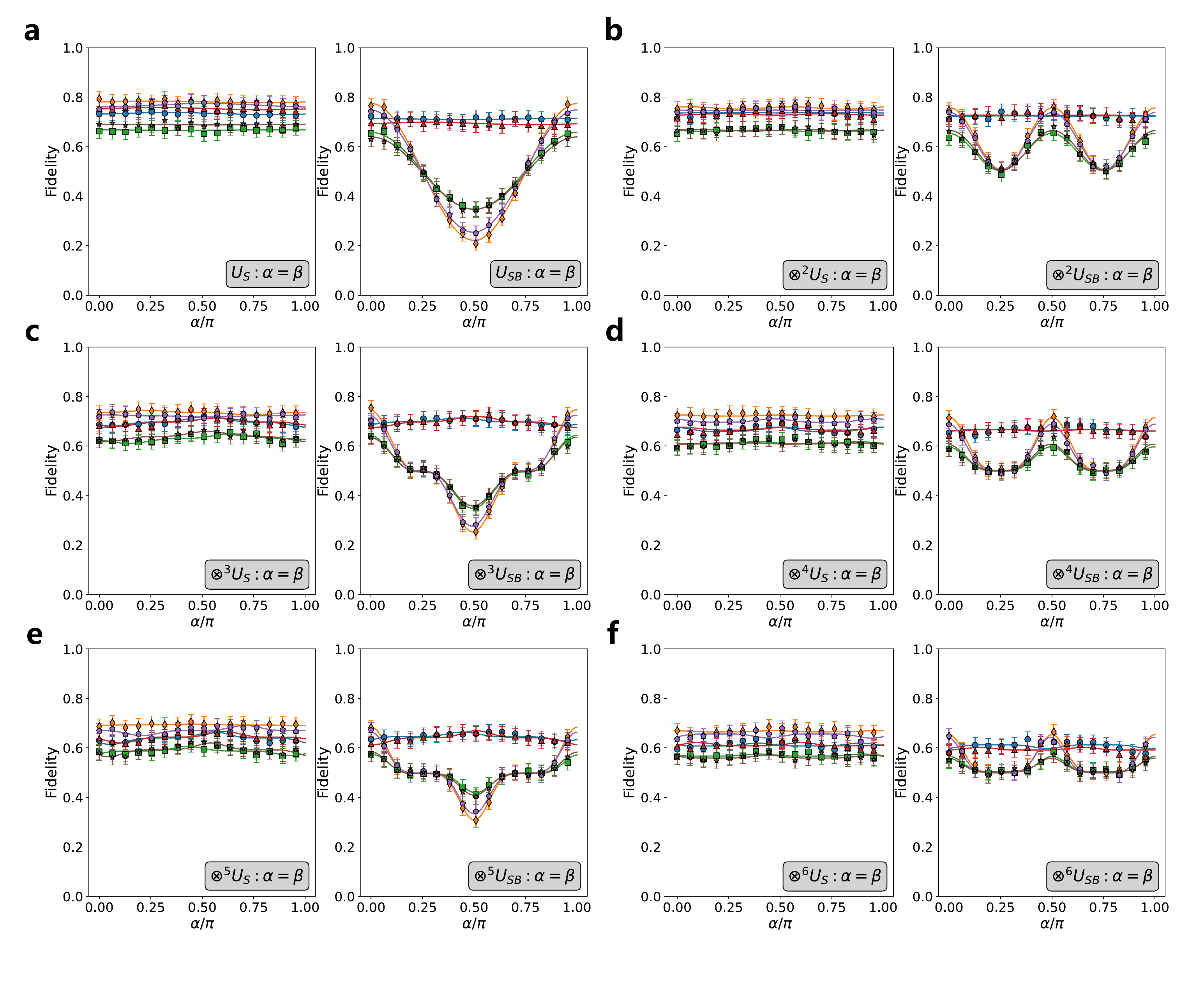}
    \caption{Experimental teleportation fidelity as a function of the phase parameter \( \alpha/\pi \) for different input states under odd-parity symmetry-breaking operations applied across multiple groups of qubits. The simulations were conducted on a teleportation circuit with 25 qubits. The curves represent fits based on trigonometric functions, with the degree of the fit matching the number of parallel groups. For instance, the data for \(n\) groups is fitted using an \(n\)-degree trigonometric function. The fidelity oscillation, defined as twice the amplitude of the fit, quantifies the oscillatory behavior. (a-f) correspond to results for one to six groups of symmetry operations applied to different qubits, respectively. Each data point is the mean of 10{,}000 repeated measurements, with error bars showing the Hoeffding 99.7\% confidence bounds. 
    }
    \label{si:fig:spt-2group}
\end{figure}

As the quantum device is prone to readout errors, we now discuss the combination of measurement error mitigation with our post-selection-free quantum teleportation protocol. The readout errors could cause a severe deviation from the estimated fidelity of the teleported state compared to what it first appears. The rationale is rooted in the fact that the correction unitary depends on the measurement outcome of the first $n$ qubits. As such, the incorrect measurement results resulting from the readout error are counted by an amount of $n*p_{\rm avg}$, where $p_{\rm avg}$ is the average single-qubit error rate if the correlated error rate is small. Consequently, because the correction unitary depends on the measurement outcome's parity, the output state's fidelity may deteriorate linearly as the qubit number grows. This observation is verified in our experimental results in the main text Fig.~4B.

\begin{figure}[htbp]
\renewcommand{\baselinestretch}{1}\selectfont
    \centering
    \includegraphics[width=\linewidth,height=0.48\textheight,keepaspectratio]{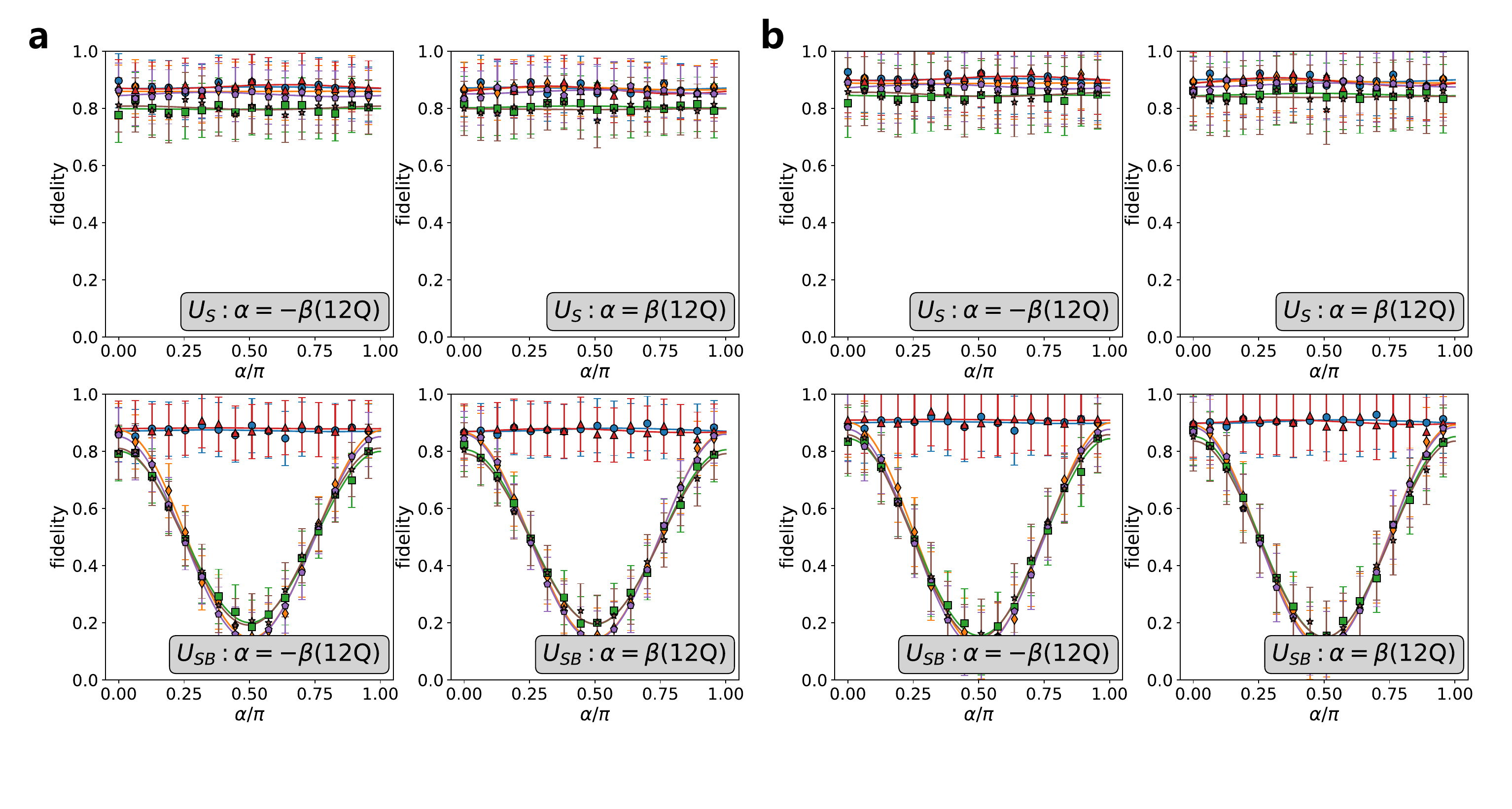}
    \caption{Comparison between SPT results with and without readout error mitigation. The results of 12 qubits are presented here for demonstration. The curves in the figure are derived from trigonometric fits to the data points, with the amplitude of the fit defined as the fidelity fluctuation. (a) Without mitigation. (b) With mitigation. Each data point is the mean of 1000 repeated measurements, with error bars showing the Hoeffding 99.7\% confidence bounds.
    }
    \label{si:fig:spt-mitigation-compare}
\end{figure}

We provide details on measurement error mitigation for the TOD tasks here. 
In the ideal case, the fidelity between the input state $\rho_{\rm in}$ and the output state $\rho_{\rm out}$ can be expressed as:
\begin{equation}
    \mathcal{F}(\rho_{\rm in}, \rho_{\rm out}) = \sum_x p(x)\operatorname{Tr}\left[U(x)\rho_{\rm out} U(x)^\dagger \rho_{\rm in}\right]
\end{equation}
where $p(x)$ is the probability distribution over the measurement outcomes.
As mentioned earlier, noise can flip the measurement outcomes according to some probability distribution. Let us denote this noise by  $p(y | x)$, where  $x = x_1 x_2 \dots x_n$  represents the input measurement results and $y = y_1 y_2 \dots y_n$ represents the noisy output.
For a qubit-independent readout noise model, the noise is assumed to act independently on each qubit. Hence, we can write the probability distribution of the noisy output as $p(y) = p(y_1 | x_1) \cdot \dots \cdot p(y_n | x_n) \cdot p(x)$,
which leads to a noisy fidelity expression
\begin{equation}
    \Tilde{\mathcal{F}} = \sum_x p(x) \sum_y p(y_1 | x_1) \cdots p(y_n | x_n) \operatorname{Tr}\left[\rho_{\rm out} \cdot U(y)^\dagger \rho_{\rm in} U(y)\right].
\end{equation}
This can be approximated by performing measurements and averaging over the $M$ measurement outcomes  $y_m$  as follows
\begin{equation}
    \Tilde{\mathcal{F}} \approx \frac{1}{M}\sum_{m=1}^M \operatorname{Tr}\left[\rho_{\rm out} \cdot U(y_m)^\dagger \rho_{\rm in} U(y_m)\right],
\end{equation}
where $y_m$ is the $m$-th measurement outcome.

To mitigate the noise effect, we apply the inverse of the noise model. Specifically, the corrected fidelity can be written as:
\begin{equation}
    \mathcal{F} = \sum_x p(x) \sum_y p(y_1 | x_1) \cdots p(y_n | x_n) \sum_{x^{\prime}} p^{-1}(y_1 | x_1^{\prime}) \cdots p^{-1}(y_n | x_n^{\prime}) \times \operatorname{Tr}\left[\rho_{\rm out} \cdot U(x^{\prime})^\dagger \rho_{\rm in} U(x^{\prime})\right],
\end{equation}
where  $p^{-1}(y | x)$  represents the inverse of the noise model. 
In practice, we estimate the fidelity from the measurement outcomes by averaging over $M$ independent trials
\begin{equation}
    \hat{\mathcal{F}} = \frac{1}{M} \sum_{m=1}^M \sum_{x^{\prime}} p^{-1}(y_{m1} | x_1^{\prime}) \cdots p^{-1}(y_{mn} | x_n^{\prime}) \operatorname{Tr}\left[\rho_{\rm out} \cdot U(x^{\prime})^\dagger \rho_{\rm in} U(x^{\prime})\right],
\end{equation}
where  $y_m = y_{m1}y_{m2}\dots y_{mn}$  denotes the  $m$-th measurement outcome, and  $p^{-1}(y_i | x_i^{\prime})$  is the inverse noise function for the  $i$-th qubit.
In this expression, the term inside the trace operator lies within the range $[-1, 1]$, and the inverse noise model amplifies it. 
The amplification factor is characterized by the coefficient
\begin{equation}
    \Gamma = \max_{x,y}\left[p^{-1}(y_1 | x_1) \cdots p^{-1}(y_n | x_n)\right] = \prod_{i=1}^n \frac{1+\left|\epsilon_i-\eta_i\right|}{1-\epsilon_i-\eta_i},
\end{equation}
which corresponds to the tensor product overhead $\Gamma_{\rm TP}$ discussed in Eq.~(\ref{si:eq:TP_overhead}).
Thus, by similar reasoning, we can estimate the fidelity within an error bound $\epsilon = \Gamma \sqrt{\frac{2\log(2/\delta)}{M}}$, where $M$ is the number of measurements and the failure probability is lower than $\delta$.

\subsection{Fidelity oscillations in the experimental simulation of symmetry-protected topological phases}

Using the circuit in the main text Fig.~4, we performed experimental simulations of the symmetry-protected phase and demonstrated its robustness against symmetry verification operations by measuring the teleportation fidelity. Specifically, we introduced a single symmetry verification operation into the circuit and scanned the variations of the phase parameter \( \alpha \) in the operation to evaluate its impact on the teleportation fidelity. As shown in Fig.~\ref{si:fig:spt-1group}a, we conducted simulations on a one-dimensional 5-qubit cluster state and fitted the experimental data points using the trigonometric function \( f(x) = A \sin(t/T + \theta) + B \). The fidelity fluctuation was defined as twice the amplitude of the fitted function. Results revealed that, under a single symmetry-preserving operation, the teleportation fidelity exhibited strong robustness to variations in the phase parameter, with fidelity fluctuations below 10\%. In contrast, under a symmetry-breaking operation, the robustness of fidelity to phase parameter variations depended on the input state. For input states \( |\psi_{\text{in}}\rangle = |\pm\rangle \), the symmetry-protected phase retained symmetry and fidelity fluctuations remained below 10\%, whereas for \( |\psi_{\text{in}}\rangle = \pm |0\rangle, \pm |i\rangle \), the symmetry was broken, resulting in fidelity fluctuations exceeding 50\%. 

Using the same protocol, we increased the experimental scale and verified the robustness of symmetry-protected phases in teleportation circuits ranging from 5-qubit to 95-qubit scales. As shown in Fig.~\ref{si:fig:spt-1group}, the experimental results consistently demonstrated robustness, although visibility decreased due to accumulated decoherence and gate errors as the circuit size increased. 

To further validate the robustness of the symmetry-protected phase, we introduced the symmetry verification group with multi parallel symmetry verification operations and scanned the variations of the phase parameter \( \alpha \) for 25-qubit teleportation circuits, following the same approach as the single-operation case, as shown in Fig.~\ref{si:fig:spt-2group}. Similar to the case with a single verification operation, the results further indicated that the fidelity of teleportation exhibits the same state dependence under symmetry-breaking operations. 

By comparing the results of symmetry-preserving and symmetry-breaking operations, we confirmed the impact of these operations on the robustness of the symmetry-protected phase. All experimental data were obtained without error mitigation. Additionally, we compared the fidelity data with and without error mitigation (Fig.~\ref{si:fig:spt-mitigation-compare}). The error-mitigated results showed improved fidelity performance, although the error bars remained noticeable. This indicates that further experimental repetitions are needed to reduce statistical uncertainties.

We highlight some of the exotic behavior of the experimental results and provide explanations here. First, for the case of a single symmetry-breaking unitary introduced, we observe that for $\ket{+}$ and $\ket{-}$ input states, the fidelities are stable and close to the symmetry-preserving cases as shown in Fig.~\ref{si:fig:spt-1group}. This manifests that for specific input states, the quantum teleportation succeeds regardless of whether the resource state is compromised or not. 
To explain this, we use knowledge from measurement-based quantum computation (MBQC) that is formally introduced in {Sec.~\ref{si:sec:mbqc_basics}}. 
These mathematical tools provide a general framework for identifying what properties we need from the resource state for MBQC to function. This will, in turn, help us explain why, in certain cases, even though the $\mathbb{Z}_2\times \mathbb{Z}_2$ breaks, the quantum teleportation still succeeds; equivalently, the state still could serve as a resource state for the specific occasion.

The quantum teleportation corresponds to the identity gate in MBQC. 
Taking the input state $\ket{+}$ and the number of qubits $n\in\mathbb{Z}_{\rm even}$ in the resource cluster state as an instance. The general scheme for using MBQC to realize the universal gate set is provided by [Ref.~\cite{raussendorf2003measurement}, Theorem 1]. Specifically, we will apply this theorem and the deduction shown in {Sec.~I G} (1) of Ref.~\cite{raussendorf2003measurement}. In such cases, the overall state (denoted as $\ket{\phi}$) becomes a $(n+1)$-qubit cluster state after Step 1 of the scheme introduced in {Sec.~\ref{si:sec:post-free}}. Let us note that the $e^{i \beta Z_1 X_2 Z_3}$ in the symmetry-breaking unitary will only add a global phase to the overall state as it is exponential of one of the stabilizers of the cluster state. This leaves the $e^{i \alpha Y_3}$ the term responsible for breaking the $\mathbb{Z}_2\times \mathbb{Z}_2$ symmetry. Following {Sec.~I G} (1) of Ref.~\cite{raussendorf2003measurement}, we get the resulting eigenvalue equation of the overall cluster state as 
\begin{equation}\label{si:eq:plus_ev_eq}
    \prod_{i=0,i\in\mathbb{Z}_{\rm even}}^n X_i\ket{\phi}=\ket{\phi}.
\end{equation}
Now, adding the $e^{i \alpha Y_3}$ term gives
\begin{equation}
    \prod_{i=0,i\in\mathbb{Z}_{\rm even}}^n X_i\ket{\phi^\prime}=\ket{\phi^\prime},
\end{equation}
where $\ket{\phi^\prime}=e^{i \alpha Y_3}\ket{\phi}$. Clearly, the state $\ket{\phi^\prime}$ has the exact same $\prod_{i=0,i\in\mathbb{Z}_{\rm even}}^n X_i$ correlation as a typical cluster state, and the rest of the deduction is similar to that of Ref.~\cite{raussendorf2003measurement}. This explains the results. For other input states, such as $\ket{0}$ or $\ket{1}$, the eigenvalue equation becomes
\begin{equation}\label{si:eq:zero_ev_eq}
    Z_0\prod_{i=1,i\in\mathbb{Z}_{\rm odd}}^{n-1} X_iZ_n\ket{\phi}=\ket{\phi},
\end{equation}
which will be affected by the symmetry-breaking unitary. Yet, when the term in symmetry-breaking unitary becomes $e^{i \alpha Y_2}$, the input state $\ket{0}$ or $\ket{1}$ is unaffected. 
Intriguingly, depending on whether the input state relates to either Eq.~\eqref{si:eq:plus_ev_eq} or \eqref{si:eq:zero_ev_eq}, the symmetry-breaking unitary could behave differently. That is $U_{\rm SB}$ breaks one of the $\mathbb{Z}_2\times \mathbb{Z}_2$ (either odd or even), may not affect the other symmetry so that input state that respects the symmetry is unaffected. The implications are that for noises break the $\mathbb{Z}_2\times \mathbb{Z}_2$ symmetry, the $\ket{0}$, $\ket{1}$, $\ket{+}$ and $\ket{-}$ input states are only affected by either the odd or even symmetry broken. Yet, the $\ket{i}$ and $\ket{-i}$ states could exhibit more vulnerability, which is shown in the experimental results with lower fidelities than the other input states .

\section{Measurement-based quantum computation on a superconducting quantum device}

The measurement-based quantum computation (MBQC) provides an alternative way to implement the quantum circuit on an input quantum state through only single-qubit measurements at the cost that encoding the logical qubits into a larger physical qubit system. For a review of MBQC, see Ref.~\cite{Briegel2009}. Here, apply the two-dimensional cluster state as a resource for MBQC and report the largest implementation to date of the Deutsch-Jozsa algorithm with an input number of up to 20. Our construction leverages several techniques that fully exploit the $2$D pattern of the cluster state and make the computation spatially compact.

\subsection{Deutsch-Jozsa algorithm}
\label{si:sec:dj}
The Deutsch-Jozsa (DJ) algorithm~\cite{Deutsch1992}   is one of the paradigmatic ideas showing quantum speedup in early developments of quantum computing and algorithms. The problem is that consider a function $f(x)\in\{0,1\},~\forall x\in \{0,1\}^n$ promised to be either (i) constant: The output is $0$ or $1$ regardless of the input $x$; (ii) balanced: The number of output $0$ and $1$ is equal for the $2^n$ inputs. Our task is to determine which is the case using as few queries as possible for the function. Classically, the number of queries needed is given by $(2^{n-1}+1)$. By stark contrast, the number of queries for the DJ algorithm is $1$, showing an unbounded speed-up. The quantum circuit layout of the DJ algorithm is illustrated in the main text Fig.~5A top. The main difference between the classical and quantum cases is that in the latter we query a unitary $U_f$ that encodes the function such that $U_f\ket{x}\ket{y}\mapsto \ket{x}\ket{y\oplus f(x)}$, where a single-qubit ancillary register $y$ is introduced. 

In our experiments, we consider four types of instantiation of $f(x)$, which are:
\begin{itemize}
    \item Constant: $f(x)=0$ or $1$ for all $x$;
    \item Balanced: $f(x)=\oplus_{i=1}^n x_i$ or $\overline{\oplus_{i=1}^n x_i}$, where $\overline{x}$ denotes the inverse of $x$.
\end{itemize}
In the constant cases, $U_f\ket{x}\ket{y}\mapsto \ket{x}\ket{y}$ or $U_f\ket{x}\ket{y}\mapsto \ket{x}\ket{\overline{y}}$ so that it is either identity or a Pauli $X$ gate acting on the ancillary qubit.
In the balanced cases, the functions just take the parity (or inverse parity) of $x$. This gives us a simple implementation of $U_f$, which is given by a network of CNOT gates that adds the parity of $x$ to $y$, followed by a Pauli $X$ gate in the inverse parity case. We note that the circuits we obtained are all Clifford in all cases, which is crucial for our implementation using the MBQC protocol.
 
\subsection{Basics of measurement-based quantum computation}
\label{si:sec:mbqc_basics}

\begin{figure}[htbp]
\renewcommand{\baselinestretch}{1}\selectfont
    \centering
    \includegraphics[width=\linewidth,height=0.48\textheight,keepaspectratio]{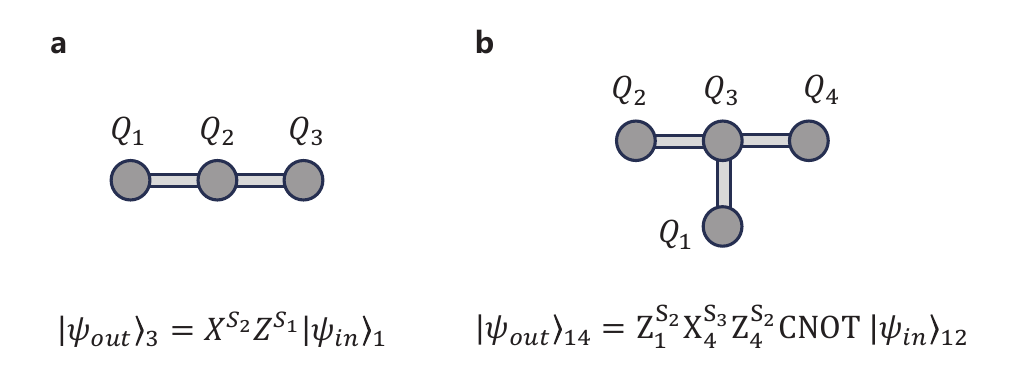}
    \caption{Schematic illustration of MBQC implementation for the identity (I) gate and the controlled-NOT (CNOT) gate. (a) Implementation of the single-qubit I gate using a three-qubit entanglement structure following a cluster-state generation procedure. The input state \(|\psi_{\text{in}}\rangle\) is initialized on qubit 1, while qubits 2 and 3 are prepared in the \(|+\rangle\) state. CZ gates establish a one-dimensional chain, after which X-basis measurements on qubits 1 and 2 yield outcomes \(s_1\) and \(s_2\), determining the final state of qubit 3 as \(|\psi_{\text{out}}\rangle_3 = X^{s_2}Z^{s_1} |\psi_{\text{in}}\rangle_1\). The byproduct operators \(X^{s_2}Z^{s_1}\) can be corrected via classical post-processing. (b) Implementation of the two-qubit CNOT gate using a four-qubit entanglement structure based on a cluster-state preparation protocol. The input state \(|\psi_{\text{in}}\rangle\) is initialized on qubits 1 and 2, while qubits 3 and 4 are in \(|+\rangle\). Three CZ gates are applied according to the illustrated connectivity. X-basis measurements on qubits 2 and 3 yield outcomes \(s_2\) and \(s_3\), resulting in the final state \(|\psi_{\text{out}}\rangle_{14} = Z_1^{s_2} X_4^{s_3} Z_4^{s_2} \cdot \text{CNOT} |\psi_{\text{in}}\rangle_{12}\), with byproduct operators corrected through post-processing. }
    \label{si:fig:MBQC_gate}
\end{figure}

In this section, we review some basic notions and principles of MBQC. 
The MBQC implements the quantum circuits by single-qubit measurements. It is well known that two-dimensional cluster states can serve as resources for universal quantum computing for MBQC such that arbitrary $\rm{SU}(2)$ gates and two-qubit primitive gates, such as the CNOT gate, are realizable. While for the one-dimensional case, we can only implement single-qubit gates through measurement. For a $2$D cluster state supported on a constant-degree graph, i.e.~graph state, we will refer to the graph as a $2$D pattern. The graph is partitioned into several parts: (i) the input part, (ii) the body part, and (iii) the output part. The input part prepares the qubit in the state $\ket{\psi_{\rm in}}$. Then, by single-qubit measurements on a pre-determined basis, the state is propagated through the body part. The resulting state is obtained in the output part such that $\ket{\psi_{\rm out}}=U\ket{\psi_{\rm in}}$, where $U$ is the targeted unitary. The above process describes the basic format of MBQC. That is, there exists an underlying information flow in the measurement process such that the input becomes the output state. To explain this, we introduce the following results about MBQC.
\begin{fact}\label{si:fact:su4}
An arbitrary $SU(2)$ gate can be implemented by a $4$-qubit $1$D cluster state chain with feedforward quantum operations. 
\end{fact}
\begin{proof}
An illustration of the realization can be found in [Ref.~\cite{browne2016one}, Figure 21.3b], where the input qubit is on the leftmost of the $1$D chain. After measuring the first $3$ qubits and performing feedforward quantum operations, we obtain the aimed output state on the rightmost qubit.
To realize an arbitrary $\rm{SU}(2)$ gate, it suffices to decompose the gate into products of three gates through Euler angles. This is given by $U_Z(\alpha)U_X(\beta)U_X(\gamma)$, where $U_\sigma(\theta)$ is the single rotation around the $\sigma$ axis of angle $\theta$. For the derivation, we refer the readers to Ref.~\cite{browne2016one,raussendorf2003measurement}. It should be noted that the feedforward operations are determined by the measurement outcome, which is similar to the correction unitary of quantum teleportation (i.e.~identity gate, as shown in Fig.~\ref{si:fig:MBQC_gate}a).
\end{proof}

\begin{fact}\label{si:fact:cnot}
The CNOT gate can be implemented by a $2$D pattern illustrated in Fig.~\ref{si:fig:MBQC_gate}b.
\end{fact}
\begin{proof}
The 2D measurement pattern for the CNOT gate was first proposed by Ref.~\cite{Raussendorf2001}. The entanglement of qubits is initially set up using control-Z (CZ) gates to form a cluster state, with subsequent measurements in specific bases generating the desired logical operations.

Consider a single qubit input state $\ket{\psi_{\rm in}} = \alpha \ket{0} + \beta \ket{1}$ entangled with a $\ket{+}$ state via a CZ gate. After measuring the first qubit in the $X$-basis, the output state is:
\begin{equation}
    \ket{\psi_{\rm out}} = X^a H \ket{\psi_{\rm in}},
\end{equation}
where $a \in \{0, 1\}$ is the measurement outcome, and the transformation is equivalent to applying a Hadamard gate, up to a Pauli $X^a$ correction. The same effect is discussed in Ref.~\cite{nielsen2006cluster,browne2016one}.

For the CNOT gate, we consider a 2D cluster state with a control qubit $C$, a target qubit $T$, and intermediary qubits initialized in $\ket{+}$ states as shown in Fig.~\ref{si:fig:MBQC_gate}b. The control and target qubits are entangled using CZ gates. Measurement of Q2 and Q3 in the $X$ basis results in a CNOT gate between Q1 and Q4, with Q1 being the control qubit. The reason for this implementation is simple. The measurements of Q2 and Q3 lead to two Hadamard gates (up to Pauli corrections) sandwiching the CZ gate natively prepared by the cluster state. As such, it becomes a CNOT gate.
\end{proof}

It should be noted that the feedforward quantum operation is Pauli. 
This is not an exception; all Clifford gates executed by the MBQC have the feedforward quantum operations as Pauli operators. Besides, from the derivation of Fact~\ref{si:fact:cnot}, we also know that a single-qubit (Pauli $X$ basis) measurement of $1$D cluster state chain is equivalent to applying a Hadamard gate to the input state up to some correction unitary. We thus clarify the choice of measurement basis in Step 2 in the postselection-free fidelity estimation protocol in Sec.~\ref{si:sec:post-free}. The rationale is that for a $1$D chain of $n\in \mathbb{Z}_{\rm odd}$, the $(n-1)$ measurement results in an even number of Hadamard gates acting on the input state so that up to some Pauli correction, we can detect the teleported state directly. Yet, for $n\in \mathbb{Z}_{\rm even}$, the teleported state is $H{\ket{\psi_{\rm in}}}$ so that the measurement basis is transformed by a Hadmard gate to account for the effect.

We provide a result that provides more flexibility for a fixed pattern.
\begin{fact}\label{si:fact:remove}
A Pauli $Z$ basis measurement can delete the corresponding qubit in the graph state up to a correction unitary, while a Pauli $X$ basis measurement between two adjacent graph states can connect them into one.
\end{fact}
\noindent This is a well-known fact for graph state and we refer to Ref.~\cite{raussendorf2003measurement} for the explanation. As such, we can see that a measurement pattern of MBQC has a one-to-one correspondence to a quantum circuit. As illustrated in Fig.~5A middle, the information flows from left to right such that horizontal measurements apply single-qubit gates to the input state. Vertical measurements connect two $1$D chains, applying two-qubit gates between the two logical qubits. Unnecessary qubits are removed by Pauli $Z$ basis measurements.

It is worth pointing out that the MBQC requires multi-round measurements. For each non-Clifford gate (which corresponds to a non-Pauli basis measurement), we need to implement an operation depending on the measurement results. This leads to the following results.
\begin{fact}\label{si:fact:clifford}
Arbitrary stabilizer states can be implemented by MBQC through one round of measurements.
\end{fact}
\begin{proof}
The fact is pointed out in Ref.~\cite{raussendorf2003measurement,jozsa2006introduction}. We also provide our explanation here. As we explained above, the correction unitaries of measurement patterns that correspond to Clifford gates are Paulis. That is the output state becomes
\begin{equation}
    \prod_{i} P_i C_i\ket{\psi_{\rm in}}=\prod_i P^\prime_i\prod_{i} C_i\ket{\psi_{\rm in}}.
\end{equation}
Here, $\ket{\psi_{\rm in}}$ is the input stabilizer state. Each $P_i$ is a Pauli correction followed by the measurement that applies $C_i$. The set of $\{P^\prime_i\}$ is obtained by pulling each Pauli operator $P_i$ through the Clifford operations. Because the Pauli group is the normal subgroup of the Clifford group, each $P^\prime_i$ is still a Pauli operator. An additional consequence is that we may not even execute the Pauli correction after the measurement, but when estimating properties of the outcome state, we accordingly change the explanation of the basis.
\end{proof}
Such an effect is also present in other measurement-based protocols of quantum computation, e.g.~the teleportation-based quantum computation~\cite{gottesman1999demonstrating,jozsa2006introduction}.
To facilitate later discussion, we also provide the following result from Fact~\ref{si:fact:clifford}.
\begin{corollary}\label{si:cor:pauli_clifford}
Any Pauli gates present in a stabilizer state preparation need not be explicitly mapped to the measurement pattern in its MBQC realization.
\end{corollary}
\begin{proof}
The corollary states that for a set of gates $\{g_i\}\subset \mathcal{C}$, we only need to execute $g_i\in \mathcal{C}\setminus \mathcal{P}$ in MBQC, where $\mathcal{C}$ and $\mathcal{P}$ denote the Clifford and Pauli groups,  respectively. Additionally, the Pauli gates can be accounted for through classical post-processing. To see this, we follow the idea in Fact~\ref{si:fact:clifford} that we construct a measurement pattern to implement gates $g_i\in \mathcal{C}\setminus \mathcal{P}$:
\begin{equation}\label{si:eq:pauli_clifford}
    \prod_{i,g\in \mathcal{C}\setminus \mathcal{P}} P_i g\ket{\psi_{\rm in}}=\prod_i P^\prime_i\prod_{i,g\in \mathcal{C}\setminus \mathcal{P}} g_i\ket{\psi_{\rm in}}=\prod_i P^\prime_i \prod_{i,Q\in\mathcal{P}} Q_i\prod_{i,g\in \mathcal{C}} g_i\ket{\psi_{\rm in}},
\end{equation}
where in the last equality, we add the Pauli gates into $\{g_i\}$ and $\prod_i Q_i$ cancels (to identity) with the Pauli gates pulled through the Clifford circuits.
\end{proof}
As such, from Eq.~\eqref{si:eq:pauli_clifford} one can only account for the gates $g_i\in \mathcal{C}\setminus \mathcal{P}$ for MBQC implementation, while the effect of Pauli gates can be added by the correction unitary $\prod_i P^\prime_i \prod_{i,Q\in\mathcal{P}} Q_i$ through classical post-processing.

\subsection{Deutsch-Jozsa algorithm by measurement-based quantum computation}

We now discuss the implementation of the DJ algorithm using MBQC. Our execution of the unitaries considered in the last section allows us to combine three main ingredients to simplify the MBQC pattern. For starters, we note that because of Corollary~\ref{si:cor:pauli_clifford}, the four different kinds of functions illustrated in Sec.~\ref{si:sec:dj} can be reduced to two cases. That is, we only need to consider each case for either the balanced or constant function. The other two functions just differ by Pauli gates, whose effects can be accounted for through alternative interpretation of measurement basis. As such, we only consider the $f(x)=0$ and $\oplus_{i=1}^n x_i$ in the constant and balanced cases. 

Moreover, we invoke Fact~\ref{si:fact:clifford}, which frees us from executing feedforward quantum operations because the circuit implementations of both $U_f$ consist of Clifford gates. It follows that the identity gates in the constant case, can be realized by quantum teleportation, which we have discussed in Sec.~\ref{si:sec:post-free}. For each qubit in the system, at least two extra qubits are introduced for quantum teleportation as we have discussed in the last section. For the balanced function, the quantum circuit is composed of $n$ CNOT gates such that each $\ket{x_i}$ register controls the $\ket{y}$ register. The implementation of CNOT with MBQC is given by Fact~\ref{si:fact:cnot}. Using this measurement pattern, we find a compact way to compose the MBQC pattern for the balanced function. That is, we iteratively include each $x_i$ qubit in the quantum system and apply a CNOT gate to the $y$ register. This results in the quasi-$1$D line pattern as illustrated in Fig.~5A such that the initial $\ket{y}$ state is prepared in the leftmost qubit on the 1D line. Then, as information flows from left to right by measurement, the resulting state on the 1D line becomes $\ket{y\oplus_{i=1}^k x_i}$ after $x_k$ has been measured so that it provides the desirable output at the end of the $1D$ line. The correction unitary, which can be determined by [Ref.~\cite{raussendorf2003measurement}, Theorem 1] or methods given by Ref.~\cite{danos2007measurement}, is also illustrated in Fig.~5A.

Experimentally, the success probability of the algorithm is characterized by the identification probability. That is, for both the balanced and constant cases, we compute the fidelity of the output states against the ideal output. For the constant case, the ideal measurement outcomes are all zero, i.e.~$\ket{00\dots0}$, and some other outcome for the balanced case. For the special balanced function that we studied in this experiment, the ideal output state is the all-one state, i.e.~$\ket{11\dots 1}$.

\subsection{Theoretical analysis of query times and error}
In the ideal Deutsch-Jozsa (DJ) quantum algorithm, the fidelity of the output state and $|000\dots 0\rangle$ can be expressed as
\begin{equation}
\mathbb{E}(v) = \sum_x p(x) \langle x|000\dots 0\rangle,
\end{equation}
where $p(x)$ represents the probability distribution over the computational basis states. Intuitively, this fidelity measures how closely the algorithm's output matches the expected result in the absence of noise or errors.
This theoretical expectation can be empirically estimated through repeated measurements. By averaging over $N$ measurement outcomes, the fidelity approximation is given by
\begin{equation}
\bar{v}=\frac{1}{N}\sum_{m=1}^N v_m=\frac{1}{N}\sum_{m=1}^N \prod_{i=1}^n\left\langle 0 | s_{m, i}\right\rangle,
\end{equation}
where the term $\prod_{i=1}^n\left\langle 0 | s_{m, i}\right\rangle$ evaluates to 1 if the measured state \( \ket{s} \) equals \( \ket{000\dots 0} \), and 0 otherwise.
In an ideal, noiseless scenario, the Deutsch-Jozsa algorithm perfectly distinguishes between constant and balanced functions. A constant function will always produce measurement outcomes of 1, while a balanced function will yield 0. However, in real-world implementations, quantum noise and operational imperfections cause the fidelities to deviate slightly from their ideal values.
These deviations can be modeled statistically as two Bernoulli distributions. For the constant function, the mean fidelity  $p_0$  is close to 1, indicating a high probability of correctly identifying the output state. For the balanced function, the mean fidelity  $p_1$  is close to 0, reflecting the small likelihood of incorrect identification. These parameters, $p_0$ and $p_1$, capture the impact of noise and can provide insight into the algorithm's performance under non-ideal conditions.

\begin{table}[htbp]
\renewcommand{\baselinestretch}{1}\selectfont
    \centering
    \begin{adjustbox}{max width=\linewidth}
\begin{tabular}{|c|c|c|c|c|c|c|}
\hline 
& \multicolumn{2}{c|}{measurement-based} & \multicolumn{2}{c|}{circuit-based} & \multicolumn{2}{c|}{classical} \\
\hline 
DJ input \# & query times N & failure probability & query times N & failure probability & query times N & failure probability \\
\hline 
2 & 3 & $6.42 \times 10^{-4}$ & 3 & $4.36 \times 10^{-4}$ & 3 & $0$ \\
\hline 
3 & 3 & $5.91 \times 10^{-4}$ & 3 & $4.72 \times 10^{-4}$ & 5 & $0$ \\
\hline 
4 & 3 & $6.33 \times 10^{-4}$ & 3 & $4.29 \times 10^{-4}$ & 9 & $0$ \\
\hline 
5 & 3 & $5.42 \times 10^{-4}$ & 3 & $2.58 \times 10^{-4}$ & 11 & $9.77 \times 10^{-4}$ \\
\hline 
6 & 3 & $6.42 \times 10^{-4}$ & 3 & $3.01 \times 10^{-4}$ & 11 & $9.77 \times 10^{-4}$ \\
\hline 
7 & 3 & $7.13 \times 10^{-4}$ & 3 & $4.17 \times 10^{-4}$ & 11 & $9.77 \times 10^{-4}$ \\
\hline 
8 & 3 & $5.18 \times 10^{-4}$ & 3 & $3.83 \times 10^{-4}$ & 11 & $9.77 \times 10^{-4}$ \\
\hline 
9 & 3 & $4.71 \times 10^{-4}$ & 3 & $4.71 \times 10^{-4}$ & 11 & $9.77 \times 10^{-4}$ \\
\hline 
10 & 3 & $7.95 \times 10^{-4}$ & 3 & $7.95 \times 10^{-4}$ & 11 & $9.77 \times 10^{-4}$ \\
\hline 
11 & 4 & $7.55 \times 10^{-4}$ & 4 & $5.38 \times 10^{-5}$ & 11 & $9.77 \times 10^{-4}$ \\
\hline 
12 & 4 & $7.28 \times 10^{-4}$ & 4 & $4.36 \times 10^{-5}$ & 11 & $9.77 \times 10^{-4}$ \\
\hline 
13 & 4 & $7.81 \times 10^{-4}$ & 4 & $6.18 \times 10^{-5}$ & 11 & $9.77 \times 10^{-4}$ \\
\hline 
14 & 4 & $6.90 \times 10^{-4}$ & 4 & $1.20 \times 10^{-4}$ & 11 & $9.77 \times 10^{-4}$ \\
\hline 
15 & 4 & $6.16 \times 10^{-4}$ & 4 & $1.31 \times 10^{-4}$ & 11 & $9.77 \times 10^{-4}$ \\
\hline 
16 & 4 & $5.92 \times 10^{-4}$ & 4 & $1.37 \times 10^{-4}$ & 11 & $9.77 \times 10^{-4}$ \\
\hline 
17 & 4 & $5.69 \times 10^{-4}$ & 4 & $1.29 \times 10^{-4}$ & 11 & $9.77 \times 10^{-4}$ \\
\hline 
18 & 4 & $5.80 \times 10^{-4}$ & 4 & $1.32 \times 10^{-4}$ & 11 & $9.77 \times 10^{-4}$ \\
\hline 
19 & 4 & $5.02 \times 10^{-4}$ & 4 & $1.45 \times 10^{-4}$ & 11 & $9.77 \times 10^{-4}$ \\
\hline 
20 & 4 & $4.91 \times 10^{-4}$ & 4 & $2.19 \times 10^{-4}$ & 11 & $9.77 \times 10^{-4}$ \\
\hline 
\end{tabular}
\end{adjustbox}
    \caption{
        Minimum number of queries required for the measurement-based and circuit-based Deutsch-Jozsa (DJ) algorithms (both using bound $c=1.5$), and the classical determination algorithm, ensuring failure probability $<0.1\%$. 
        For the quantum algorithms, the DJ algorithm is implemented through quantum gates that create superposition and interference, requiring only $O(1)$ queries. 
        For the classical algorithm, the strategy involves sequentially querying the black box: if any two queries differ, the function is balanced; otherwise, constant. 
        The classical failure probability $(1/2)^{N-1}$ occurs when a balanced function yields identical outputs in all $N$ queries. 
        For small input sizes, exhaustive search achieves zero failure probability.
        }
    \label{si:tab:query_times}
\end{table}

\begin{figure}[htbp]
\renewcommand{\baselinestretch}{1}\selectfont
    \centering
    \includegraphics[width=\linewidth,height=0.48\textheight,keepaspectratio]{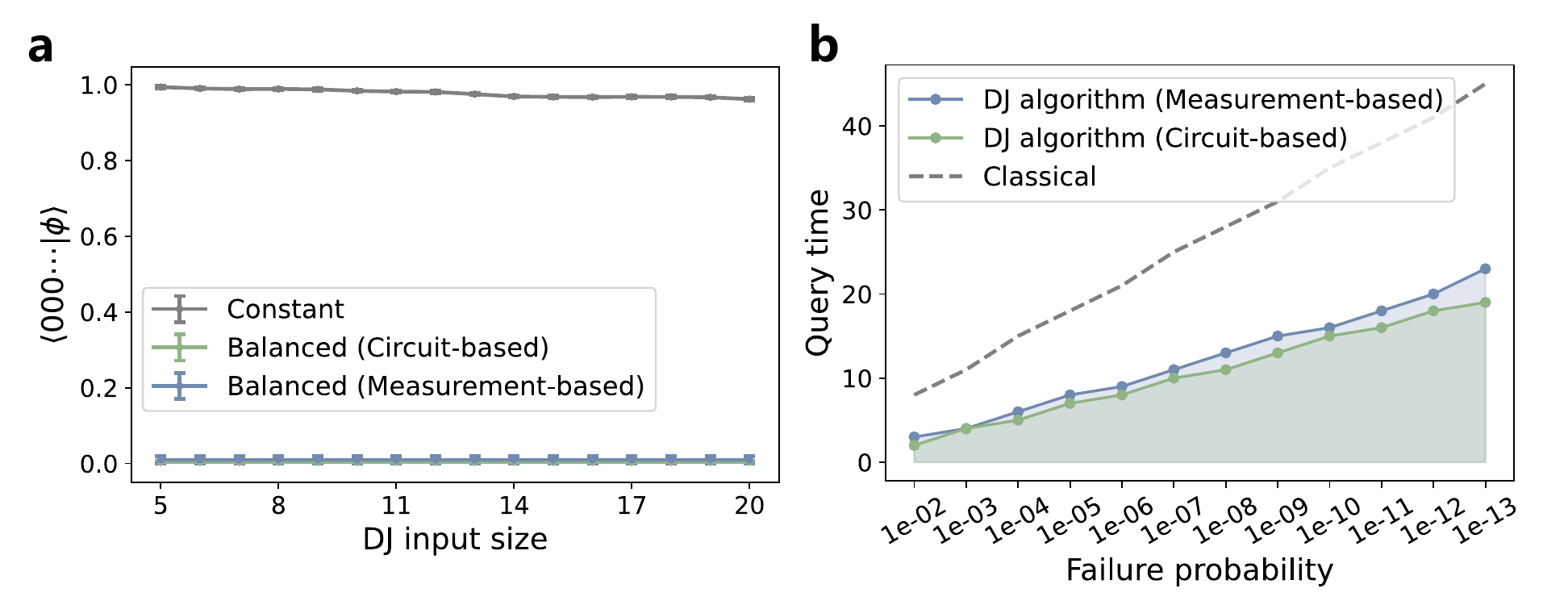}
    \caption{(a) The output probability of the all-zero state in the Deutsch-Jozsa algorithm for different function types, based on which the number of queries required to achieve a given confidence level can be determined.(b) Comparison of the minimum number of queries required for the classical method, circuit-based quantum implementation, and measurement-based quantum implementation of the DJ algorithm at different failure probabilities. The function input size is fixed at 20. Data points in (a) are mean values, with error bars showing the Hoeffding 99.7\% confidence bounds; each point is obtained from 1,000,000 samples, except for the MBQC data, which use 100,000 samples. 
    }
    \label{si:fig:DJ20_querytimes}
\end{figure}

\begin{figure}[htbp]
\renewcommand{\baselinestretch}{1}\selectfont
    \centering
    \includegraphics[width=\linewidth,height=0.48\textheight,keepaspectratio]{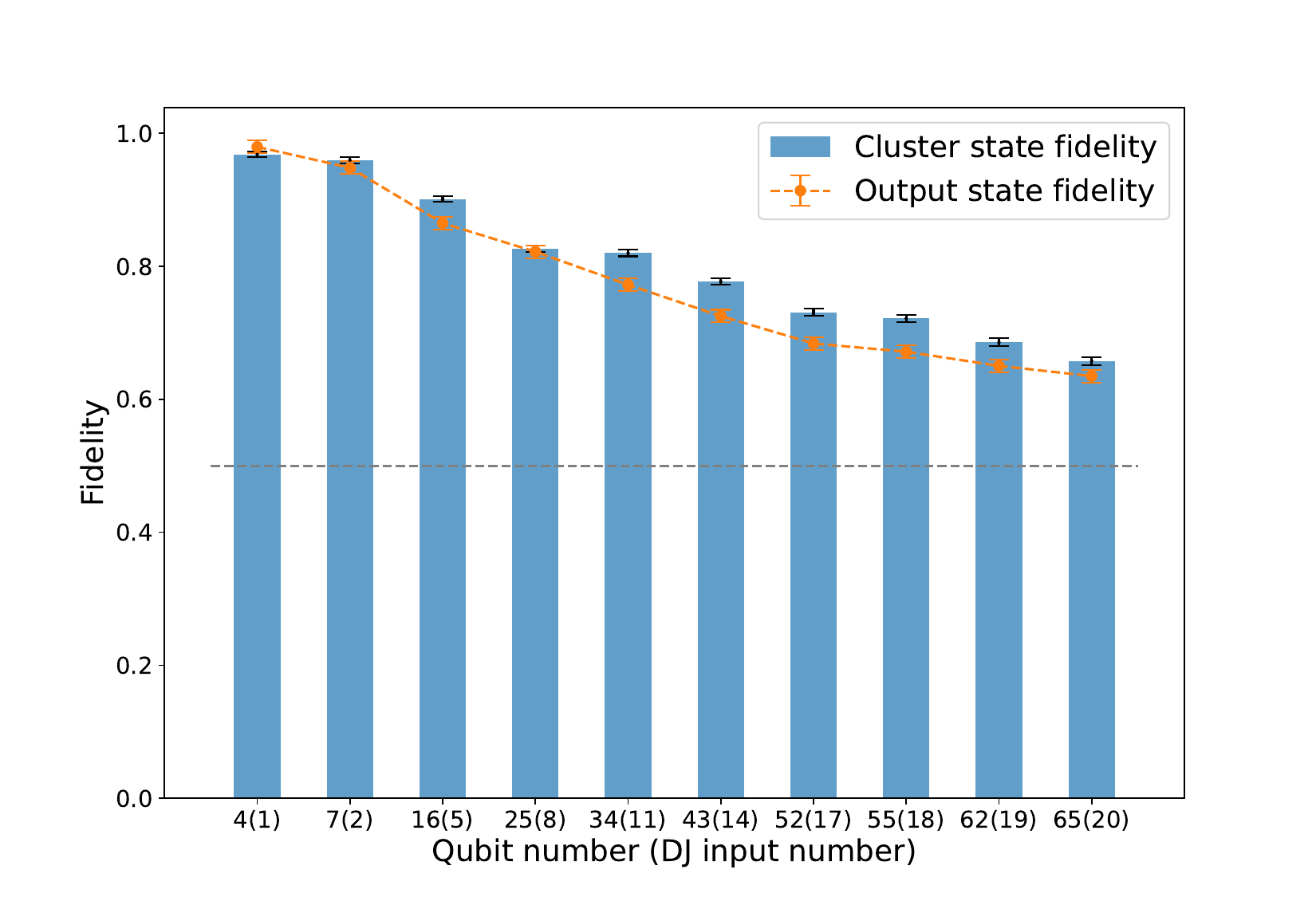}
    \caption{The fidelity of the cluster states used to implement the measurement-based DJ algorithm. The algorithm output state fidelity for different input sizes is displayed as a bar chart, while the corresponding cluster state fidelity is shown as a line graph. For input number \( n = 20 \), the algorithm output state fidelity is \( 0.6376\pm 0.0095 \), and the corresponding cluster state fidelity is \( 0.6577\pm 0.0061 \), with both experimental results showing excellent agreement. 
    Cluster state fidelities follow the sampling configuration in Table~\ref{si:tab:2d_dj_results}, and all error bars represent the Hoeffding 99.7\% confidence bounds. Each DJ-algorithm output state data point is the mean of 100{,}000 measurements. 
    }
    \label{si:fig:djcs_fidelity}
\end{figure}

\begin{figure}[htbp]
\renewcommand{\baselinestretch}{1}\selectfont
    \centering
    \includegraphics[width=\linewidth,height=0.48\textheight,keepaspectratio]{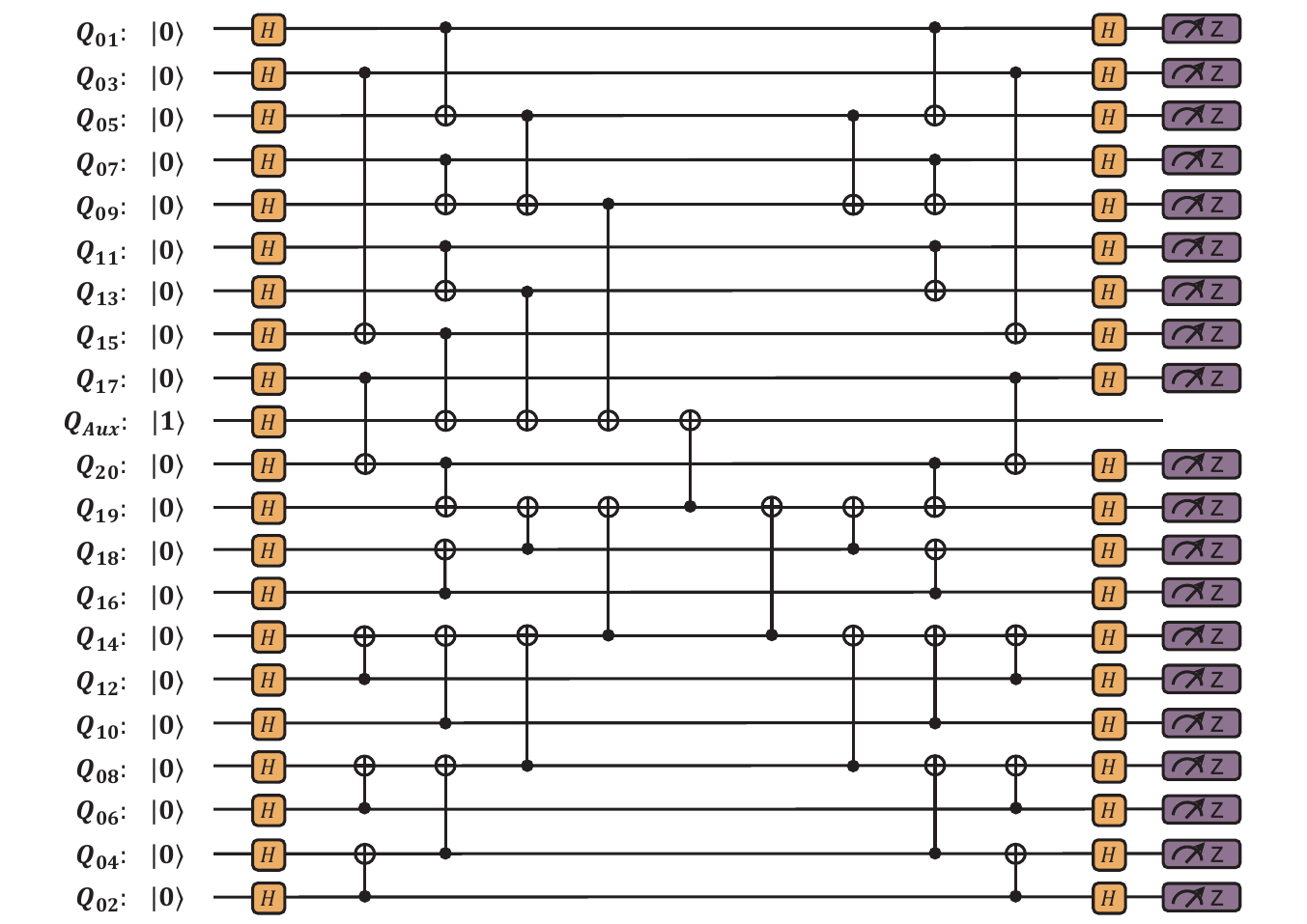}
    \caption{The quantum circuit for the DJ algorithm with input number \( n = 20 \), where \( f(x) \) is balanced.}
    \label{si:fig:DJ_circuit_based_circuit}
\end{figure}

\begin{figure}[htbp]
\renewcommand{\baselinestretch}{1}\selectfont
    \centering
    \includegraphics[width=\linewidth,height=0.48\textheight,keepaspectratio]{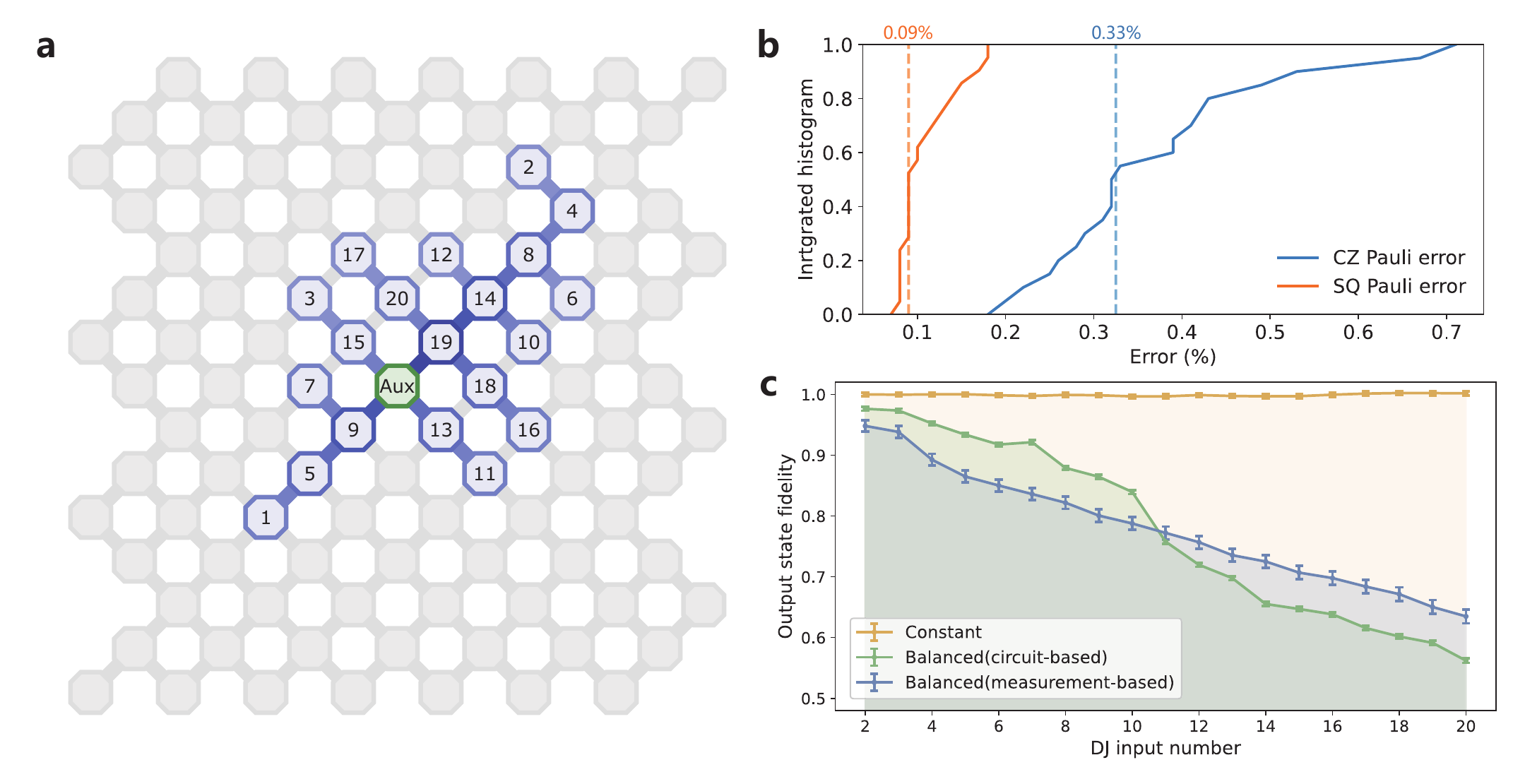}
    \caption{
    Quantum resources for the circuit-based DJ algorithm and the identification probability of the DJ algorithm, under readout error mitigation, in two quantum computation frameworks. 
    (a) The qubits employed in the implementation of the circuit-based DJ algorithm are illustrated, with distinct color shadings representing different layers of CZ gates. 
    (b) The single-qubit gate errors and CZ gate errors associated with the circuit, with median error rates of 0.900\textperthousand{} and 3.250\textperthousand{}, respectively.
    (c) Output state fidelity with measurement error mitigation as a function of the number of DJ inputs for three cases: constant function, and balanced functions implemented using measurement-based and circuit-based approaches.
    In (c), MBQC data points are the mean of 100{,}000 measurements, others the mean of 1{,}000{,}000, with all error bars showing the Hoeffding 99.7\% confidence bounds. 
    }
    \label{si:fig:DJ_MBQC_mitigation}
\end{figure}

Here, we propose a hypothesis testing framework to determine whether the given function is constant or balanced based on the measurement outcomes of the Deutsch-Jozsa (DJ) algorithm.
The null hypothesis, $H_0$, is defined as the function being constant, which corresponds to a fidelity of  $p_0$. Conversely, the alternative hypothesis,  $H_1$, represents the function being balanced, with a fidelity of  $p_1$. Typically, it is assumed that  $p_0$ > $p_1$, reflecting the higher probability of observing correct outcomes for constant functions compared to balanced ones.
We define the test statistic $T$ as the total count of ones in the measurement outcomes \[T=\sum_{m=1}^{N}v_m.\]
The rejection region for  $H_0$  is specified as $T \leq c$, where the threshold $c$ is determined by the desired significance level  $\alpha$  and the test's power  $1-\beta$ . To ensure a $99.9\%$ confidence level in distinguishing between the two hypotheses, we impose a stringent condition such that, in the worst case, the probabilities of Type I and Type II errors are both less than 0.001. Mathematically, this can be expressed as:
\begin{equation}
\begin{aligned}
\text{Type I error: } \operatorname{Pr}(T \leq c \mid H_0) = \alpha, \\
\text{Type II error: } \operatorname{Pr}(T > c \mid H_1) = \beta, \\
\max(\alpha, \beta) \leq 0.001.
\end{aligned}
\end{equation}
Given the number of query iterations  $N$  and the mean fidelities  $p_0$  and  $p_1$, the probabilities of Type I and Type II errors can be computed as
\begin{equation}
    \begin{aligned}
        \alpha &= \sum_{k=0}^{\lfloor c \rfloor} \tbinom{n}{k}p_0^k(1-p_0)^{N-k},\\
        \beta &= \sum_{k=\lfloor c\rfloor +1}^{N} \tbinom{n}{k}p_1^k(1-p_1)^{N-k}.
    \end{aligned}
\end{equation}
The optimal threshold  $c$  is chosen to minimize the maximum error probability  $\delta = \max(\alpha, \beta)$. Subsequently, we iterate over the number of query iterations  $N$  to determine the minimum  $N$  required to ensure  $\delta < 0.001$.
Using the estimated mean fidelity from our scheme, as shown in main text Fig.~5(C), we compute the least number of queries needed for various input sizes and failure probabilities. The results are summarized in Table~\ref{si:tab:query_times} and Fig.~\ref{si:fig:DJ20_querytimes}, providing a benchmark for practical implementation of the DJ algorithm under noise and imperfections. 

Table~\ref{si:tab:query_times} reports the query times and the resulting actual failure probabilities for different input sizes, evaluated at a target failure probability threshold of $10^{-3}$. 
As the input size of the Deutsch-Jozsa (DJ) problem increases, the number of queries $N$ required by the measurement-based quantum computation (MBQC) approach becomes lower than that of the classical method, demonstrating a slight advantage in efficiency.

In contrast, Fig.~\ref{si:fig:DJ20_querytimes}(b) shows the minimum number of queries required at various failure probability thresholds for a fixed function input size of 20.
The results clearly show that, across all tested failure probability thresholds-from $10^{-2}$ down to $10^{-13}$-the quantum algorithm consistently outperforms the classical approach in terms of query efficiency. Moreover, the performance of the circuit-based and measurement-based implementations is quite close, indicating that both approaches are comparably efficient under realistic noise and imperfections. While all methods require more queries as the failure probability decreases, the quantum approaches exhibit a much slower growth in query complexity compared to the classical counterpart, highlighting the practical advantage of the DJ algorithm even at modest input sizes.

\begin{table}[ht]
\renewcommand{\baselinestretch}{1}\selectfont
\centering
\begin{adjustbox}{max width=\linewidth}
\begin{tabular}{|c|c|c|c|c|c|c|c|c|c|c|}
\hline
Qubits & 4 & 7 & 16 & 25 & 34 & 43 & 52 & 55 & 62 &65 \\
\hline
$N$ ($\times10^5$) & 5.46 & 6.12 & 8.38 & 11.78 & 15.74 & 20.86 & 27.56 & 29.06 & 34.40 & 36.74 \\
\hline
Stabilizer circuits & 273 & 306 & 419 & 589 & 787 & 1043 & 1378 & 1453 & 1720 & 1837  \\
\hline
Repeats per stabilizer & 2000 & 2000 & 2000 & 2000 & 2000 & 2000 & 2000 & 2000 & 2000 & 2000  \\
\hline
$\Gamma_{\mathrm{TP}}$ & 1.09 & 1.16 & 1.39 & 1.72 & 2.03 & 2.51 & 2.98 & 3.14 & 3.62 & 3.89 \\
\hline
Fidelity (\%) & 96.83 & 95.96 & 90.11 & 82.59 & 81.98 & 77.71 & 73.08 & 72.17 & 68.58 & 65.77 \\
\hline
Error bar (\%) & 0.44 & 0.44 & 0.46 & 0.48 & 0.49 & 0.52 & 0.54 & 0.55 & 0.59 & 0.61 \\
\hline
\end{tabular}
\end{adjustbox}
\caption{Experimental results for 2D DJ cluster states}
\label{si:tab:2d_dj_results}
\end{table}

\begin{table}[htbp]
\renewcommand{\baselinestretch}{1}\selectfont
    \centering
    \begin{adjustbox}{max width=\linewidth}
\begin{tabular}{|c|c|c|c|c|c|c|}
        \hline
        & \multicolumn{3}{c|}{circuit-based} & \multicolumn{3}{c|}{measurement-based} \\
        \hline
        DJ input \# & circuit depth & circuit width & state fidelity & circuit depth & circuit width & state fidelity \\
        \hline
        2 & 3 & 3 & 0.9763 & 3 & 7 & 0.9483 \\
        \hline
        3 & 3 & 4 & 0.9735 & 3 & 10 & 0.9386 \\
        \hline
        4 & 5 & 5 & 0.9523 & 3 & 13 & 0.8930 \\
        \hline
        5 & 5 & 6 & 0.9338 & 3 & 16 & 0.8654 \\
        \hline
        6 & 5 & 7 & 0.9176 & 3 & 19 & 0.8507 \\
        \hline
        7 & 5 & 8 & 0.9214 & 3 & 22 & 0.8362 \\
        \hline
        8 & 7 & 9 & 0.8792 & 3 & 25 & 0.8216 \\
        \hline
        9 & 7 & 10 & 0.8645 & 3 & 28 & 0.8009 \\
        \hline
        10 & 7 & 11 & 0.8395 & 3 & 31 & 0.7879 \\
        \hline
        11 & 7 & 12 & 0.7578 & 3 & 34 & 0.7723 \\
        \hline
        12 & 7 & 13 & 0.7194 & 3 & 37 & 0.7572 \\
        \hline
        13 & 7 & 14 & 0.6977 & 3 & 40 & 0.7356 \\
        \hline
        14 & 7 & 15 & 0.6553 & 3 & 43 & 0.7262 \\
        \hline
        15 & 7 & 16 & 0.6470 & 3 & 46 & 0.7082 \\
        \hline
        16 & 9 & 17 & 0.6383 & 3 & 49 & 0.6994 \\
        \hline
        17 & 9 & 18 & 0.6156 & 3 & 52 & 0.6855 \\
        \hline
        18 & 9 & 19 & 0.6058 & 3 & 55 & 0.6745 \\
        \hline
        19 & 9 & 20 & 0.5914 & 3 & 62 & 0.6524 \\
        \hline
        20 & 9 & 21 & 0.5621 & 3 & 65 & 0.6376 \\
        \hline
    \end{tabular}
\end{adjustbox}
    \caption{Comparison of quantum circuit resources for the Deutsch-Jozsa (DJ) algorithm oracle implementations of the balanced function. The table presents the required resources for both circuit-based and measurement-based approaches, including circuit depth (the number of CZ gate layers), circuit width (the number of qubits), and state fidelity (the fidelity of the output state with readout error mitigation) for various input sizes in the DJ algorithm. For the balanced function considered, the theoretical output state is the all-1 state (\(|1 ^{\otimes n} \rangle\)).}
    \label{si:tab:resource}
\end{table}

When considering the effects of noise, the measurement outcomes are transformed from  $x = x_1x_2\cdots x_n$  to  $y = y_1y_2\cdots y_n$  according to a conditional probability distribution  $p(y|x)$. For noise that acts independently across qubits, the joint distribution can be expressed as
\begin{equation}
p(y) = p(y_1|x_1) \cdot p(y_2|x_2) \cdots p(y_n|x_n) \cdot p(x).
\end{equation}
To mitigate the impact of noise, an inverse noise model is applied to the measurement results. The corrected fidelity estimate is given by
\begin{equation}
\begin{aligned}
\hat{\bar{v}} &= \frac{1}{N}\sum_{m=1}^N \sum_{x^\prime} p^{-1}(y_{m1}|x_1^\prime) \cdots p^{-1}(y_{mn}|x_n^\prime) \cdot \langle x^\prime|000\dots 0\rangle \\
&= \frac{1}{N}\sum_{m=1}^N \sum_{x^\prime} \prod_{i=1}^n p^{-1}(y_{mi}|x_i^\prime) \delta_{x_i^\prime,0}\\
&= \frac{1}{N}\sum_{m=1}^N \prod_{i=1}^n p^{-1}(y_{mi}|0) 
\end{aligned}
\label{si:eq:DJ_QEM}
\end{equation}
where  $y_m = y_{m1}y_{m2}\cdots y_{mn}$  denotes the  $m$-th measurement outcome, and  $p^{-1}(y_i|x_i^\prime)$  is the inverse of the noise model for the $i$-th qubit.
In practical scenarios, where  $p(y|0) \in [0, 1]$, the inverse noise terms typically satisfy  $p^{-1}(0|0) > 1$  and  $p^{-1}(1|0) \in [-1, 0]$. Consequently, each term in the summation is amplified by a factor
\begin{equation}
\Gamma_{\rm DJ} = \prod_{i=1}^n \frac{1-\epsilon_i}{1-\epsilon_i-\eta_i},
\end{equation}
where $\epsilon_i$ and $\eta_i$ quantify the noise affecting the  $i$-th qubit, as defined in Eq.~(\ref{si:eq:TP_form}). As shown in Eq.~(\ref{si:eq:DJ_QEM}), the original Bernoulli distribution is scattered into an unknown distribution, but its range is bounded by  $(-\Gamma_{\rm DJ}, \Gamma_{\rm DJ}]$  due to the inverse noise terms.
Using Hoeffding's inequality, the amplified error bound is
\begin{equation}
\epsilon = \Gamma_{\rm DJ} \sqrt{\frac{2\log(1/\delta)}{N}},
\end{equation}
and substituting  $\delta = 0.001$, the error bound becomes
\begin{equation}
    \epsilon = \frac{3.717\cdot\Gamma_{\rm DJ}}{\sqrt{N}} .
\end{equation}
This error bound applies equally to both balanced and constant function cases, representing a Type I error. To ensure the error remains manageable, the required query count  $N^\prime$  must satisfy
\begin{equation}
    2\times \frac{3.717\cdot\Gamma_{\rm DJ}}{\sqrt{N^\prime}} \leq p_0^\prime-p_1^\prime,
\end{equation}
where  $p_0^\prime$ and $p_1^\prime$ are the noise-mitigated fidelities for the constant and balanced cases, respectively. Solving for $N^\prime$, the least number of queries required for error-mitigated fidelity is
\begin{equation}
    N \geq \frac{55.26\cdot\Gamma_{\rm DJ}^2}{(p_0^\prime-p_1^\prime)^2}.
\end{equation}

\begin{figure}[htbp]
\renewcommand{\baselinestretch}{1}\selectfont
    \centering
    \includegraphics[width=\linewidth,height=0.48\textheight,keepaspectratio]{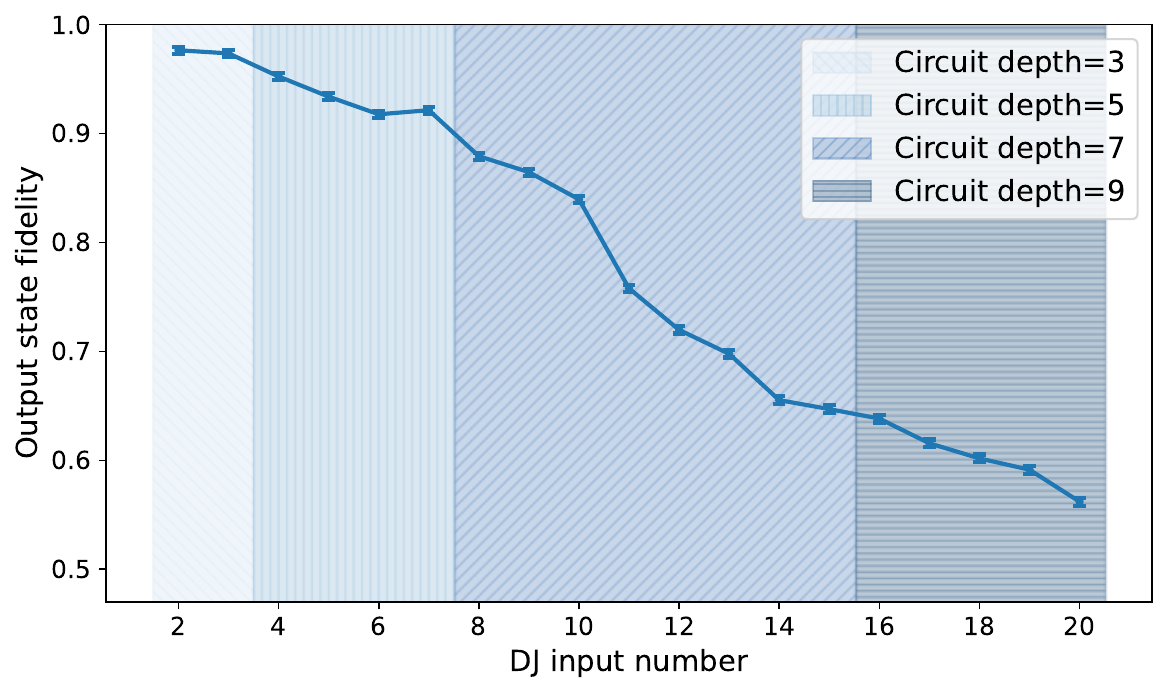}
    \caption{Output state fidelity of the DJ algorithm for the balanced function in the circuit-based framework, with readout error mitigation. Different colored shaded regions represent the corresponding circuit depths (CZ gate layers) for each input size.
    Each data point is the mean of 1{,}000{,}000 measurements, with error bars showing the Hoeffding 99.7\% confidence bounds.
    }
    \label{si:fig:fidelity-depth}
\end{figure}

\subsection{Comparison with circuit-based: detailed performance analysis}

We provide additional analysis and supplementary details on the experimental data presented in the main text. Fidelity results for the two-dimensional cluster state, the quantum resource used in the MBQC implementation of the DJ algorithm for the balanced function, are shown in Fig.~\ref{si:fig:djcs_fidelity}, illustrating the relationship between cluster state fidelity and the DJ algorithm's output state fidelity. The sampling configuration for the DJ cluster state is provided in Table~\ref{si:tab:2d_dj_results}. 
Fig.~\ref{si:fig:DJ_circuit_based_circuit} depicts the quantum circuit diagram for the circuit-based DJ algorithm implementation of the balanced function with \( n = 20 \), while the corresponding qubit resources are shown in Fig.~\ref{si:fig:DJ_MBQC_mitigation}(a). To optimize performance, we fine-tuned the single-qubit and CZ gates in the circuit-based implementation, achieving median errors of \( 0.900 \times 10^{-3} \) and \( 3.250 \times 10^{-3} \), respectively, as shown in Fig.~\ref{si:fig:DJ_MBQC_mitigation}(b). However, the circuit-based approach requires 9 layers of CZ gates, resulting in pronounced decoherence effects that affect the output fidelity. In contrast, the MBQC approach utilizes only 3 layers of CZ gates, providing improved robustness. A detailed comparison of the quantum circuit resources for both circuit-based and measurement-based approaches is presented in Table~\ref{si:tab:resource}. Fig.~\ref{si:fig:fidelity-depth} shows the relationship between the output state fidelity of the circuit-based DJ algorithm for the balanced function and the corresponding circuit depth. 
After applying the error mitigation strategies outlined earlier, we compared the DJ algorithm's output state fidelity under both quantum computation frameworks (Fig.~\ref{si:fig:DJ_MBQC_mitigation}c). The results reveal that as the input number increases, the measurement-based DJ algorithm maintains a more stable and reliable advantage.


\endgroup